\documentclass[11pt]{article}
\usepackage{jheppub}
\usepackage{bm,bbm}
\usepackage{booktabs,amsmath}
\usepackage{accents}
\usepackage{multirow}
\usepackage{graphics}
\usepackage{braket}
\usepackage{mathtools}
\usepackage{comment}
\usepackage{autobreak}
\usepackage{subcaption}
\usepackage{enumerate}
\usepackage{blkarray}
\usepackage{longtable}
\usepackage{url}

\usepackage{tikz}
\usetikzlibrary{patterns}
\usetikzlibrary{arrows,shapes,positioning}
\usetikzlibrary{decorations.markings}
\usetikzlibrary{decorations.pathmorphing}
\tikzset{snake it/.style={decorate, decoration=snake}}

\hypersetup{
    colorlinks=true,
    linkcolor=blue,
    filecolor=magenta,      
    urlcolor=cyan,
    pdfpagemode=FullScreen,
    }

\usepackage{arydshln}
\usepackage{mathtools}

\edef\restoreparindent{\parindent=\the\parindent\relax}
\usepackage{parskip}
\restoreparindent
\usepackage{ascmac}
\usepackage{mathrsfs}

\usepackage{amsthm}
\newtheoremstyle{break}
  {\topsep}{\topsep}%
  {\upshape}{}%
  {\bfseries}{}%
  {\newline}{}%
\theoremstyle{break}
\newtheorem{proposition}{Proposition}[section]
\newtheorem{theorem}[proposition]{Theorem}
\newtheorem{corollary}[proposition]{Corollary}

\def\Tr{{\rm Tr}}

\def\i{{\rm i}}

\def\CH{{\cal H}}

\def\CN{{\cal N}}
\def\CO{{\cal O}}

\def\BF{\mathbb{F}}
\def\BN{\mathbb{N}}

\def\BR{\mathbb{R}}

\def\BZ{\mathbb{Z}}

\title{
Fermionic CFTs from classical codes over finite fields}

\author[a,b]{Kohki Kawabata}
\author[a]{and Shinichiro Yahagi}

\affiliation[a]{Department of Physics, Faculty of Science,
The University of Tokyo,\\
Bunkyo-Ku, Tokyo 113-0033, Japan}

\affiliation[b]{Department of Physics, Osaka University,\\
Machikaneyama-Cho 1-1, Toyonaka 560-0043, Japan}

\preprint{OU-HET-1178}

\abstract{We construct a class of chiral fermionic CFTs from classical codes over finite fields whose order is a prime number. We exploit the relationship between classical codes and Euclidean lattices to provide the Neveu–Schwarz sector of fermionic CFTs.
On the other hand, we construct the Ramond sector using the shadow theory of classical codes and Euclidean lattices.
We give various examples of chiral fermionic CFTs through our construction.
We also explore supersymmetric CFTs in terms of classical codes by requiring the resulting fermionic CFTs to satisfy some necessary conditions for supersymmetry.
}

\begin{document}
\maketitle
\flushbottom

\section{Introduction}
In this paper, we construct a class of chiral fermionic CFTs from classical codes.
The construction of CFTs from error-correcting codes has a long history, which started with chiral bosonic CFTs from a specific class of classical binary codes \cite{frenkel1984natural,frenkel1989vertex,Dolan:1994st}.\footnote{More recently, the construction of non-chiral CFTs has been developed using quantum error-correcting codes \cite{Dymarsky:2020bps,Dymarsky:2020qom,Kawabata:2022jxt} and a class of classical codes \cite{Dymarsky:2021xfc,Yahagi:2022idq,Angelinos:2022umf}. The further progress in this direction can be found in \cite{Dymarsky:2020pzc,Henriksson:2021qkt,Henriksson:2022dnu,Buican:2021uyp,Furuta:2022ykh,Henriksson:2022dml,Dymarsky:2022kwb}.}
Recently, another type of classical code called ternary codes, which are based on the finite field $\BF_3=\{0,1,2\}$, have been exploited to construct chiral fermionic CFTs \cite{Gaiotto:2018ypj}.
This paper generalizes the construction of fermionic CFTs to classical $p$-ary codes.
Classical $p$-ary codes are a natural generalization of binary and ternary codes corresponding to $p=2$ and $p=3$, respectively.
For a prime number $p$, classical $p$-ary codes are based on the finite field $\BF_p=\{0,1,2,\cdots,p-1\}$ and can be formulated in the same manner as the binary and ternary case.
The goal is to extend the construction from ternary codes to classical codes over $\BF_p$ where $p$ is a prime number.

In analogy with the ternary case \cite{Gaiotto:2018ypj}, one of the main tools for building up chiral fermionic CFTs is the relationship between classical codes and Euclidean lattices \cite{leech1971sphere,conway2013sphere}.
We exploit the relationship to construct Euclidean lattices from classical codes.
Each Euclidean lattice can be associated with a set of vertex operators and provides the Neveu-Schwarz (NS) sector of the corresponding fermionic CFT.
On the other hand, constructing the Ramond (R) sector requires us to introduce the shadow of a lattice \cite{conway1990new}, which is a kind of half-shifts of a Euclidean lattice.
The shadow of a lattice can be uniquely determined for each Euclidean lattice constructed from classical codes.
Although the shadow has been defined purely in coding theory, it helps us to yield the R sector of our fermionic CFTs.
In fact, we endow with the R sector by associating each element of the shadow with a vertex operator as in the NS sector (see the left panel of Fig.\ref{fig:overview}).

In particular, the binary case allows us to construct the R sector more directly from classical codes.
While some classical binary codes have been known to provide chiral bosonic CFTs~\cite{frenkel1984natural,frenkel1989vertex,Dolan:1994st}, we employ a different type of binary code to construct chiral fermionic CFTs.
For that class of binary codes, one can introduce the shadow of a code~\cite{conway1990newcode}, which is the analogous notion with the shadow of a lattice. 
It is known that the shadow of a binary code can be lifted to the shadow of the lattice constructed from a code \cite{elkies1999lattices}.
Combining this relationship with the one between the shadow of a lattice and the R sector, we can endow with the R sector from the shadow of a binary code (see the right panel of Fig.\ref{fig:overview}).

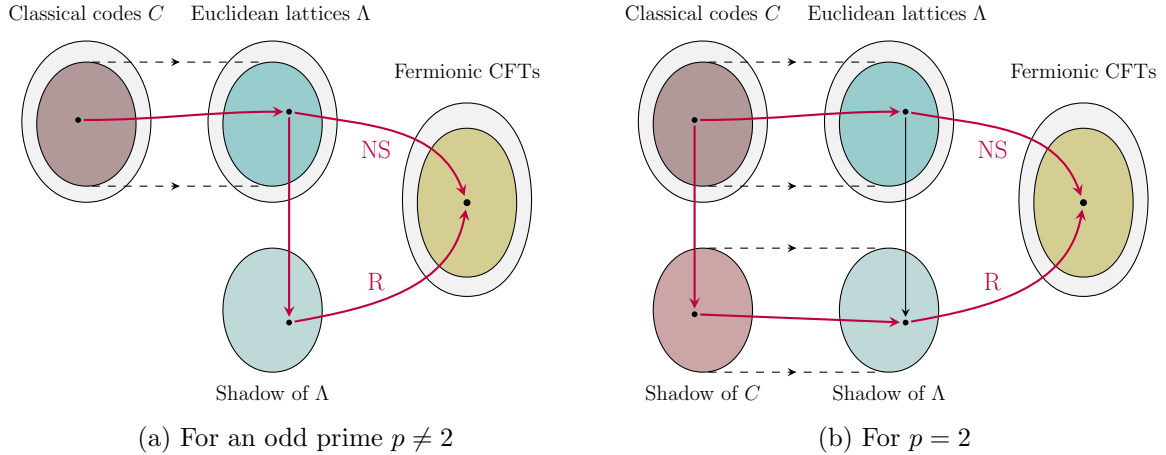
\begin{figure}
    \centering
    \begin{tabular}{cc}
    \begin{minipage}[t]{0.5\linewidth}
    \begin{tikzpicture}[transform shape, scale=0.55, >=stealth]
    \centering
        \begin{scope}[xshift=0cm]
            \begin{scope}[scale=1.3, yshift=-0.3cm]
                \filldraw[fill=lightgray!20] (0,0) to[out=180,in=-90] (-1.2,1.5) to[out=90,in=180] (0,3) to [out=0,in=90] (1.2,1.5) to[out=-90,in=0] (0,0);
            \end{scope}
            
            \filldraw[fill=pink!40!gray] (0,0) to[out=180,in=-90] (-1.2,1.5) to[out=90,in=180] (0,3) to [out=0,in=90] (1.2,1.5) to[out=-90,in=0] (0,0);
            
            \node (lefttop) at (0,3) {};
            \node (leftbot) at (0,0) {};
            
            \fill (-0.2,1.6) node (ul) {} circle[radius=0.07cm];
        \end{scope}
        
        \begin{scope}[xshift=4.5cm]
            \begin{scope}[scale=1.3, yshift=-0.3cm]
                \filldraw[fill=lightgray!20] (0,0) to[out=180,in=-90] (-1.2,1.5) to[out=90,in=180] (0,3) to [out=0,in=90] (1.2,1.5) to[out=-90,in=0] (0,0);
            \end{scope}
            
            \filldraw[fill=teal!40] (0,0) to[out=180,in=-90] (-1.2,1.5) to[out=90,in=180] (0,3) to [out=0,in=90] (1.2,1.5) to[out=-90,in=0] (0,0);
            
            \node (midtop) at (0,3) {};
            \node (midbot) at (0,0) {};
            
            \fill (0.4,1.8) node (um) {} circle[radius=0.07cm];
        \end{scope}
        
    \begin{scope}[xshift=9.2cm,yshift=-2.2cm, scale=1.2]
        \begin{scope}[scale=1.3, yshift=-0.3cm]
            \filldraw[fill=lightgray!20] (0,0) to[out=180,in=-90] (-1.0,1.5) to[out=90,in=180] (0,3) to [out=0,in=90] (1.0,1.5) to[out=-90,in=0] (0,0);
        \end{scope}
        \filldraw[fill=olive!40] (0,0) to[out=180,in=-90] (-1.0,1.5) to[out=90,in=180] (0,3) to [out=0,in=90] (1.0,1.5) to[out=-90,in=0] (0,0);
            
        \node (righttop) at (0,3) {};
        \node (rightbot) at (0,0) {};
        
        \fill (0.,1.5) node (ur) {} circle[radius=0.07cm];
    \end{scope}
    
    \begin{scope}[yshift=-4.5cm]
        
        \begin{scope}[xshift=4.5cm]
            \filldraw[fill=lightgray!50!teal!40] (0,0) to[out=180,in=-90] (-1.2,1.5) to[out=90,in=180] (0,3) to [out=0,in=90] (1.2,1.5) to[out=-90,in=0] (0,0);
            
            \node (smidtop) at (0,3) {};
            \node (smidbot) at (0,0) {};
            
            \fill (0.4,1.2) node (dm) {} circle[radius=0.07cm];
        \end{scope}
    \end{scope}
        
        \draw[->, >=stealth, purple, thick] (ul) to[out=0,in=180] (um);
        \draw[->, >=stealth, purple, thick] (um) to[out=-10,in=110] (ur);
        \draw[->, >=stealth, purple, thick] (dm) to[out=10,in=260] (ur);
        \draw[->, >=stealth, purple, thick] (um) to[out=-90,in=90] (dm);
    
        \begin{scope}[decoration={markings, mark=at position 0.5 with {\arrow{>}}}]
            \draw[dashed, postaction=decorate] (lefttop)--(midtop);
            \draw[dashed, postaction=decorate] (leftbot)--(midbot);
        \end{scope}

        \draw[thick] (0,4.2) node {\Large \textrm{Classical codes $C$}};
        \draw[thick] (4.7,4.2) node {\Large \textrm{Euclidean lattices $\Lambda$}};
        \draw[thick] (9.2,2.8) node {\Large\textrm{Fermionic CFTs}};
        
        \draw (7,0.9) node[very thick,purple] {\LARGE{$\mathrm{NS}$}};
        \draw (7,-2.3) node[very thick,purple] {\LARGE{$\mathrm{R}$}};
        
        \draw[thick] (4.5,-5) node {\Large\textrm{Shadow of $\Lambda$}};
        
        \end{tikzpicture}
        \subcaption{For an odd prime $p\neq2$}
    \end{minipage} &
    \begin{minipage}[t]{0.47\linewidth}
    \begin{tikzpicture}[transform shape, scale=0.55, >=stealth]
    \centering
        \begin{scope}[xshift=0cm]
            \begin{scope}[scale=1.3, yshift=-0.3cm]
                \filldraw[fill=lightgray!20] (0,0) to[out=180,in=-90] (-1.2,1.5) to[out=90,in=180] (0,3) to [out=0,in=90] (1.2,1.5) to[out=-90,in=0] (0,0);
            \end{scope}
            
            \filldraw[fill=pink!40!gray] (0,0) to[out=180,in=-90] (-1.2,1.5) to[out=90,in=180] (0,3) to [out=0,in=90] (1.2,1.5) to[out=-90,in=0] (0,0);
            
            \node (lefttop) at (0,3) {};
            \node (leftbot) at (0,0) {};
            
            \fill (-0.2,1.6) node (ul) {} circle[radius=0.07cm];
        \end{scope}
        
        \begin{scope}[xshift=4.5cm]
            \begin{scope}[scale=1.3, yshift=-0.3cm]
                \filldraw[fill=lightgray!20] (0,0) to[out=180,in=-90] (-1.2,1.5) to[out=90,in=180] (0,3) to [out=0,in=90] (1.2,1.5) to[out=-90,in=0] (0,0);
            \end{scope}
            
            \filldraw[fill=teal!40] (0,0) to[out=180,in=-90] (-1.2,1.5) to[out=90,in=180] (0,3) to [out=0,in=90] (1.2,1.5) to[out=-90,in=0] (0,0);
            
            \node (midtop) at (0,3) {};
            \node (midbot) at (0,0) {};
            
            \fill (0.4,1.8) node (um) {} circle[radius=0.07cm];
        \end{scope}
        
    \begin{scope}[xshift=9.2cm,yshift=-2.2cm, scale=1.2]
        \begin{scope}[scale=1.3, yshift=-0.3cm]
            \filldraw[fill=lightgray!20] (0,0) to[out=180,in=-90] (-1.0,1.5) to[out=90,in=180] (0,3) to [out=0,in=90] (1.0,1.5) to[out=-90,in=0] (0,0);
        \end{scope}
        \filldraw[fill=olive!40] (0,0) to[out=180,in=-90] (-1.0,1.5) to[out=90,in=180] (0,3) to [out=0,in=90] (1.0,1.5) to[out=-90,in=0] (0,0);
            
        \node (righttop) at (0,3) {};
        \node (rightbot) at (0,0) {};
        
        \fill (0.,1.5) node (ur) {} circle[radius=0.07cm];
    \end{scope}
    
    \begin{scope}[yshift=-4.5cm]
    
        \filldraw[fill=pink!60!gray] (0,0) to[out=180,in=-90] (-1.2,1.5) to[out=90,in=180] (0,3) to [out=0,in=90] (1.2,1.5) to[out=-90,in=0] (0,0);
        
        \node (slefttop) at (0,3) {};
        \node (sleftbot) at (0,0) {};
        
        \fill (-0.2,1.4) node (dl) {} circle[radius=0.07cm];
        
        \begin{scope}[xshift=4.5cm]
            \filldraw[fill=lightgray!50!teal!40] (0,0) to[out=180,in=-90] (-1.2,1.5) to[out=90,in=180] (0,3) to [out=0,in=90] (1.2,1.5) to[out=-90,in=0] (0,0);
            
            \node (smidtop) at (0,3) {};
            \node (smidbot) at (0,0) {};
            
            \fill (0.4,1.2) node (dm) {} circle[radius=0.07cm];
        \end{scope}
    \end{scope}
        
        \draw[->, >=stealth, purple, thick] (ul) to[out=0,in=180] (um);
        \draw[->, >=stealth, purple, thick] (ul) to (dl);
        \draw[->, >=stealth, purple, thick] (dl) to (dm);
        \draw[->, >=stealth, purple, thick] (um) to[out=-10,in=110] (ur);
        \draw[->, >=stealth, purple, thick] (dm) to[out=10,in=260] (ur);
        \draw[->, >=stealth, black] (um) to[out=-90,in=90] (dm);
    
        \begin{scope}[decoration={markings, mark=at position 0.5 with {\arrow{>}}}]
            \draw[dashed, postaction=decorate] (lefttop)--(midtop);
            \draw[dashed, postaction=decorate] (leftbot)--(midbot);
        \end{scope}
        
        \begin{scope}[decoration={markings, mark=at position 0.5 with {\arrow{>}}}]
            \draw[dashed, postaction=decorate] (slefttop)--(smidtop);
            \draw[dashed, postaction=decorate] (sleftbot)--(smidbot);
        \end{scope}
        
        \draw[thick] (0,4.2) node {\Large \textrm{Classical codes $C$}};
        \draw[thick] (4.7,4.2) node {\Large \textrm{Euclidean lattices $\Lambda$}};
        \draw[thick] (9.2,2.8) node {\Large\textrm{Fermionic CFTs}};
        
        \draw (7,0.9) node[very thick,purple] {\LARGE $\mathrm{NS}$};
        \draw (7,-2.3) node[very thick,purple] {\LARGE$\mathrm{R}$};
        
        \draw[thick] (0,-5) node {\Large\textrm{Shadow of $C$}};
        \draw[thick] (4.5,-5) node {\Large\textrm{Shadow of $\Lambda$}};
        
        \end{tikzpicture}
        \subcaption{For $p=2$}
    \end{minipage}
    \end{tabular}
    \caption{An illustration of our construction of fermionic CFTs from classical codes for an odd prime $p\neq2$ (the left panel) and $p=2$ (the right panel). We employ the relationship between classical codes and Euclidean lattices to provide the NS sector of fermionic CFTs. Each Euclidean lattice can be associated with a unique shadow, which helps us to provide the R sector. In particular, for the binary case ($p=2$), we can introduce the shadow of a code and make use of it to yield the shadow of a lattice and the R sector.
    }
    \label{fig:overview}
\end{figure}

Our construction of fermionic CFTs establishes the relation between the spectrum of classical codes, Euclidean lattices, and fermionic CFTs.
While the spectrum of a lattice and a CFT are captured by the lattice theta function and torus partition function, its counterpart for a classical code $C$ is the complete weight enumerator $W_C(\{x_a\})$.
Let us construct a fermionic CFT from a classical code $C$ through our construction. Then the torus partition functions depending on the choice of spin structures $(\alpha,\beta)\in\BZ_2\times\BZ_2$ can be written as
\begin{align}
    Z^{\alpha,\beta}(\tau) = \frac{1}{\eta(\tau)^n}\,W_C\left(\left\{f_a^{\alpha,\beta}(\tau)\right\}\right)\,,
\end{align}
where $\tau$ is the modulus of the torus, $\eta(\tau)$ is the Dedekind eta function, and $f^{\alpha,\beta}_a(\tau)$ are functions of $\tau$ depending on spin structures $(\alpha,\beta)$.
Regardless of the choice of spin structures, the partition functions with four sectors can be expressed in terms of the complete weight enumerator.
From the property of the complete weight enumerator, the torus partition functions show the expected modular transformations.
A dictionary between codes, lattices, and CFTs for an odd prime $p\neq 2$ and $p=2$ can be found in Table \ref{table:dictionary_odd_prime} and \ref{table:dictionary_2}, respectively.

To figure out the profiles of our class of fermionic CFTs, we explore fermionic CFTs with supersymmetry.
Instead of constructing supercurrents explicitly, we consider some necessary conditions to be supersymmetric.
Along the lines of \cite{Bae:2020xzl,Bae:2021lvk}, we impose the three supersymmetry conditions:
One requires the NS sector to contain a spin-$\frac{3}{2}$ primary operator.
Another requires the energy of the Ramond sector to be positive.
The other requires the Ramond-Ramond torus partition function to be constant.
These conditions are insufficient to guarantee the existence of supersymmetry, but they strongly suggest it.
Our fermionic CFTs enable us to rewrite the supersymmetry conditions in the language of classical codes and verify them easily.
We apply the supersymmetry conditions to various classical codes and obtain non-trivial CFTs likely to be supersymmetric in our construction.

The paper is organized as follows.
In section \ref{sec:classicalcode}, we review the fundamentals of classical linear codes over $\BF_p$ focusing on the mathematical framework of coding theory. We also define an important class of linear codes called self-dual codes underlying our construction.
In section \ref{sec:fermionic_code_CFTs}, we start with reviewing the construction of odd self-dual lattices from self-dual codes, which can be leveraged to the NS sector.
We also explicitly demonstrate the construction of the R sector using a characteristic vector and the shadow of a lattice. When computing the torus partition functions, we give a canonical choice of characteristic vectors for an odd prime $p$ and $p=2$. In particular, for $p=2$, we can exploit the shadow of a code to provide the shadow of a lattice and the R sector. Then, we illustrate our construction using some simple examples of self-dual codes. In section \ref{sec:SUSY}, we explore fermionic CFTs with supersymmetry. After reviewing supersymmetric CFTs, we rewrite the supersymmetry conditions for our fermionic CFTs and reduce them to some simple constraints on classical codes. Using various examples of self-dual codes, we find fermionic CFTs satisfying all the conditions. Section \ref{sec:discussion} concludes with discussions and future directions. Appendix \ref{sec:list} shows the list of our notations used throughout this paper.
In Appendix \ref{app:nonconstant}, we discuss non-constant Ramond-Ramond partition functions focusing on binary self-dual codes.

\section{Classical linear codes over $\BF_p$}
\label{sec:classicalcode}
In this paper, we focus on classical codes over the finite field, especially $\BF_p=\{0,1,\cdots,p-1\}$ where $p$ is a prime number.
This section introduces mathematical fundamentals of classical codes, such as self-duality, the complete weight enumerator, and the MacWilliams identity following \cite{nebe2006self,justesen2004course} (see \cite{macwilliams1977theory,welsh1988codes,conway2013sphere,elkies2000lattices1,elkies2000lattices} for more details).

Mathematically, a classical linear code can be formulated as a vector space over a finite field.
For a prime number $p$, a $p$-ary linear code $C$ of length $n$ is a subspace of $\BF_p^n$. 
Each element $c\in C$ is called a codeword of $C$.
We denote the dimension of $C$ by $k\leq n$.
The standard inner product of codewords taking values in $\BF_p$ can be introduced by 
\begin{align}
\label{eq:metric}
    g = \mathrm{diag}\,(1,1,\cdots,1)\,.
\end{align}
We denote the inner product between elements $x$, $x'\in\BF_p^n$ by $x\cdot x'$.
Then, for a pair of elements $x = (x_1,\cdots,x_n)\in \BF_p^n$, $x'= (x_1',\cdots,x_n')\in \BF_p^n$, we have $x\cdot x' = \sum_i x_i\,x_i'\in\BF_p$.

One can define a linear code $C$ in terms of a generator matrix $G$ and a parity check matrix $H$.
We assume that they satisfy the relation
\begin{align}
\label{eq:generator_parity}
    G\,H^T = 0 \;\;\, \mathrm{mod}\;\,p\,,
\end{align}
where the generator matrix $G$ and parity check matrix $H$ are a full rank $k\times n$ matrix over $\BF_p$ and full rank $ (n-k)\times n$ matrix over $\BF_p$, respectively.
A generator matrix $G$ generates all codewords $c\in C$:
\begin{align}
    c = x \,G\;\;\,\mathrm{mod}\;\;p\,,\qquad x\in \BF_p^k\,,
\end{align}
where $x$ is a $k$-dimensional row vector. Then, a linear code $C$ can be written as
\begin{align}
    C = \left\{c\in\BF_p^n\;\left|\; c = x\,G,\;\;x\in\BF_p^k\right.\right\}\,.
\end{align}
Note that a generator matrix $G$ is not uniquely determined by a linear code $C$. Each generator matrix corresponds to a different choice of a basis of a linear subspace $C$.

Alternatively, we can use a parity check matrix to characterize a linear code $C$.
From the condition \eqref{eq:generator_parity}, a parity check matrix $H$ vanishes all codewords $c$: $c\,H^T = 0$ mod $p$. 
Since a parity check matrix $H$ is full rank, its kernel space is $k$-dimensional.
Then, a parity check matrix determines a linear code $C$ by
\begin{align}
    C = \left\{y\in\BF_p^n\;\left|\; y\,H^T = 0\;\;\mathrm{mod}\;\;p\right.\right\}\,.
\end{align}
As well as a generator matrix, a parity check matrix $H$ is also not unique to a linear code.

Let us consider a linear code $C^\perp$ generated by a parity check matrix $H$ of $C$:
\begin{align}
    C^\perp = \left\{c'\in\BF_p^n\;\left|\; c' = x\,H,\;\;x\in\BF_p^{n-k}\right.\right\}\,.
\end{align}
From the relation \eqref{eq:generator_parity}, the codewords $c'\in C^\perp$ satisfy $c'\,G^T = 0$ mod $p$. Since the matrix $G$ is full rank, the linear code $C^\perp$ can be written as
\begin{align}
    C^\perp = \left\{y\in\BF_p^n\;\left|\; y\,G^T = 0\;\;\mathrm{mod}\;\;p\right.\right\}\,.
\end{align}
Therefore, while the original code $C$ has the generator matrix $G$ and the parity check matrix $H$, the linear code $C^\perp$ has the generator matrix $H$ and the parity check matrix $G$.
We call a linear code $C^\perp$ the dual code of $C$.
For a linear code $C\subset\BF_p^n$ with the dimension $k$, the dual code $C^\perp$ has length $n$ and the dimension $n-k$. 

We can introduce dual codes without the help of generator matrices and parity check matrices.
Generally, codewords of the dual code $C^\perp$ are orthogonal to ones of $C$: $c\cdot c' = 0$ mod $p$ for $c\in C$ and $c'\in C^\perp$.
For each linear code $C$, we can define its dual code $C^\perp$ by
\begin{align}
    C^\perp = \left\{\left.c'\in\BF_p^n \;\right|\; c\cdot c' = 0 \,\;\mathrm{mod}\;\,p\,,\;\,c\in C\right\}\,.
\end{align}
In this formulation, it is obvious that the dual code $C^\perp$ does not depend on a choice of a generator matrix and parity check matrix.

Let us introduce some important classes of linear codes.
We call a linear code $C$ to be self-orthogonal if $C\subset C^\perp$ and self-dual if $C=C^\perp$.
As explained above, for a linear code $C\subset\BF_p^n$ with the dimension $k$, its dual code has the dimension $n-k$. Then, the dimension of a self-orthogonal code should be $k\leq \frac{n}{2}$, and the dimension of a self-dual code must be $k = n/2$, which also implies that the length of self-dual codes should be even: $n\in2\BZ$.

In the binary case ($p=2$), we can classify self-orthogonal codes further.
All codewords of a binary self-orthogonal code satisfy $c\cdot c\in2\BZ$.
If $C$ satisfies $c\cdot c\in 4\BZ$ for $c\in C$, $C\subset\BF_2^n$ is called doubly-even. 
If a self-orthogonal code $C\subset\BF_2^n$ is not doubly-even, $C$ is called singly-even.

To measure the error-correcting ability of a classical code, it is useful to define the Hamming distance on $\BF_p^n$.
The Hamming distance between elements $x=(x_1,\cdots,x_n)$, $y=(y_1,\cdots,y_n)\in\BF_p^n$ is defined as
\begin{align}
    d_H(x,y) = \left|\left\{i\in\{1,2,\cdots,n\}\mid x_i \neq y_i\right\}\right|\,.
\end{align}
The minimum distance $d(C)$ of a code $C\subset\BF_p^n$ is
\begin{align}
\label{def:minimum_dis}
    d(C) = \min\{d_H(x,y)\mid x,\,y\in C\,,\;\;x\neq y\}\,.
\end{align}
A classical code $C$ with the minimum distance $d(C)=d$ can correct up to $\lfloor(d-1)/2\rfloor$, so the minimum distance can capture the error-correcting property.
It is customary to call a $k$-dimensional linear code $C\subset \BF_p^n$ with the minimum distance $d$ an $[n,k,d]_p$ code. 

If a linear code $C$ contains $x,y\in \BF_p^n$, then it also has $x-y\in C$ due to linearity.
Thus, for linear codes, the minimum distance $d(C)$ reduces to
\begin{align}
    d(C) = \min\{d_H(c,0)\mid c\in C\,,\;\;c\neq 0\}\,.
\end{align}
Then it is useful to define
\begin{align}
\label{eq:def_hamming_wt}
    \mathrm{wt}(x) = d_H(x,0) = \left|\{i\in\{1,2,\cdots,n\}\mid x_i\neq 0\}\right|\,,
\end{align}
which is called the Hamming weight of $x\in\BF_p^n$. Of course, the minimum Hamming weight of a linear code $C$ gives the minimum distance of $C$.

More generally, for $a\in\BF_p$, we define
\begin{align}
\label{eq:composition}
    \mathrm{wt}_a(c) = \left|\{i\mid c_i= a\}\right|\,,
\end{align}
which is called the composition of $c\in C$ in \cite{nebe2006self,macwilliams1977theory}. Then we have $\mathrm{wt}(c) = n - \mathrm{wt}_0(c)$.

Furthermore, we introduce the Lee weight and Euclidean norm \cite{nebe2006self} for later convenience. For $a\in\BF_p$, the Lee weight and Euclidean norm are
\begin{align}
    \mathrm{Lee}(a) &= \min\{a,p-a\}\,,\\
    \mathrm{Norm}(a) &= \left(\mathrm{Lee}(a)\right)^2\,.
\end{align}
For a codeword $c = (c_1,\cdots,c_n)$, the associated Lee weight and Euclidean norm are
\begin{align}
\label{eq:def_Lee_wt}
    \mathrm{Lee}(c) &= \sum_{i=1}^n \mathrm{Lee}(c_i)\,,\\
    \mathrm{Norm}(c) &= \sum_{i=1}^n \mathrm{Norm}(c_i)\,.
    \label{eq:def_norm_c}
\end{align}

Let us illustrate a difference between the Hamming weight, the Lee weight, and the Euclidean norm by using an example $c = (1,2,4,3)\in \BF_5^4$.
Each weight of $c = (1,2,4,3)\in \BF_5^4$ is given by as follows:
\begin{itemize}
    \item The Hamming weight is $\mathrm{wt}(c) = |\{i\in\{1,2,3,4\}\mid c_i\neq0\}| = |\{1,2,3,4\}|=4$.
    \item The Lee weight is $\mathrm{Lee}(c) = \sum_{i=1}^4 \mathrm{Lee}(c_i) = 1 + 2 + (5-4) + (5-3) = 6$.
    \item The Euclidean norm is $\mathrm{Norm}(c) = \sum_{i=1}^4 \mathrm{Lee}(c_i)^2 =  1^2 +2^2 + (5-4)^2 + (5-3)^2 = 10$.
\end{itemize}
From this example, these weights generally have different values for an element $c\in\BF_p^n$.
Note that, for the binary case ($p=2$), the Lee weight and Euclidean norm reduce to the Hamming weight: $\mathrm{Lee}(c) = \mathrm{Norm}(c) = \mathrm{wt}(c)$.

Let us define the complete weight enumerator of a classical code $C$ by
\begin{align}
\label{eq:def_complete_weight}
    W_C(\{x_a\}) = \sum_{c\,\in\, C} \prod_{a\,\in\,\BF_p} x_a^{\mathrm{wt}_a(c)}\,,
\end{align}
where $x_a$ ($a\in\BF_p$) are formal variables and $\mathrm{wt}_a(c)$ are given in \eqref{eq:composition}.
The complete weight enumerator of the dual code $C^\perp$ is uniquely determined by the one of $C$. The relation between them is given by the MacWilliams identity \cite{macwilliams1962combinatorial,macwilliams1963theorem} (see Theorem 10 of Chapter 5 in \cite{macwilliams1977theory})
\begin{align}
\label{eq:MacWilliams_id}
    W_{C^\perp}(\{x_{a}\}) = \frac{1}{|C|} W_C(\{\Tilde{x}_{a}\})\,,\qquad
    \Tilde{x}_a = \sum_{b\,\in\,\BF_p}\,e^{2\pi\i \frac{ ab}{p}}\, x_b\,.
\end{align}
We have an alternative representation
\begin{align}
\label{eq:macwilliams_alt}
    W_C(\{x_a\}) = \frac{|C|}{p^n}\,W_{C^\perp}(\{\tilde{x}_a\})\,,\qquad
    \tilde{x}_a = \sum_{b\,\in\,\BF_p}\,e^{-2\pi\i \frac{ ab}{p}}\,x_b\,.
\end{align}
For a self-dual code $C = C^\perp$, the MacWilliams identity implies $W_C(\{x_a\}) = \frac{1}{|C|} W_C(\{\Tilde{x}_a\})$. Then the complete weight enumerator of a self-dual code is invariant under the change of variables $x_a\to\Tilde{x}_a$ followed by dividing $|C|=p^{\frac{n}{2}}$.

Finally, we give some examples of $p$-ary linear codes.
Let us take a classical binary code $C\subset\BF_2^2$ generated by the generator matrix
\begin{align}
    G = \left[
    \begin{array}{cc}
        1 & 1
    \end{array}
    \right].
\end{align}
We can choose a parity check matrix of the code as $H = G$ since it satisfies $G\,H^T=0$ mod $2$.
The linear code consists of only two codewords: $C=\{(0,0),(1,1)\}$.
The dual code $C^\perp$, which is generated by $H$, is also $C=\{(0,0),(1,1)\}$ because $G = H$.
Then, the linear code is self-dual.
Furthermore, since $c\cdot c \in2\BZ$ for $c = (1,1)$, it is a singly-even self-dual code.

The minimum distance $d(C)$ is obviously $d(C)=2$. Then the linear code is an $[2,1,2]_2$ code. The complete weight enumerator is
\begin{align}
    W_C(\{x_a\}) = x_0^2 + x_1^2\,.
\end{align}
Under the transformation $x_0\mapsto\tilde{x}_0 = x_0 + x_1$ and $x_1\mapsto\tilde{x}_1 = x_0 - x_1$, the complete weight enumerator becomes $W_C(\{\tilde{x}_a\}) = 2(x_0^2 +x_1^2) = |C|\, W_C(\{x_a\})$, which reproduces the MacWilliams identity for a self-dual code $C\subset\BF_2^2$.

Next, let us consider a ternary code $C\subset\BF_3^4$ generated by
\begin{align}
    G = \left[
    \begin{array}{cccc}
        1 & 0 & 1 & 1 \\
        0 & 1 & 1 & 2
    \end{array}
    \right].
\end{align}
We can verify $G\,G^T = 0 $ mod $3$, so a parity check matrix can be chosen as $H = G$, which means that the linear code is self-dual $C=C^\perp$.
The complete weight enumerator is
\begin{align}
    W_C(\{x_a\}) = x_0^4 + x_0x_1^3 + 3x_0x_1^2x_2 + 3 x_0x_1x_2^2 + x_0x_2^3\,.
\end{align}
From the complete weight enumerator, we can easily know that the minimum distance of the code is $d(C) = 3$. Then the linear code is an $[4,2,3]_3$ code.
By the transformation $x_a\mapsto \tilde{x}_a = \sum_{b=0}^2 e^{2\pi\i\frac{ab}{3}}\,x_b$, the complete weight enumerator behaves as $W_C(\{\tilde{x}_a\}) = 3^2\,W_C(\{x_a\})$, which is the MacWilliams identity for a self-dual code $C\subset\BF_3^4$.

\section{Fermionic code CFTs}
\label{sec:fermionic_code_CFTs}
In this section, we give a systematic construction of fermionic CFTs from self-dual codes over $\BF_p$. 
In section \ref{ss:Construction_A}, we construct self-dual lattices from self-dual codes over $\BF_p$ for a prime $p$ via Construction A \cite{leech1971sphere,conway1990new}.
Using the Euclidean lattices and their shadows, we define both the Neveu-Schwarz (NS) sector and the Ramond (R) sector of fermionic CFTs in section \ref{ss:fermionCFT}.
In section \ref{ss:toruspart}, we provide a canonical construction of the shadow of the Construction A lattice. In particular, we construct the shadow of a lattice from the shadow of a code for $p=2$.
We compute the partition functions depending on the choice of spin structures.
Our construction gives an explicit relation between the spectrum of classical codes and fermionic code CFTs.

\subsection{Odd self-dual lattices via Construction A}
\label{ss:Construction_A}

In this section, we review the construction of odd self-dual lattices from classical self-dual codes over $\BF_p$ following \cite{nebe2006self,elkies2000lattices} (refer to Section 3 of Chapter 8 in~\cite{conway2013sphere} for more general construction).
Let $C$ be a classical code over $\BF_p$ of length $n$ where $p$ is a prime number.
The construction of lattices from classical codes has been studied to seek dense sphere packings, and the simplest construction is called Construction A \cite{leech1971sphere}.
It defines the Construction A lattice $\Lambda(C)$ as follows:
\begin{align}
\label{eq:def_constA_lattice}
    \Lambda(C) = \left\{v/\sqrt{p}\mid v\,\in\,\BZ^n,\;\,v=c\;\;\mathrm{mod}\;\, p\,,\;\,c\in C\right\}\,.
\end{align}
The Construction A lattice $\Lambda(C)$ is a Euclidean lattice with respect to the standard Euclidean metric $g$ given in \eqref{eq:metric}. We simply denote the inner product with respect to the metric $g$ by~$\cdot$. While we use the same symbol for the inner product on classical codes, it would be clear from the context which one is being used.

This construction ensures that $\Lambda(C) = \Lambda(C')$ if and only if $C=C'$.
This follows from that the Construction A lattice $\Lambda(C)$ returns to the classical code $C$ by identifying $\lambda\sim\lambda + \sqrt{p}\,\BZ^n$ where $\lambda\in\Lambda(C)$.

By analogy with classical linear codes, we can define the dual lattice.
For each lattice $\Lambda\subset\BR^n$, one can associate its dual $\Lambda^*$:
\begin{align}
\label{eq:dual_lattice}
    \Lambda^* = \left\{\lambda'\in\BR^n\mid \lambda\cdot \lambda'\in\BZ\, ,\,\;\lambda\in\Lambda\right\}\,.
\end{align}
A lattice $\Lambda$ is integral if $\Lambda\subset\Lambda^*$ and self-dual if $\Lambda=\Lambda^*$.
A lattice $\Lambda$ is even if $\lambda\cdot\lambda\in2\BZ$ for $\lambda\in\Lambda$.
Note that every even lattice is integral.
A lattice $\Lambda$ is called odd if it is integral but not even.
This implies that an odd lattice $\Lambda$ contains a vector $\lambda$ whose norm is odd: $\lambda\cdot\lambda\in2\BZ+1$.
Also, a self-dual lattice that contains a vector with an odd norm is called an odd self-dual lattice.

The following proposition guarantees the construction of self-dual lattices from self-dual codes over $\BF_p$ for a prime number $p$ (including $p=2$).

\begin{proposition}
For a prime number $p$, the Construction A lattice $\Lambda(C)$ is self-dual if and only if a classical code $C$ over $\BF_p$ is self-dual.
\label{prop:self-dual}
\end{proposition}

\begin{proof}
In what follows, we prove $\Lambda(C)^* = \Lambda(C^\perp)$.
We begin to prove $\Lambda(C)^*\subset \Lambda(C^\perp)$.
Let $\lambda'$ be an element of $\Lambda(C)^*$. Then it satisfies $\lambda\cdot\lambda'\in \BZ$ for $\lambda\in\Lambda(C)$.
An element of the Construction A lattice $\Lambda(C)$ can be written as
\begin{align}
    \lambda = \frac{c + p \,m}{\sqrt{p}}\,,\qquad \mathrm{for}\quad c\in C\,,\quad m\in \BZ^n\,.
\end{align}
Let us choose $c = 0$, $m = e_i$ where $e_i = (0,\cdots,0,1,0,\cdots,0)$ is an $n$-dimensional unit vector with $1$ at the $i$-th component and $0$s elsewhere.
Then the condition $\lambda\cdot\lambda' \in \BZ$ implies $\sqrt{p}\,e_i\cdot\lambda'\in \BZ$. Therefore, $\lambda'\in \left(\BZ/\sqrt{p}\right)^n$ and the elements $\lambda'\in \Lambda(C)^*$ should be written by
\begin{align}
    \lambda' = \frac{c' + p\,m'}{\sqrt{p}}\,,\qquad \mathrm{for}\quad c'\in \BF_p^n\,,\quad m'\in \BZ^n\,.
\end{align}
Then the condition $\lambda\cdot\lambda'\in \BZ$ reduces to
\begin{align}
    \lambda\cdot\lambda' = \frac{c\cdot c'}{p} + m'\cdot c + c' \cdot m + p\, m\cdot m' \in \BZ\,,
\end{align}
which means $c\cdot c' = 0$ mod $p$ for $c\in C$. In other words, $c' \in C^\perp = \{c'\in \BF_p^n\mid c\cdot c' = 0\;\;\mathrm{mod}\;p\}$. Thus $\Lambda(C)^*\subset \Lambda(C^\perp)$.

Let us move on to the proof of $\Lambda(C)^*\supset\Lambda(C^\perp)$.
Suppose that $\lambda'\in \Lambda(C^\perp)$. Then we can write down
\begin{align}
    \lambda' = \frac{c' + p\,m'}{\sqrt{p}}\,,\qquad \mathrm{for}\quad c'\in C^\perp \,,\quad m'\in \BZ^n\,.
\end{align}
The inner products between $\lambda\in\Lambda(C)$ and $\lambda'\in\Lambda(C^\perp)$ are $\lambda\cdot\lambda'\in\BZ$ since $c\cdot c' = 0$ mod $p$. Then $\Lambda(C)^*\supset \Lambda(C^\perp)$.

Eventually, we get $\Lambda(C)^* = \Lambda(C^\perp)$. If we start with a self-dual code $C = C^\perp$, then the Construction A lattice satisfies $\Lambda(C)^* = \Lambda(C^\perp) = \Lambda(C)$, which implies $\Lambda(C)$ is a self-dual lattice. On the other hand, for a self-dual lattice $\Lambda(C) = \Lambda(C)^*$, the relation $\Lambda(C)^* = \Lambda(C^\perp)$, which we have shown, leads to $\Lambda(C) = \Lambda(C^\perp)$.
Since $\Lambda(C) = \Lambda(C')$ if and only if $C = C'$, we conclude $C = C^\perp$ (self-dual).
Thus $\Lambda(C)$ is self-dual if and only if $C$ is self-dual.
\end{proof}

The following propositions clarify the conditions to construct odd lattices from linear codes over $\BF_p$. Note that there is a subtle difference between the binary case $(p=2)$ and the other cases.

\begin{proposition}
For an odd prime $p$, the Construction A lattice $\Lambda(C)$ is odd if and only if a linear code $C\subset\BF_p^n$ is self-orthogonal.
\label{prop:odd_p}
\end{proposition}
\begin{proof}
From the construction, $\lambda\in\Lambda(C)$ can be written as $\lambda = \frac{c}{\sqrt{p}}+\sqrt{p}\,m$ where $c\in C$ and $m\in\BZ^n$.
The norm of $\lambda$ is given by
\begin{align}
    \lambda\cdot\lambda = \frac{c\cdot c}{p} + 2\,m\cdot c + p\,m\cdot m\,.
\end{align}
Suppose that the Construction A lattice $\Lambda(C)$ is odd. Then the norm $\lambda\cdot\lambda$ should be an integer and there must exist an element $\lambda\in\Lambda(C)$ whose norm is odd.
The first condition reduces to $c\cdot c = 0$ mod $p$ for any $c\in C$, which implies that $C$ is self-orthogonal.
Then the second condition is automatically satisfied for an odd prime $p$ since the Construction A lattice $\Lambda(C)$ contains $\lambda = \sqrt{p}\,e_i$ where $e_i$ is an $n$-dimensional unit vector with 1 at the $i$-th component and with $0$s elsewhere. On the other hand, suppose that $c\cdot c = 0$ mod $p$ for any $c\in C$. Then $\Lambda(C)$ is integral and contains an element whose norm is odd. Thus the Construction A lattice $\Lambda(C)$ is odd.
\end{proof}

\begin{proposition}
For $p=2$, the Construction A lattice $\Lambda(C)$ is odd if and only if a linear code $C\subset\BF_2^n$ is singly-even self-orthogonal.
\label{prop:odd_2}
\end{proposition}
\begin{proof}
We prove this proposition along the same line as the previous one.
An element $\lambda\in\Lambda(C)$ can be written as $\lambda = \frac{c}{\sqrt{2}}+\sqrt{2}\,m$ where $c\in C$ and $m\in\BZ^n$.
The norm of $\lambda$ is given by
\begin{align}
    \lambda\cdot\lambda = \frac{c\cdot c}{2} + 2\,m\cdot c + 2\,m\cdot m\,.
\end{align}
Suppose that the Construction A lattice $\Lambda(C)$ is odd. Then the norm $\lambda\cdot\lambda$ has to be an integer and there must exist an element $\lambda\in\Lambda(C)$ whose norm is odd. The first condition leads to $c\cdot c = 0$ mod $2$ for any $c\in C$. This means that $C$ is a self-orthogonal code over $\BF_2$. The second condition requires that there exists a codeword $c\in C$ such that $c\cdot c\in 4\BZ+2$, which concludes that $C$ is singly-even self-orthogonal.
On the other hand, let $C$ be a singly-even self-orthogonal code. Then, the Construction A lattice $\Lambda(C)$ is integral and contains an element $\lambda$ with an odd norm, which is the definition of an odd lattice.
\end{proof}

Combining Proposition \ref{prop:self-dual} and Proposition \ref{prop:odd_p} or \ref{prop:odd_2}, we arrive at the following theorems.
These theorems ensure that a class of $p$-ary linear codes provides odd self-dual lattices in both the binary and the other cases.

\begin{theorem}
For an odd prime $p$, the Construction A lattice $\Lambda(C)$ is odd self-dual if and only if a linear code $C\subset\BF_p^n$ is self-dual.
\end{theorem}

\begin{theorem}
For $p=2$, the Construction A lattice $\Lambda(C)$ is odd self-dual if and only if a linear code $C\subset\BF_2^n$ is singly-even self-dual.
\end{theorem}

\subsection{Construction of fermionic CFTs}
\label{ss:fermionCFT}

We have described the construction of odd self-dual lattices from linear codes over $\BF_p$ for a prime $p$.
In this section, we associate the odd self-dual lattices with fermionic CFTs and define the fermionic code CFTs.
We also illustrate that while an odd self-dual lattice directly gives the Neveu-Schwarz (NS) sector of fermionic CFTs, the Ramond (R) sector can be provided by the shadow of the lattice.
This is the generalization of the work \cite{Gaiotto:2018ypj}, which focuses on the ternary case ($p=3$).

Let $\Lambda$ be an odd self-dual lattice. A lattice vector $\chi\in\Lambda$ is called characteristic \cite{conway2013sphere,serre2012course,milnor1973symmetric} if, for all $\lambda\in\Lambda$, it satisfies
\begin{equation}
\label{eq:def_characteristic}
    \lambda\cdot\lambda = \chi\cdot\lambda \mod 2.
\end{equation}
Then, the following proposition follows immediately.

\begin{proposition}
For an odd self-dual lattice $\Lambda$, a characteristic vector $\chi$ satisfies $\frac{\chi}{2}\notin \Lambda$.
\label{prop:half_shift}
\end{proposition}

\begin{proof}
Let $\lambda\in\Lambda$ be an element of an odd self-dual lattice $\Lambda$ such that $\lambda\cdot\lambda\in2\BZ+1$.
Then, by definition \eqref{eq:def_characteristic}, we have $\frac{\chi}{2}\cdot \lambda = \frac{\lambda\cdot\lambda}{2} + k\in \BZ+\frac{1}{2}$ where $k\in\BZ$. 
Since $\Lambda$ is assumed to be integral, this is a contradiction unless $\frac{\chi}{2}\notin \Lambda$.
\end{proof}

The choice of a characteristic vector is not unique to $\Lambda$. In what follows, we choose a specific one for the Construction A lattice and denote it by $\chi\in\Lambda(C)$. While we do not explicitly write down the characteristic vector in this section, we will give a canonical choice for the Construction A lattice later.

Using the characteristic vector $\chi\in\Lambda(C)$, we can divide the Construction A lattice $\Lambda(C)$ into two disjoint subsets: $\Lambda(C) = \Lambda_0\cup \Lambda_2$ (following the notation of \cite{conway2013sphere}) where
\begin{align}
\begin{aligned}
\label{eq:lattice_NS}
    \Lambda_0 &=  \left\{ \lambda\in\Lambda(C) \mid \chi\cdot\lambda\equiv0 \mod 2 \right\}\,,
    \\
    \Lambda_2 &= \left\{ \lambda\in\Lambda(C) \mid \chi\cdot\lambda\equiv1 \mod 2 \right\}\,.
\end{aligned}
\end{align}
From the definition of characteristic vectors, the subset $\Lambda_0$ consists of lattice vectors with even norms in $\Lambda(C)$ and $\Lambda_2$ contains only lattice vectors with odd norms in $\Lambda(C)$. 

Let us introduce the following half-shifted subsets:
\begin{align}
\begin{aligned}
\label{eq:def_lambda1_lambda3}
    \Lambda_1 &= 
    \begin{dcases}
    \Lambda_0 + \frac{\chi}{2} & (\chi\cdot\chi\in4\BZ)\,,\\
    \Lambda_2 + \frac{\chi}{2} & (\chi\cdot\chi\in4\BZ+2)\,,
    \end{dcases}\\
    \Lambda_3 &= 
    \begin{dcases}
    \Lambda_2 + \frac{\chi}{2} & (\chi\cdot\chi\in4\BZ)\,,\\
    \Lambda_0 + \frac{\chi}{2} & (\chi\cdot\chi\in4\BZ+2)\,,
    \end{dcases}
\end{aligned}
\end{align}
where $\frac{\chi}{2}\notin \Lambda(C)$ by Proposition \ref{prop:half_shift}.
The union set $S(\Lambda(C)):=\Lambda_1\cup\Lambda_3$ is called the shadow of the lattice $\Lambda(C)$ (originally introduced in \cite{conway1990new}). It can also be written as 
\begin{equation}
    S(\Lambda(C)) = \Lambda(C)+\frac{\chi}{2} = \left\{ \lambda+\frac{\chi}{2} \;\middle|\; \lambda\in\Lambda(C) \right\}.
\end{equation}
The shadow $S(\Lambda(C))$ is not a lattice because it does not contain the origin. Also, it is not closed under addition. In fact, a sum of two shadow vectors $\lambda+\frac{\chi}{2}$, $\lambda'+\frac{\chi}{2}\in S(\Lambda(C))$ is in the original lattice $\Lambda(C)$: $\lambda+\lambda'+\chi\in\Lambda(C)$.
Note that the shadow does not depend on the choice of a characteristic vector $\chi$.

These four subsets $\Lambda_i$ ($i=0,1,2,3$) provide the corresponding Hilbert spaces of fermionic CFTs.
Let $X(z)$ be an $n$-dimensional chiral free boson.
We construct the Neveu-Schwarz (NS) sector of fermionic CFTs by lifting an odd self-dual lattice to a set of vertex operators~\cite{kac1998vertex}. For the Construction A lattice $\Lambda(C)$, the vertex operators are given by
\begin{align}
    V_\lambda(z) = \,: e^{\i\lambda\cdot X(z)}:\,,\qquad \lambda\in\Lambda(C)\,,
\end{align}
where we omit the cocycle factors because it does not matter for our purpose.
The correlation functions of the vertex operators are given by
\begin{align}
\begin{aligned}
    V_\lambda(z)\,V_{\lambda'}(w) &= {(z-w)^{\lambda\cdot\lambda'}}\,V_{\lambda+\lambda'}(w)\left(1+\CO(z-w)\right)\,,\\
    &= {(z-w)^{\lambda\cdot\lambda'}}\,V_{\lambda+\lambda'}(w) + \cdots\,.
\end{aligned}
\end{align}
The amplitude is a meromorphic function and does not have a branch cut since the Construction A lattice $\Lambda(C)$ is integral ($\lambda\cdot\lambda'\in\BZ$).

The vertex operators $V_\lambda(z)$ are mapped to the momentum states $\ket{\lambda}$ via the state-operator isomorphism. 
These states provide the Hilbert space of the NS sector:
\begin{align}
    \CH_{\mathrm{NS}}(C) = \left\{\left.\alpha_{-k_1}^{i_1}\cdots\alpha_{-k_r}^{i_r}\ket{\lambda}\;\right| \; \lambda\in\Lambda(C)\,,\;r\in\BZ_{\geq0}\right\}\,,
\end{align}
where $\alpha_k^i$ ($i=1,2,\cdots,n$) are the bosonic oscillators that satisfy the algebra
\begin{align}
    \left[\alpha_k^i,\,\alpha_l^j\right] = k\,\delta_{k+l,0}\,\delta^{i,j}\,.
\end{align}
The conformal weight of $\alpha_{-k_1}^{i_1}\cdots\alpha_{-k_r}^{i_r}\ket{\lambda}$ is given by $h=\frac{\lambda^2}{2}+\sum_{j=1}^r k_j$ and $\bar{h}=0$. Then its spin $s=h-\bar{h}$ is an integer ($s\in\BZ$) for $\lambda\in\Lambda_0$ and a half-integer ($s\in\BZ+\frac{1}{2}$) for $\lambda\in\Lambda_2$.
Then the resulting Hilbert space $\CH_\mathrm{NS}(C)$ contains not only bosonic states with $\lambda\in\Lambda_0$ but also fermionic ones with $\lambda\in\Lambda_2$:
\begin{align}
    \begin{aligned}
    \label{eq:NS_grading}
        \CH_{\mathrm{NS}}^+(C) &= \left\{\left.\alpha_{-k_1}^{i_1}\cdots\alpha_{-k_r}^{i_r}\ket{\lambda}\;\right| \; \lambda\in\Lambda_0\,,\;r\in\BZ_{\geq0}\right\}\,,\\
        \CH_{\mathrm{NS}}^-(C) &= \left\{\left.\alpha_{-k_1}^{i_1}\cdots\alpha_{-k_r}^{i_r}\ket{\lambda}\;\right| \; \lambda\in\Lambda_2\,,\;r\in\BZ_{\geq0}\right\}\,,
    \end{aligned}
\end{align}
where $\CH_\mathrm{NS}(C) = \CH_{\mathrm{NS}}^+(C)\cup\CH_{\mathrm{NS}}^-(C)$.
Thus $\Lambda_0$ and $\Lambda_2$ provide bosonic and fermionic states in the NS sector, respectively.
This means that the $\BZ_2$ grading fixed by the characteristic vector $\chi\in\Lambda(C)$ gives the one by the fermion parity $(-1)^F$.

On the other hand, the Ramond sector can be built from the vertex operators associated with the shadow $S(\Lambda(C))$
\begin{align}
    V_{\lambda+\frac{\chi}{2}}(z) = \,:e^{\i\left(\lambda+\frac{\chi}{2}\right)\cdot X(z)}:\,,\qquad \lambda\in\Lambda(C)\,.
\end{align}
These vertex operators have the following correlation function:
\begin{align}
    V_{\lambda+\frac{\chi}{2}}(z) \, V_{\lambda' + \frac{\chi}{2}} (w) = {(z-w)^{(\lambda+\frac{\chi}{2})\cdot(\lambda'+\frac{\chi}{2})}} \,V_{\lambda+\lambda' + \chi}(w) + \cdots\,,
\end{align}
where $\lambda$, $\lambda'\in\Lambda(C)$. Note that the NS sector operator appears on the right-hand side.
This reflects the fact that the product of two R sector operators returns an NS sector operator.

\begin{figure}
    \centering
    \begin{tikzpicture}[scale=1.5,cap=round,>=latex]
    \draw[decoration={markings, mark=at position 1 with {\arrow[very thick]{>}}}, postaction={decorate}] ({1/sqrt(2)},{1/sqrt(2)}) arc (45:395:1cm);
    \draw[decoration={markings, mark=at position 0.5 with {\arrow[very thick]{>}}}, postaction={decorate}] ({1/sqrt(2)},{1/sqrt(2)}) arc (45:395:1cm);
    \fill[fill=black!50!gray] (0,0) circle[radius=0.05cm];
    \draw (0,0) node [above left] {$w$};
    \draw[decorate,decoration={snake,amplitude=.4mm,segment length=2mm}] (2,0) -- (0,0);
    \fill[fill=black!50!gray] ({1/sqrt(2)},{1/sqrt(2)}) node [above right] {$z$} circle[radius=0.05cm];
    \end{tikzpicture}
    \caption{The operator product between the NS sector operator at the position $z$ and the R sector operator at $w$. Under the rotation around $w$, the NS sector operator is subject to the action of the fermion parity $(-1)^F$, which extends from the R sector operator.}
    \label{fig:NS_rotate}
\end{figure}
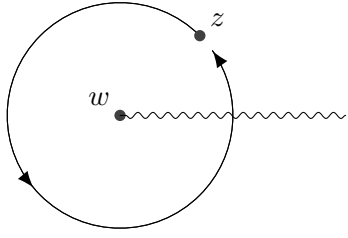

The operators in the Ramond sector change the periodicity of fermions, while they preserve the one of bosons \cite{Ginsparg:1988ui,Polchinski:1998rr}.
To see this property, consider the following operator product (a similar discussion has been done in \cite{Ashwinkumar:2021kav}):
\begin{align}
    V_{\lambda}(z) \, V_{\lambda'+\frac{\chi}{2}}(w) = {(z-w)^{\lambda\cdot(\lambda'+\frac{\chi}{2})}}\,V_{\lambda+\lambda'+\frac{\chi}{2}}(w) + \cdots\,,
\end{align}
where $V_{\lambda}(z)$ and $V_{\lambda'+\frac{\chi}{2}}(w)$ ($\lambda$, $\lambda'\in\Lambda(C)$) are in the NS and R sector, respectively.
Let us pick up one of them and move it around the other as in Fig.\ref{fig:NS_rotate}.
Under the rotation $z-w\to e^{2\pi\i}(z-w)$, the amplitude acquires the phase $e^{\pi\i\,\chi\cdot\lambda}=(-1)^{\chi\cdot\lambda}$. 
It tells us that while bosonic operators $V_{\lambda}(z)$ ($\lambda\in\Lambda_0$) are periodic under the rotation, fermionic ones $V_{\lambda}(z)$ ($\lambda\in\Lambda_2$) are anti-periodic. 
This can be understood from that R sector operators are non-local operators attached to the $(-1)^F$ line \cite{Lin:2019hks} (see Fig.\ref{fig:NS_rotate}).
Therefore, the NS operator $V_{\lambda}(z)$ receives the action of $(-1)^F$ as it goes around an operator in the Ramond sector.

Using the state-operator isomorphism, the Hilbert space $\CH_\mathrm{R}(C)$ of the Ramond sector is given by
\begin{align}
    \CH_{\mathrm{R}}(C) = \left\{\left.\alpha_{-k_1}^{i_1}\cdots\alpha_{-k_r}^{i_r}\ket{\tilde{\lambda}}\;\right| \; \tilde{\lambda}\in S(\Lambda(C))\right\}\,,
\end{align}
where $\alpha_k^i$ are the bosonic oscillators as in the NS sector.

We propose to define the $\BZ_2$ grading in the Ramond sector in analogy with the NS sector.
As in Fig.\ref{fig:NS_rotate}, we consider the product of two R sector operators and rotate an operator around the other.
The Ramond sector of the fermionic CFTs is well-defined only up to an overall fermion parity. 
We fix the ambiguity by considering the operator product
\begin{align}
    V_{\lambda+\frac{\chi}{2}}(z)\,V_{\frac{\chi}{2}}(w) = (z-w)^{(\lambda+\frac{\chi}{2})\cdot\frac{\chi}{2}}\, V_{\lambda+\chi}(w) +\cdots\,.
\end{align}
Under the rotation $z-w\to e^{2\pi\i}(z-w)$, the $\BZ_2$ line, which extends from $V_{\frac{\chi}{2}}(w)$, acts on the R sector operator $V_{\lambda+\frac{\chi}{2}}(z)$ and the above amplitude obtains the phase $\gamma = (-1)^{\chi\cdot(\lambda+\frac{\chi}{2})}$.
It is straightforward to know $\gamma = +1$ if $\lambda+\frac{\chi}{2}\in\Lambda_1$ and $\gamma=-1$ if $\lambda+\frac{\chi}{2}\in\Lambda_3$.
Therefore, the $\BZ_2$ grading of the Ramond sector is given by
\begin{align}
    \begin{aligned}
    \label{eq:R_grading}
        \CH_{\mathrm{R}}^+(C) &= \left\{\left.\alpha_{-k_1}^{i_1}\cdots\alpha_{-k_r}^{i_r}\ket{\tilde{\lambda}}\;\right| \; \tilde{\lambda}\in\Lambda_1\,,\;r\in\BZ_{\geq0}\right\}\,,\\
        \CH_{\mathrm{R}}^-(C) &= \left\{\left.\alpha_{-k_1}^{i_1}\cdots\alpha_{-k_r}^{i_r}\ket{\tilde{\lambda}}\;\right| \; \tilde{\lambda}\in\Lambda_3\,,\;r\in\BZ_{\geq0}\right\}\,.
    \end{aligned}
\end{align}
Combining \eqref{eq:NS_grading} and \eqref{eq:R_grading}, we see that the only difference between the four sectors is the underlying subset $\Lambda_i$.
Then we can interpret the four sectors of the Hilbert space as the four subsets $\Lambda_i$ defined in terms of lattices.
Along these lines, we show the Hilbert spaces of the fermionic CFTs (including the NS and R sectors) in Table \ref{table:sector}.
In the next section, we verify that the resulting partition functions exhibit the expected modular transformations.

\begin{table}
  \centering
  \begin{minipage}[t]{.45\linewidth}
  \centering
  \begin{tabular}{c|cc}
     & NS sector  &  R sector  \\
    \hline \\[-0.4cm]
    even  & $\CH_{\mathrm{NS}}^+$  & $\CH_{\mathrm{R}}^+$ \\[0.1cm]
    odd  & $\CH_{\mathrm{NS}}^-$   & $\CH_{\mathrm{R}}^-$
  \end{tabular}
  \end{minipage}
  \begin{minipage}[t]{.45\linewidth}
  \centering
  \begin{tabular}{c|cc}
      & NS sector  &  R sector  \\
    \hline\\[-0.4cm] 
    even  & $\Lambda_0$  & $\Lambda_1$ \\[0.1cm]
    odd  & $\Lambda_2$   & $\Lambda_3$ 
  \end{tabular}
  \end{minipage}
  \vspace{0.5cm}
  \caption{Each NS and R sector of fermionic code CFTs is divided into two sectors by the fermion parity (the left panel). The $\BZ_2$ grading \eqref{eq:NS_grading} and \eqref{eq:R_grading} tell us that each sector of fermionic code CFTs corresponds to a subset $\Lambda_i$ in terms of the lattice (the right panel).}
  \label{table:sector}
\end{table}

\subsection{Torus partition functions}
\label{ss:toruspart}
Let us consider fermionic code CFTs living on the torus with modulus $\tau = \tau_1 + \i\tau_2$.
The torus has two cycles, which we will call the spacial and timelike periodicity
\begin{align}
   \mathrm{spacial}: w \sim w + 2\pi\,,\quad \mathrm{timelike}:  w\sim w+2\pi\tau\,,
\end{align}
where $w=\sigma^1+\i\sigma^2$ is the cylindrical coordinate.
In analogy with the spacelike case, we call the timelike periodicity NS if fermion operators are antiperiodic, and R if fermions are periodic.
The torus has four sectors specified by (spatial, timelike) boundary conditions and we denote them by
\begin{equation}
\label{eq:def_sector_label}
    \mathrm{NS}: (\mathrm{NS}, \mathrm{NS})\,,\quad
    \widetilde{\mathrm{NS}}: (\mathrm{NS}, \mathrm{R})\,,\quad
    \mathrm{R}: (\mathrm{R}, \mathrm{NS})\,,\quad
    \widetilde{\mathrm{R}}: (\mathrm{R}, \mathrm{R})\,.
\end{equation}
These sectors correspond to the four different choices of spin structures associated with the torus.
The partition functions of each sector are ($q=e^{2\pi\i\tau}$)
\begin{align}
\begin{aligned}
\label{eq:def_partition_function_ff}
    Z_{\mathrm{NS}}(\tau;\Lambda(C)) &= \Tr_{\CH_{\mathrm{NS}}}\left[{q^{L_0-\frac{c}{24}}}\right] =  \frac{1}{\eta(\tau)^n} \sum_{\lambda\,\in\,\Lambda(C)}  q^{\frac{1}{2}\lambda^2}\,, \\
    Z_{\widetilde{\mathrm{NS}}}(\tau;\Lambda(C)) &= \Tr_{\CH_{\mathrm{NS}}}\left[{(-1)^F\,q^{L_0-\frac{c}{24}}}\right] = \frac{1}{\eta(\tau)^n} \sum_{\lambda\,\in\,\Lambda(C)} (-1)^{\chi\cdot\lambda} q^{\frac{1}{2}\lambda^2}\,, \\
    Z_{\mathrm{R}}(\tau;\Lambda(C)) &= \Tr_{\CH_{\mathrm{R}}}\left[{q^{L_0-\frac{c}{24}}}\right] = \frac{1}{\eta(\tau)^n} \sum_{\lambda\,\in\,\Lambda(C)} q^{\frac{1}{2}(\lambda+\frac{\chi}{2})^2}\,, \\
    Z_{\widetilde{\mathrm{R}}}(\tau;\Lambda(C)) &= \Tr_{\CH_{\mathrm{R}}}\left[{(-1)^F\,q^{L_0-\frac{c}{24}}}\right] = \frac{1}{\eta(\tau)^n} \sum_{\lambda\,\in\,\Lambda(C)} (-1)^{\chi\cdot(\lambda+\frac{\chi}{2})} q^{\frac{1}{2}(\lambda+\frac{\chi}{2})^2}\,.
    \end{aligned}
\end{align}
where $\chi$ is a characteristic vector of $\Lambda(C)$ and $\eta(\tau)$ is the Dedekind eta function
\begin{align}
    \eta(\tau) = q^{\frac{1}{24}}\prod_{i=1}^\infty \,(1-q^m)\,.
\end{align}

Alternatively, we can write down the partition functions in terms of the four subsets $\Lambda_i$ ($i=0,1,2,3$). Let us denote the associated partition functions by
\begin{align}
    Z(\tau;\Lambda_i) = \frac{1}{\eta(\tau)^n}\,\sum_{\lambda\,\in\,\Lambda_i}q^{\frac{1}{2}\lambda^2}\,,
\end{align}
where $\sum_{\lambda\,\in\,\Lambda_i}q^{\frac{1}{2}\lambda^2}$ is the theta function associated with $\Lambda_i$.
Our construction of fermionic CFTs naturally relates the $\BZ_2$ grading of the Hilbert space to four subsets $\Lambda_i$ as in Table~\ref{table:sector}.
Then, the partition functions of each sector are given by
\begin{align} \label{eq:sector_partition_functions}
    \begin{aligned}
    Z_\mathrm{NS}(\tau;\Lambda(C)) &= Z(\tau;\Lambda_0) + Z(\tau;\Lambda_2)\,,\\
    Z_{\widetilde{\mathrm{NS}}}(\tau;\Lambda(C)) &= Z(\tau;\Lambda_0) - Z(\tau;\Lambda_2)\,,\\
    Z_\mathrm{R}(\tau;\Lambda(C)) &= Z(\tau;\Lambda_1) + Z(\tau;\Lambda_3)\,,\\
    Z_{\widetilde{\mathrm{R}}}(\tau;\Lambda(C)) &= Z(\tau;\Lambda_1) - Z(\tau;\Lambda_3)\,.
    \end{aligned}
\end{align}

In the rest of this section, we explicitly check the modular transformations of the partition functions. 
While partition functions of bosonic CFTs should be modular invariant, we do not require modular invariance for fermionic CFTs.
Instead, we ask that partition functions associated with different spacial and timelike periodicity transform covariantly.
For this purpose, we denote the partition functions collectively as
\begin{align}
    \label{eq:collective_part}
    Z^{\alpha,\beta}(\tau;\Lambda(C)) = \frac{1}{\eta(\tau)^n} \sum_{\lambda\,\in\,\Lambda(C)} (-1)^{\beta\chi\cdot(\lambda+\alpha\frac{\chi}{2})} q^{\frac{1}{2}(\lambda+\alpha\frac{\chi}{2})^2}
\end{align}
where $(\alpha,\beta)$ corresponds to each sector by
\begin{align}
\label{eq:sector_notation}
    \mathrm{NS}: (0, 0)\,,\quad \widetilde{\mathrm{NS}}: (0, 1)\,,\quad
    \mathrm{R}: (1, 0)\,,\quad \widetilde{\mathrm{R}}: (1, 1)\,.
\end{align}
Note that $Z^{\alpha,\beta}(\tau;\Lambda(C))=Z^{\alpha',\beta'}(\tau;\Lambda(C))$ for integers $\alpha\equiv\alpha'$, $\beta\equiv\beta' \mod 2$.

\subsubsection{For odd prime $p\neq2$}
For an odd prime $p$, the canonical choice of a characteristic vector $\chi$ is given by the following proposition.
In what follows, we fix the choice of a characteristic vector by $\chi = \sqrt{p}\,(1,1,\cdots,1)$ for an odd prime case.

\begin{proposition}
Let $p$ be an odd prime and $C\subset\BF_p^n$ a self-dual code. Then $\chi=\sqrt{p}\,(1,1,\cdots,1)$ is a characteristic vector of the Construction A lattice $\Lambda(C)$.
\end{proposition}

\begin{proof}
From the construction of $\Lambda(C)$, $\chi=\sqrt{p}\,(1,1,\cdots,1)\in\Lambda(C)$ and $\lambda\in\Lambda(C)$ can be written as $\lambda=\frac{1}{\sqrt{p}}\,s$ ($s\in\mathbb{Z}^n$). Since $\Lambda(C)$ is self-dual and $p$ is odd, $\lambda\cdot\lambda$ is an integer and
\begin{equation}
    \lambda\cdot\lambda \equiv p\,\lambda\cdot\lambda = s\cdot s \mod 2\,.
\end{equation}
On the other hand, for any integer $s_i \equiv s_i^2 \mod 2$, thus
\begin{equation}
    \chi\cdot\lambda = \sum_{i=1}^n s_i \equiv \sum_{i=1}^n s_i^2 \mod 2\,.
\end{equation}
Combining them, we get
\begin{equation}
    \lambda\cdot\lambda = \chi\cdot\lambda \mod 2\,,
\end{equation}
which is the definition of a characteristic vector.
\end{proof}

The above choice of a characteristic vector $\chi = \sqrt{p}\,(1,1,\cdots,1)$ reduces the partition functions \eqref{eq:collective_part} to the complete weight enumerator of a classical linear code $C$ as follows:
\begin{proposition}
Let $p$ be an odd prime and $C \subset \BF_p^n$ a self-dual code. Then the partition functions of each sector can be written by the complete weight enumerator $W_C(\{x_a\})$ as
\begin{equation}
    Z^{\alpha,\beta}(\tau;\Lambda(C)) = \frac{1}{\eta(\tau)^n} W_C \left( f^{\alpha,\beta}_{0,p}(\tau),\,f^{\alpha,\beta}_{1,p}(\tau),\,\cdots,\,f^{\alpha,\beta}_{p-1,p}(\tau) \right)
\end{equation}
where
\begin{equation}
\label{eq:def_f_function}
    f^{\alpha,\beta}_{a,p}(\tau) = \i^{\alpha\beta p} \sum_{k\,\in\,\BZ}\, (-1)^{\beta(k+a)}\, q^{\frac{p}{2}\left(k+\alpha\frac{1}{2}+\frac{a}{p}\right)^2}.
\end{equation}
\label{prop:part_weight}
\end{proposition}

\begin{proof}
The proof is given by a straightforward computation.
For $\chi=\sqrt{p}\,(1,1,\cdots,1)$, the partition functions $Z^{\alpha,\beta}(\tau;\Lambda(C))$, excluding the Dedekind eta functions, are
\begin{align}
    \begin{aligned}
    &\sum_{\lambda\,\in\,\Lambda(C)}\, (-1)^{\beta\chi\cdot(\lambda+\alpha\frac{\chi}{2})} q^{\frac{1}{2}(\lambda+\alpha\frac{\chi}{2})^2}\,, \\
    =& \sum_{c\,\in\, C} \sum_{m\,\in\,\mathbb{Z}^n} (-1)^{\beta\sqrt{p}\sum_{i=1}^n \left(\frac{1}{\sqrt{p}}c_i+\sqrt{p}m_i+\alpha\frac{\sqrt{p}}{2}\right)} q^{\frac{1}{2}\sum_{i=1}^n\left(\frac{1}{\sqrt{p}}c_i+\sqrt{p}m_i+\alpha\frac{\sqrt{p}}{2}\right)^2}\,, \\
    =& \sum_{c\,\in\, C}\, \prod_{i=1}^n\, \i^{\alpha\beta p} \sum_{m_i\in\mathbb{Z}} (-1)^{\beta(c_i+m_i)} q^{\frac{p}{2}\left(\frac{1}{p}c_i+m_i+\alpha\frac{1}{2}\right)^2}\,, \\
    =& \sum_{c\,\in\, C}\, \prod_{i=1}^n\, f^{\alpha,\beta}_{c_i,p}
    = W_C \left (f^{\alpha,\beta}_{0,p}(\tau),f^{\alpha,\beta}_{1,p}(\tau),\cdots,f^{\alpha,\beta}_{p-1,p}(\tau) \right)\,.
    \end{aligned}
\end{align}
\end{proof}
Combining with \eqref{eq:sector_partition_functions}, Proposition~\ref{prop:part_weight} endows with the relationship between the spectrum of codes, lattices, and CFTs, which can be measured by the complete weight enumerator, the theta functions, and the partition functions, respectively.

We can use Proposition \ref{prop:part_weight} to calculate the modular transformations of the partition functions. Under $T: \tau\to\tau+1$ and $S: \tau\to -1/\tau$, $f^{\alpha,\beta}_{a,p}(\tau)$ becomes
\begin{align}
\label{eq:f_T_transform}
    f^{\alpha,\beta}_{a,p}(\tau+1) &= (-1)^\frac{\alpha(\alpha+1)}{2} e^{\i\pi\frac{p\alpha^2}{4}} e^{\i\pi \frac{p+1}{p} a^2} f^{\alpha,\alpha+\beta+1}_{a,p}(\tau)\,, \\
    f^{\alpha,\beta}_{a,p} (-1/\tau) &= \sqrt{-\i\tau} \,(-\i)^{\alpha\beta p}\, \frac{1}{\sqrt{p}} \,\sum_{b=0}^{p-1}\, e^{-2\pi\i \frac{ab}{p}} \,f^{\beta,\alpha}_{b,p}(\tau)\,.
\end{align}
This can be shown by direct computation using the Poisson summation formula. 

Let us see the modular $T$ transformations of the partition functions by using \eqref{eq:f_T_transform}.
Since all terms in the complete weight enumerator have the even degree $n$, the phase $(-1)^\frac{\alpha(\alpha+1)}{2}$ is canceled in the partition functions. In addition, the contribution of $e^{\i\pi \frac{p+1}{p} a^2}$ to each term from $c\in C$ is
\begin{align}
    \prod_{i=1}^n e^{\i\pi \frac{p+1}{p} {c_i}^2} = e^{\i\pi \frac{p+1}{p} c\cdot c} = 1
\end{align}
since $c\cdot c \in p\,\BZ$ from self-duality and $p+1\in2\BZ$. Then
\begin{align}
\label{eq:T_transformation}
\begin{aligned}
    Z^{\alpha,\beta}(\tau+1;\Lambda(C)) &= \frac{1}{\eta(\tau+1)^n}\, W_C(\{f^{\alpha,\beta}_{a,p}(\tau+1)\}) \,,\\
    &= e^{-\i\pi\frac{n}{12}} \frac{1}{\eta(\tau)^n}\, e^{\i\pi\frac{p \alpha^2}{4}n}\, W_C(\{f^{\alpha,\alpha+\beta+1}_{a,p}(\tau)\})\,, \\
    &= e^{\i\pi(3p\alpha^2-1)\frac{n}{12}}\, Z^{\alpha,\alpha+\beta+1}(\tau;\Lambda(C))\,, \\
    &= e^{\i\pi(3\alpha^2-1)\frac{n}{12}}\, Z^{\alpha,\alpha+\beta+1}(\tau;\Lambda(C))\,.
\end{aligned}
\end{align}
For the last equality, we used the fact that the length of a self-dual code over $\BF_p$ is $n\in2\BZ$ if $p\in4\BZ+1$ and $n\in4\BZ$ if $p\in4\BZ+3$. Then, the phase depends on $n$, which is the central charge of our fermionic CFTs, but not on an odd prime $p$.

Similarly, the modular $S$ transformation can be given by
\begin{align}
\begin{aligned}
\label{eq:S_transformation}
    Z^{\alpha,\beta}(-1/\tau;\Lambda(C)) &= \frac{1}{\eta(-1/\tau)^n}\, W_C(\{f^{\alpha,\beta}_{a,p}(-1/\tau)\})\,, \\
    &= (-1)^{\alpha\beta\frac{n}{2}} \,\frac{1}{\eta(\tau)^n}\, W_C\left(\left\{\frac{1}{\sqrt{p}} \sum_{b=0}^{p-1} e^{-2\pi\i \frac{ab}{p}} f^{\beta,\alpha}_{b,p}(\tau)\right\}\right)\,, \\
    &= (-1)^{\alpha\beta\frac{n}{2}} \frac{1}{\eta(\tau)^n} W_{C}(\{f^{\beta,\alpha}_{a,p}(\tau)\})\,, \\
    &= (-1)^{\alpha\beta\frac{n}{2}}\, Z^{\beta,\alpha}(\tau;\Lambda(C))\,.
    \end{aligned}
\end{align}
For the third equality we used the MacWilliams identity \eqref{eq:MacWilliams_id} for a self-dual code $C=C^\perp$.

Let us couple holomorphic fermionic code CFTs with anti-holomorphic Majorana-Weyl fermions and probe a gravitational anomaly of our fermionic code CFTs by implementing the modular transformations.
The torus partition functions of a single Majorana-Weyl fermion with $(c_L,c_R) = \left(\frac{1}{2},0\right)$ are
\begin{align}
\label{eq:part_weyl}
    Z_\mathrm{NS}(\tau;\psi) = \sqrt{\frac{\theta_3(\tau)}{\eta(\tau)}}\,,\quad
    Z_{\widetilde{\mathrm{NS}}}(\tau;\psi) = \sqrt{\frac{\theta_4(\tau)}{\eta(\tau)}}\,,\quad
    Z_\mathrm{R}(\tau;\psi) = \sqrt{\frac{\theta_2(\tau)}{\eta(\tau)}}\,,\quad
    Z_{\widetilde{\mathrm{R}}}(\tau;\psi) = 0\,,
\end{align}
where $\theta_i(\tau)$ ($i=2,3,4$) are the Jacobi theta functions.
Following the notation \eqref{eq:sector_notation}, we denote the partition functions by $Z^{\alpha,\beta}(\tau;\psi)$.
The modular transformations of these partition functions are well known (see for example \cite{Ginsparg:1988ui})
\begin{align}
\label{eq:majoranaweyl_T}
    Z^{\alpha,\beta}(\tau+1;\psi) &= e^{\i\pi\frac{3\alpha^2-1}{24}}\,Z^{\alpha,\alpha+\beta+1}(\tau;\psi)\,,\\
    Z^{\alpha,\beta}(-1/\tau;\psi) &= Z^{\beta,\alpha}(\tau;\psi)\,,
\end{align}
where we formally write down the modular transformations for $(\alpha,\beta)=(1,1)$, which is the $\widetilde{\mathrm{R}}$ sector, because $Z_{\widetilde{\mathrm{R}}}(\tau;\psi)$ is vanishing.

For a fermionic code CFT with $(c_L,c_R) = \left(n,0\right)$, we couple $N=2n$ anti-holomorphic Majorana-Weyl fermions with $(c_L,c_R) = (0,n)$.
The resulting non-chiral CFT with $(c_L,c_R) = (n,n)$ has partition functions
\begin{align}
    Z^{\alpha,\beta}(\tau,\bar{\tau}) = Z^{\alpha,\beta}(\tau;\Lambda(C)) \overline{Z^{\alpha,\beta}(\tau;\psi^{2n})}\,.
\end{align}
Note that since $Z^{1,1}(\tau;\psi)=0$, the whole partition function $Z^{1,1}(\tau,\bar{\tau})$ is also vanishing. 

The modular $T$ transformation of the partition functions is
\begin{align}
    Z^{\alpha,\beta}(\tau+1,\bar{\tau}+1) = Z^{\alpha,\alpha+\beta+1}(\tau,\bar{\tau})\,,
\end{align}
because the modular $T$ transformation of $2n$ Majorana-Weyl fermions from \eqref{eq:majoranaweyl_T} is exactly same as \eqref{eq:T_transformation}.
For the modular $S$ transformation, we have
\begin{align}
    Z^{\alpha,\beta}(-1/\tau,-1/\bar{\tau}) = Z^{\beta,\alpha}(\tau,\bar{\tau})\,,
\end{align}
where we use $Z_{1,1}(\tau,\bar{\tau})=0$.
These modular transformation laws exactly match modular covariance without a gravitational anomaly, which tells us that the gravitational anomaly of fermionic code CFTs can be canceled by one of the Majorana-Weyl fermions.
In this sense, our fermionic CFTs behave as properly as Majorana-Weyl fermions under modular transformation.

Note that the above discussion does not allow us to understand the modular property of $Z_{\widetilde{\mathrm{R}}}(\tau;\Lambda(C))$ because $Z_{\widetilde{\mathrm{R}}}(\tau;\psi)=0$.
We need an additional discussion for the $\widetilde{\mathrm{R}}$ sector.
According to \cite{Grigoletto:2021zyv}, the modular $S$ transformation of fermionic CFTs with $\nu = 2(c_R-c_L)$ mod $8$ should behave as
\begin{align}
    Z_{\widetilde{\mathrm{R}}}(-1/\tau,-1/\bar{\tau}) = e^{-\i\pi\frac{\nu}{4}}\,Z_{\widetilde{\mathrm{R}}}(\tau,\bar{\tau})\,.
\end{align}
For our fermionic CFTs with $(c_L,c_R) = (n,0)$, $\nu = -2n$ mod $8$. Then we have $e^{-\i\pi\frac{\nu}{4}} = (-1)^{\frac{n}{2}}$, which agrees with the transformation law \eqref{eq:S_transformation}.
This suggests that while our fermionic CFTs have a gravitational anomaly, their partition functions show the expected modular transformation laws.

For an odd prime $p$, we have formulated the construction of fermionic CFTs from self-dual codes over $\BF_p$, which relates linear codes, lattices, and CFTs. In Table \ref{table:dictionary_odd_prime}, we show a dictionary between them. This tells us that, for example, the complete weight enumerator $C$ determines the lattice theta function and the partition functions of fermionic code CFTs.

\begin{table}[t]
  \caption{A dictionary between codes, lattices, and CFTs for odd prime $p\neq2$}
  \label{table:dictionary_odd_prime}
  \centering
  {\small
  \begin{tabular*}{13.5cm}{@{\extracolsep{\fill}}ccc}
    \toprule
    Linear code & Lattice & Fermionic CFT  \\
    \midrule \\[-0.5cm]
    length $n$  & rank  & central charge \\[0.1cm] \\[-0.5cm]
    codeword $c$ & lattice vector $\lambda$ & momentum \\[0.1cm] \\[-0.5cm]
    & inner product $\chi\cdot\lambda$ & fermion parity $(-1)^F$ \\[0.1cm] \\[-0.5cm]
    linear code $C$ & Construction A lattice $\Lambda(C)$ & NS sector\\[0.1cm] \\[-0.5cm]
    & shadow of $\Lambda(C)$ & R sector \\[0.1cm] \\[-0.5cm]
    complete weight enumerator & lattice theta function & partition function \\
    \bottomrule
  \end{tabular*}
  }
\end{table}

\subsubsection{For $p=2$} \label{sec:partiton_function_p2}
In the previous section, we have canonically chosen a characteristic vector $\chi = \sqrt{p}\,(1,1,\cdots,1)$ for an odd prime $p$.
For $p=2$, however, the vector $\sqrt{2}\,(1,1,\cdots,1)$ is not characteristic for any Construction A lattice. 
To make a clear choice of a characteristic vector and define the shadow of $\Lambda(C)$, it is useful to introduce the shadow of a binary self-dual code \cite{conway1990newcode,brualdi1991weight}. 

Let $C\subset \BF_2^n$ be a singly-even self-dual code. 
We divide $C$ into two subsets $C_0$ and $C_2$. The subset $C_0$ ($C_2$) consists of doubly-even (singly-even) codewords of $C$:
\begin{align}
\begin{aligned}
\label{eq:def_C0_C2}
    C_0 &= \{ c\in C \mid \mathrm{wt}_1(c)\in4\BZ \}\,, \\
    C_2 &= \{ c\in C \mid \mathrm{wt}_1(c)\in4\BZ+2 \}\,,
\end{aligned}
\end{align}
where $C=C_0\cup C_2$. Note that $C_0$ is a linear code, in particular, doubly-even self-orthogonal. This is because for $c$, $c'\in C_0$, $\mathrm{wt}_1(c+c')=\mathrm{wt}_1(c)+\mathrm{wt}_1(c')-2c\cdot c'\in4\BZ$ and thus $c+c'\in C_0$.
In addition, we have $C_0\subset C\subset C_0^\perp$ from the self-duality of $C$. 
Then the shadow $S(C)$ of a self-dual code $C$ is defined by
\begin{align}
\label{eq:def_shadow_code}
    S(C) = C_0^\perp \setminus C\,.
\end{align}

From linearity and the dimension of the codes, any $t\in C_2$ and $s\in S(C)$ satisfy
\begin{align}
    C = C_0\cup(C_0+t)\,,\quad C_0^\perp = C\cup(C+s)\,,\quad S(C)=C+s\,.
\end{align}
Moreover, their inner products are
\begin{align}
    t\cdot s &= 1\,, \label{eq:ts} \\
    s\cdot s &= 
    \begin{dcases}
    0 & (n\in4\BZ)\,,\\
    1 & (n\in4\BZ+2)\,,
    \end{dcases}
    \label{eq:ss}
\end{align}
where we have done the mod $2$ computation.
The first equation \eqref{eq:ts} follows from the fact that if $t\cdot s=0$, then $s$ is orthogonal to $C_0\cup(C_0+t)=C$ and thus $s\in C^\perp=C$, which is contradict with $s\in S(C)$. For \eqref{eq:ss}, we consider $a=(1,1,\cdots,1)\in \BF_2^n$, which all self-dual codes contain. If $n\in4\BZ$, $s\cdot s=s\cdot a=0$ from  $a\in C_0$ and if $n\in4\BZ+2$, $s\cdot s=s\cdot a=1$ from $a\in C_2$.

Then, the elements $s\in S(C)$ in the shadow automatically satisfies
\begin{align}
\label{eq:parity_code}
    s\cdot c =
    \begin{dcases}
    0  & (c\in C_0)\\
    1  & (c\in C_2)
    \end{dcases}
    \quad \mod 2\,,
\end{align}
where we have used $s\in C_0^\perp$ and \eqref{eq:ts}.
This is analogous to the shadow of a lattice, where a characteristic vector $\chi$ divides a lattice $\Lambda$ by the mod $2$ value of $\chi\cdot\lambda$ for $\lambda\in\Lambda$ as in \eqref{eq:lattice_NS}.
The following proposition tells us that $s\in S(C)$ gives the canonical choice of a characteristic vector $\chi\in\Lambda(C)$ for a singly-even self-dual code $C\subset\BF_2^n$.

\begin{proposition}
Let $C\subset \BF_2^n$ be a singly-even self-dual code. Then, for any $s\in S(C)$, $\chi=\sqrt{2}s\in\BR^n$ is a characteristic vector of $\Lambda(C)$.
\end{proposition}
\begin{proof}
From the construction of the lattice $\Lambda(C)$, any $\lambda\in\Lambda(C)$ can be written as $\lambda=\frac{1}{\sqrt{2}}(c+2m)$ where $c\in C,m\in\BZ^n$. The inner product is
\begin{equation}
    \lambda\cdot\lambda = \frac{1}{2} (c^2+4c\cdot m+4m^2) \equiv \frac{1}{2} c^2 \equiv \left\{ \begin{aligned} &0 &&(c\in C_0) \\ &1 &&(c\in C_2) \end{aligned} \right. \mod 2\,.
\end{equation}
On the other hand,
\begin{equation} \label{eq:lambda_chi_p2}
    \lambda\cdot\chi = (c+2m)\cdot s \equiv c\cdot s \equiv \left\{ \begin{aligned} &0 &&(c\in C_0) \\ &1 &&(c\in C_2) \end{aligned} \right. \mod 2\,,
\end{equation}
where we used \eqref{eq:parity_code} in the last equation.
Combining these, we get
\begin{equation}
    \lambda\cdot\lambda \equiv \chi\cdot\lambda \mod 2,
\end{equation}
which is the definition of a characteristic vector.
\end{proof}

Let us define the Construction A subsets for later convenience.
By applying Construction A to a subset $K\subset\BF_2^n$, we obtain the following subset $\Lambda(K)\subset\BR^n$ (rather than a lattice):
\begin{align}
\label{eq:constsubset}
    \Lambda(K) = 
    \left\{v/\sqrt{2}\mid v\,\in\,\BZ^n,\;\,v=s\;\;\mathrm{mod}\;\, 2\,,\;\,s\in K\right\}\,.
\end{align}
The following corollary guarantees the relationship between the shadow of a code and the shadow of the Construction A lattice.

\begin{corollary}[{\cite{elkies1999lattices}}]
Let $C\subset \BF_2^n$ be a singly-even self-dual code. Then the shadow of $\Lambda(C)$ is identical to a subset constructed from the shadow of $C$:
\begin{equation}
    S(\Lambda(C)) = \Lambda(S(C))\,.
\end{equation}
\end{corollary}

\begin{proof}
Let $s\in S(C)$ and $\chi=\sqrt{2}\,s$ be a characteristic vector.
Suppose $k\in S(\Lambda(C))$, then it can be written as $k=\frac{1}{\sqrt{2}}(c+2m)+\frac{\chi}{2}=\frac{1}{\sqrt{2}}(c+s+2m)$ where $c\in C$, $m\in\BZ^n$, which is equivalent to $k\in\Lambda(S(C))$ since $c+s\in S(C)$.
\end{proof}

In what follows, we take a specific element $s\in S(C)$ and fix a characteristic vector by $\chi=\sqrt{2}s$. 
Let us define $C_1$, $C_3\subset \BF_2^n$ by
\begin{align}
\begin{aligned}
\label{eq:shadowcode_part}
    C_1 &= \begin{dcases}
        C_0 + s & (n\in4\BZ) \\
        C_2 + s & (n\in4\BZ+2)
    \end{dcases} \\
    C_3 &= \begin{dcases}
        C_2 + s & (n\in4\BZ) \\
        C_0 + s & (n\in4\BZ+2)
    \end{dcases} .
\end{aligned}
\end{align}
Note that $C_1$ and $C_3$ can be exchanged depending on the choice of $s\in S(C)$. For $\Lambda(C)$, the subsets $\Lambda_i\subset \BR^n$ corresponding to each sector can be written as
\begin{align}
\label{eq:coset_relat}
    \Lambda_i = \Lambda(C_i),\ i=0,1,2,3\,,
\end{align}
where we apply Construction A to the subsets $C_i$ ($i=0,1,2,3$).
We can prove this from \eqref{eq:parity_code} for $\Lambda_0$, $\Lambda_2$ and \eqref{eq:ss} for $\Lambda_1$, $\Lambda_3$.

In section \ref{ss:fermionCFT}, the $\BZ_2$ grading of the Hilbert space can be understood from underlying subsets $\Lambda_i$ as in Table \ref{table:sector}.
In particular, the binary case allows us to relate $\Lambda_i$ to a subset $C_i\subset\BF_2^n$ as in \eqref{eq:coset_relat}.
Therefore, we can go back further than lattices and give the $\BZ_2$ grading directly from classical codes $C$.
For example, we start with the subset $C_1$, which gives us $\Lambda(C_1) = \Lambda_1$ via Construction A. Then, from Table \ref{table:sector}, the associated vertex operator $:e^{\i\lambda\cdot X(z)}:$ where $\lambda\in\Lambda_1$ is in the R sector and even under the action of the fermion parity.
In this sense, we show the $\BZ_2$ grading of the Hilbert space in terms of the subsets $C_i\subset\BF_2^n$, in Table~\ref{table:sector_p=2}.

\begin{table}
  \centering
  \begin{tabular}{c|cc}
      & NS sector  &  R sector  \\
    \hline\\[-0.4cm] 
    even  & $C_0$  & $C_1$ \\[0.1cm]
    odd  & $C_2$   & $C_3$ 
  \end{tabular}
  \vspace{0.5cm}
  \caption{The $\BZ_2$ grading of the Hilbert space in terms of the subsets $C_i\subset\BF_2^n$. }
  \label{table:sector_p=2}
\end{table}

Let us check the modular property of the partition functions for $p=2$.
For notational convenience, we denote the partition function corresponding to a subset $K\subset \BF_2^n$ by
\begin{align}
\label{eq:part_subset}
    Z(\tau;K) = \frac{1}{\eta(\tau)^n} \sum_{\lambda\,\in\,\Lambda(K)} q^{\frac{1}{2}\lambda^2}\,, \quad q=e^{2\pi\i\tau}\,.
\end{align}
As we introduce the complete weight enumerator of a subset $K\subset\BF_2^n$
\begin{align}
\label{eq:subset_weight_enum}
    W_K(x_0,x_1) = \sum_{\kappa\,\in\,K} x_0^{\mathrm{wt}_0(\kappa)} x_1^{\mathrm{wt}_1(\kappa)}\,,
\end{align}
then the partition functions can be expressed as follows:
\begin{align}
\begin{aligned}
    Z(\tau;K) = \frac{1}{\eta(\tau)^n} W_K(\theta_{3}(2\tau),\theta_{2}(2\tau))\,,
\end{aligned}
\end{align}
where $\theta_3(\tau) := \sum_{n\in\BZ} q^{\frac{n^2}{2}}$ and $\theta_2(\tau):= \sum_{n\in\BZ} q^{\frac{1}{2}(n+\frac{1}{2})^2}$ are the Jacobi theta functions.
From \eqref{eq:sector_partition_functions}, the partition functions of each sector are 
\begin{align}
\label{eq:sector_partition_functions_p2}
\begin{aligned}
    Z_{\mathrm{NS}}(\tau;\Lambda(C)) &= Z(\tau;C_0) + Z(\tau;C_2) = Z(\tau;C)\,, \\
    Z_{\widetilde{\mathrm{NS}}}(\tau;\Lambda(C)) &= Z(\tau;C_0) - Z(\tau;C_2)\,, \\
    Z_{\mathrm{R}}(\tau;\Lambda(C)) &= Z(\tau;C_1) + Z(\tau;C_3) = Z(\tau;S(C))\,, \\
    Z_{\widetilde{\mathrm{R}}}(\tau;\Lambda(C)) &= Z(\tau;C_1) - Z(\tau;C_3)\, .
\end{aligned}
\end{align}

Under the modular $T$ transformation $\tau\to\tau+1$, the partition function of a subset $K\subset \BF_2^n$ behaves
\begin{align}
\begin{aligned}
    Z(\tau+1;K) &= \frac{1}{\eta(\tau+1)^n} W_K\left(\theta_3(2(\tau+1)),\theta_2(2(\tau+1))\right)\,, \\
    &= \frac{e^{-\i\pi\frac{n}{12}}}{\eta(\tau)^n} W_K(\theta_3(2\tau),\i\,\theta_2(2\tau))\,.
\end{aligned}
\end{align}
The modular $T$ transformation multiplies each term from $\kappa \in K$ by the phase $\i^{\mathrm{wt}_1(\kappa)}e^{-\i\pi\frac{n}{12}}$.
By definition, we know $\mathrm{wt}_1(\kappa)= 0$ mod $4$ for $\kappa\in C_0$ and  $\mathrm{wt}_1(\kappa) = 2$ mod $4$ for $\kappa\in C_2$.
Additionally, Corollary 1 and 2 of \cite{brualdi1991weight} tell us
$\mathrm{wt}_1(\kappa)= \frac{n}{2}$ mod $4$ for $\kappa\in S(C)$. Combining them, we conclude
\begin{align} \label{T_transformation_p2}
\begin{aligned}
    Z_{\mathrm{NS}}(\tau+1) &= (Z_{C_0}+Z_{C_2})(\tau+1) =  e^{-\i\pi\frac{n}{12}} (Z_{C_0}-Z_{C_2})(\tau) = e^{-\i\pi\frac{n}{12}}Z_{\widetilde{\mathrm{NS}}}(\tau)\,, \\
    Z_{\widetilde{\mathrm{NS}}}(\tau+1) &= (Z_{C_0}-Z_{C_2})(\tau+1) =  e^{-\i\pi\frac{n}{12}} (Z_{C_0}+Z_{C_2})(\tau) = e^{-\i\pi\frac{n}{12}}Z_{\mathrm{NS}}(\tau)\,, \\
    Z_{\mathrm{R}}(\tau+1) &= (Z_{C_1}+Z_{C_3})(\tau+1) = e^{\i\pi\frac{n}{6}} (Z_{C_1}+Z_{C_3})(\tau) = e^{\i\pi\frac{n}{6}}Z_{\mathrm{R}}(\tau)\,,\\
    Z_{\widetilde{\mathrm{R}}}(\tau+1) &= (Z_{C_1}-Z_{C_3})(\tau+1) = e^{\i\pi\frac{n}{6}} (Z_{C_1}-Z_{C_3})(\tau) = e^{\i\pi\frac{n}{6}} Z_{\widetilde{\mathrm{R}}}(\tau)\, ,
\end{aligned}
\end{align}
where we abbreviate $Z_{\ast}(\tau;\Lambda(C))$ as $Z_{\ast}(\tau)$ and $Z(\tau;K)$ as $Z_K(\tau)$.
This reproduces the same transformation law as \eqref{eq:T_transformation} for an odd prime $p$.

For the modular $S$ transformation, the following proposition is useful.

\begin{proposition}
Let $C$ be a binary linear code $C\subset\BF_2^n$ (not necessarily self-dual). Then, under the modular $S$ transformation: $\tau\to-1/\tau$, the partition function \eqref{eq:part_subset} of $C$ behaves
\begin{equation} \label{eq:Z_Strans_p2}
    Z(-1/\tau;C) = \frac{|C|}{2^{n/2}} Z(\tau;C^\perp) \,.
\end{equation}
\label{prop:modular_S}
\end{proposition}

\begin{proof}
Using the transformation law of the Jacobi theta functions and MacWilliams identity \eqref{eq:MacWilliams_id}, we get
\begin{align}
\begin{aligned}
    Z(-1/\tau;C) &= \frac{1}{\eta(-1/\tau)^n} W_{C}(\theta_3(2(-1/\tau)),\theta_2(2(-1/\tau))) \\
    &= \frac{1}{\eta(\tau)^n} W_{C}\left(\frac{1}{\sqrt{2}}(\theta_3(2\tau)+\theta_2(2\tau)),\frac{1}{\sqrt{2}}(\theta_3(2\tau)-\theta_2(2\tau))\right) \\
    &= \frac{|C|}{2^{n/2}} Z(\tau;C^\perp) \,.
\end{aligned}
\end{align}
\end{proof}

For a singly-even self-dual code $C\subset\BF_2^n$, we already know $C^\perp=C$ and $C_0^\perp=C_0\cup C_1\cup C_2\cup C_3$. In addition, from \eqref{eq:ts} and \eqref{eq:ss},
\begin{equation}
    (C_0\cup C_1)^\perp = \left\{ \begin{aligned} &C_0\cup C_1 &&(n\in4\BZ) \\ &C_0\cup C_3 &&(n\in4\BZ+2) \end{aligned} \right. .
\end{equation}
From Proposition \ref{prop:modular_S}, the $S$ transformation acts for the partition function as
\begin{align}
\begin{aligned}
    Z_{C_0} &\leftrightarrow \tfrac{1}{2} \left( Z_{C_0} + Z_{C_2} + Z_{C_1} + Z_{C_3} \right) \\
    Z_{C_2} &\leftrightarrow \tfrac{1}{2} \left( Z_{C_0} + Z_{C_2} - Z_{C_1} - Z_{C_3} \right) \\
    Z_{C_1} &\leftrightarrow \tfrac{1}{2} \left( Z_{C_0} - Z_{C_2} \pm Z_{C_1} \mp Z_{C_3} \right) \\
    Z_{C_3} &\leftrightarrow \tfrac{1}{2} \left( Z_{C_0} - Z_{C_2} \mp Z_{C_1} \pm Z_{C_3} \right)
\end{aligned}
\end{align}
where we take the upper sign when $n\in4\BZ$ and the lower sign when $n\in4\BZ+2$. Rearranging these equations, we get
\begin{align}
\begin{aligned} \label{S_transformation_p2}
    Z_{\mathrm{NS}}(-1/\tau) &= (Z_{C_0}+Z_{C_2})(-1/\tau) = (Z_{C_0}+Z_{C_2})(\tau) = Z_{\mathrm{NS}}(\tau)  \\
    Z_{\widetilde{\mathrm{NS}}}(-1/\tau) &= (Z_{C_0}-Z_{C_2})(-1/\tau) = (Z_{C_1}+Z_{C_3})(\tau) = Z_{\mathrm{R}}(\tau) \\
    Z_{\mathrm{R}}(-1/\tau) &= (Z_{C_1}+Z_{C_3})(-1/\tau) = (Z_{C_0}-Z_{C_2})(\tau) = Z_{\widetilde{\mathrm{NS}}}(\tau) \\
    Z_{\widetilde{\mathrm{R}}}(-1/\tau) &= (Z_{C_1}-Z_{C_3})(-1/\tau) = (-1)^{n/2} (Z_{C_1}-Z_{C_3})(\tau) = (-1)^{n/2} Z_{\widetilde{\mathrm{R}}}(\tau) .
\end{aligned}
\end{align}
In similar to the modular $T$ transformation, these obey the same transformation law as the case where $p$ is an odd prime. 

Therefore, the modular $T$ and $S$ transformations of the partition functions for $p=2$ are the same as an odd prime $p\neq2$.
Along the lines of an odd prime $p$, this suggests that our fermionic code CFTs exhibit the expected modular transformation while they have a gravitational anomaly.
A dictionary between codes, lattices, and CFTs is shown in Table~\ref{table:dictionary_2}.
Unlike the case with an odd prime $p\neq2$, the shadow of the Construction A lattice $\Lambda(C)$ is related to the shadow of a linear code $C$, which gives the interpretation of the Ramond sector from a linear code.

\begin{table}[t]
  \caption{A dictionary between codes, lattices, and CFTs for $p=2$}
  \label{table:dictionary_2}
  \centering
  {\small
  \begin{tabular*}{13.5cm}{@{\extracolsep{\fill}}ccc}
    \toprule
    Linear code & Lattice & Fermionic CFT  \\
    \midrule \\[-0.5cm]
    length $n$  & rank  & central charge \\[0.1cm] \\[-0.5cm]
    codeword $c$ & lattice vector $\lambda$ & momentum \\[0.1cm] \\[-0.5cm]
    inner product $s\cdot c$ & inner product $\chi\cdot\lambda$ & fermion parity $(-1)^F$ \\[0.1cm] \\[-0.5cm]
    linear code $C$ & Construction A lattice $\Lambda(C)$ & NS sector\\[0.1cm] \\[-0.5cm]
    shadow of $C$ & shadow of $\Lambda(C)$ & R sector \\[0.1cm] \\[-0.5cm]
    complete weight enumerator & lattice theta function & partition function \\
    \bottomrule
  \end{tabular*}
  }
\end{table}

\subsection{Spectral gap}
\label{ss:spectral}
In the construction of chiral bosonic CFTs from linear codes \cite{Dolan:1994st}, a distance of codewords is related to a conformal weight of the corresponding states.
For CFTs from linear codes with a large minimum distance, a spectral gap, which is the energy gap between the vacuum and the first excited state, tends to be large.
In this section, we explore a similar relation between the spectral gap of fermionic code CFTs and the distance of linear codes.

As we have seen, the operators that appear in our fermionic code CFTs are the vertex operators $:e^{\i k\cdot X(z)}:$ and their descendants, which have the conformal weights $h=\frac{1}{2}k^2+N$ where $N$ is the level for the Virasoro algebra.

Let us consider the minimum weight of the vertex operators in the NS sector, other than the identity operator. The minimum squared norm of the lattice $\Lambda(C)$ is
\begin{align}
\begin{aligned}
    \min_{\lambda\,\in\,\Lambda(C),\,\lambda\,\neq\,0}\lambda^2 &= \min\left\{ \min_{\substack{c\,\in\, C,\, c\,\neq\,0,\\ m\,\in\,\BZ^n}} \left( \frac{c+p\,m}{\sqrt{p}} \right)^2,\, \min_{m\,\in\,\BZ^n,\,m\,\neq\,0} (\sqrt{p}\,m)^2 \right\}\,, \\
    &= \min\left\{ \min_{\substack{c\,\in\, C,\, c\,\neq\,0,\\ m\,\in\,\{0,-1\}^n}} \frac{1}{p}\, (c+p\,m)^2,\; p \right\}\,, \\
    &= \min\left\{ \min_{c\,\in\, C,\, c\,\neq\,0} \frac{1}{p}\, \mathrm{Norm}(c), \;p \right\}\,,
\end{aligned}
\end{align}
where we have used the definition of $\mathrm{Norm}(c)$ for a codeword $c\in C$ in \eqref{eq:def_norm_c}.
Therefore, the minimum conformal weight of the vertex operators is
\begin{equation}
    \min_{\lambda\,\in\,\Lambda(C),\,\lambda\,\neq\,0} \frac{\lambda^2}{2} = \frac{1}{2} \min\left\{ \min_{c\,\in\, C,\, c\,\neq\,0} \frac{1}{p}\, \mathrm{Norm}(c), \;p \right\}\,.
\end{equation}
If we include the descendants, there are always states $\partial X\cong\alpha_{-1}\ket{0}$ with the conformal weight $h=1$ and thus the energy gap between the vacuum and the first excited state becomes
\begin{equation}
\label{eq:spectral_NS}
    \Delta_{\mathrm{NS}} = \min\left\{ \min_{c\,\in\, C,\, c\,\neq\,0} \frac{1}{2p}\, \mathrm{Norm}(c),\; 1 \right\}\,.
\end{equation}
Since $\frac{1}{p}\,\mathrm{Norm}(c)\in\BZ$ for a self-dual code, $\Delta_\mathrm{NS}$ only takes $\frac{1}{2}$ or $1$.
Note that, for $p=2$, the Euclidean norm reduces to the Hamming weight: $\mathrm{Norm}(c) = \mathrm{wt}(c)$.
Using the minimum distance $d(C)$, we can rewrite the spectral gap for $p=2$ as
\begin{align}
\label{eq:spectral_p=2}
    \Delta_{\mathrm{NS}} = \min\left\{\frac{d(C)}{4},\; 1 \right\}\,.
\end{align}

The R sector has no identity operator, thus the minimum weight in the R sector does not depend on whether we include descendants or not. For an odd prime $p$, a characteristic vector is chosen to be $\chi =\sqrt{p}\,(1,1,\cdots,1)$ and the minimum weight $h_{\min}$ of the R sector can be calculated by
\begin{equation}
\label{eq:min_weight_in_R}
\begin{aligned}
    h_{\min} = \min_{k\,\in\, S(\Lambda(C))} \frac{1}{2}\,k^2
    &= \min_{c\,\in \,C,\, m\,\in\,\BZ^n} \frac{1}{2}\, \left( \frac{c+p\,m}{\sqrt{p}}+\frac{\chi}{2} \right)^2 \\
    &= \min_{c\,\in\, C, \,m\,\in\,\BZ^n} \frac{1}{2p}\, \sum_{i=1}^n \,\left( c_i+p\,m_i+\frac{p}{2} \right)^2 \\
    &= \min_{c\,\in\, C} \frac{1}{2p}\, \sum_{i=1}^n\, \left( c_i-\frac{p}{2} \right)^2.
\end{aligned}
\end{equation}
In the last line, we used that $m_i=-1$ always gives the minimum value for $0\leq c_i<p$.

For the binary case ($p=2$),
\begin{equation}
\label{eq:min_weight_in_R_p2}
\begin{aligned}
    h_{\min} = \min_{k\,\in\, S(\Lambda(C))} \frac{1}{2}\,k^2
    &= \min_{s\,\in \,S(C),\, m\,\in\,\BZ^n} \frac{1}{2}\, \left( \frac{s+2m}{\sqrt{2}} \right)^2 \\
    &= \min_{s\,\in\, S(C)} \frac{1}{4}\,\mathrm{wt}_1(s)\,.
\end{aligned}
\end{equation}
Thus, the minimum conformal weight $h_{\min}$ in the R sector is proportional to the minimum Hamming weight in the shadow of the underlying code.

\subsection{Examples}
\label{ss:sec3example}

We have shown the construction of fermionic CFTs from self-dual codes over $\BF_p$ in general. In this section, we take some examples of self-dual codes and demonstrate the construction of fermionic code CFTs in detail. A large list of self-dual codes can be found in \cite{database}.

\subsubsection{$n=2\,,\,$ $p=5$} \label{sec:code_n2p5}

Let us consider a simple self-dual code $C\subset\BF_5^2$ generated by
\begin{align}
    G = \left[
    \begin{array}{cc}
        1 & 2
    \end{array}
    \right].
\end{align}
The linear code consists of only $5$ codewords: $C = \{(0,0),(1,2),(2,4),(3,1),(4,3)\}$.
The complete weight enumerator of $C$ is
\begin{align}
    W_C(\{x_a\}) = x_0^2 + x_1 x_2 + x_2 x_4 + x_3 x_1 + x_4 x_3\,.
\end{align}

The Construction A lattice of the linear code $C\subset\BF_5^2$ is given by
\begin{align}
    \Lambda(C) = \left\{\frac{c+5m}{\sqrt{5}}\in\BR^2 \;\middle|\;c\in C\,,\,m\in\BZ^2\right\}\,.
\end{align}
We show the Construction A lattice in Fig.\ref{fig:n=2p=5} as the black dots, which is just the $2$-dimensional square lattice~$\Lambda_{\mathrm{square}}=\BZ^2$ rotated by the angle $-\arctan(1/2)$.
Choosing a characteristic vector by $\chi = \sqrt{5}\,(1,1)$, the shadow of $\Lambda(C)$ is
\begin{align}
    S(\Lambda(C)) = \Lambda(C) +\frac{\chi}{2} = \left\{\frac{c+5m}{\sqrt{5}}\in\BR^2 \;\middle|\;c\in C\,,\,m\in\left(\BZ+\frac{1}{2}\right)^2\right\}\,.
\end{align}
In Fig.\ref{fig:n=2p=5}, the shadow of $\Lambda(C)$ is depicted with the red dots, which can be obtained by a half shift of the black dots.
It is easy to see the addition of shadow elements (red dots) returns to a black dot representing an element of the original Construction A lattice.

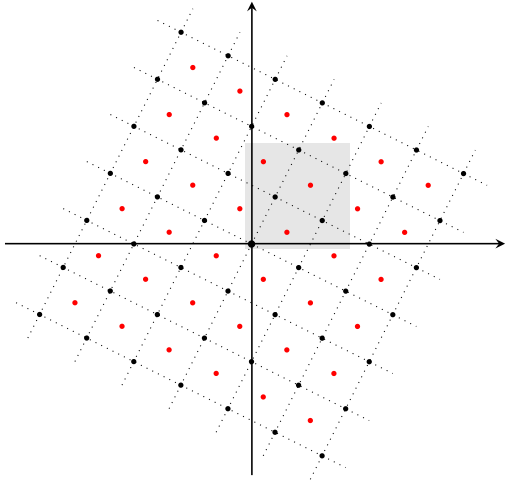
\begin{figure}[t]
    \centering
    \begin{tikzpicture}
    \begin{scope}[rotate around={-26.6:(2,2)}, scale=0.7, very thick]
    
    \begin{scope}[rotate=26.6]
    \fill[gray!20] (2.57,0.8) -- (4.55,0.8) -- (4.55,2.8) -- (2.57,2.8) -- cycle;
    \draw[semithick,->,>=stealth](-2,0.9)--(7.5,0.9);
    \draw[semithick,->,>=stealth](2.69,-3.5)--(2.69,5.5);
    \end{scope}
    
    \fill(0,0) circle[radius=0.05cm];
    \fill(1,0) circle[radius=0.05cm];
    \fill(0,1) circle[radius=0.05cm];
    \fill(-1,0) circle[radius=0.05cm];
    \fill(0,-1) circle[radius=0.05cm];
    \fill(1,1) circle[radius=0.05cm];
    \fill(-1,1) circle[radius=0.05cm];
    \fill(1,-1) circle[radius=0.05cm];
    \fill(-1,-1) circle[radius=0.05cm];
    
    \begin{scope}[xshift=3cm]
    \fill(0,0) circle[radius=0.05cm];
    \fill(1,0) circle[radius=0.05cm];
    \fill(0,1) circle[radius=0.05cm];
    \fill(-1,0) circle[radius=0.05cm];
    \fill(0,-1) circle[radius=0.05cm];
    \fill(1,1) circle[radius=0.05cm];
    \fill(-1,1) circle[radius=0.05cm];
    \fill(1,-1) circle[radius=0.05cm];
    \fill(-1,-1) circle[radius=0.05cm];
    \end{scope}
    
    \begin{scope}[yshift=3cm]
    \fill(0,0) circle[radius=0.05cm];
    \fill(1,0) circle[radius=0.05cm];
    \fill(0,1) circle[radius=0.05cm];
    \fill(-1,0) circle[radius=0.05cm];
    \fill(0,-1) circle[radius=0.05cm];
    \fill(1,1) circle[radius=0.05cm];
    \fill(-1,1) circle[radius=0.05cm];
    \fill(1,-1) circle[radius=0.05cm];
    \fill(-1,-1) circle[radius=0.05cm];
    \end{scope}
    
    \begin{scope}[xshift=3cm, yshift=3cm]
    \fill(0,0) circle[radius=0.05cm];
    \fill(1,0) circle[radius=0.05cm];
    \fill(0,1) circle[radius=0.05cm];
    \fill(-1,0) circle[radius=0.05cm];
    \fill(0,-1) circle[radius=0.05cm];
    \fill(1,1) circle[radius=0.05cm];
    \fill(-1,1) circle[radius=0.05cm];
    \fill(1,-1) circle[radius=0.05cm];
    \fill(-1,-1) node (a) {} circle[radius=0.07cm];
    \end{scope}
    
    \begin{scope}[thin]
    \draw[dotted] (-1.5,-1)--(5.5,-1);
    \draw[dotted] (-1.5,0)--(5.5,0);
    \draw[dotted] (-1.5,1)--(5.5,1);
    \draw[dotted] (-1.5,2)--(5.5,2);
    \draw[dotted] (-1.5,3)--(5.5,3);
    \draw[dotted] (-1.5,4)--(5.5,4);
    \draw[dotted] (-1.5,5)--(5.5,5);
    
    \draw[dotted] (-1,-1.5)--(-1,5.5);
    \draw[dotted] (0,-1.5)--(0,5.5);
    \draw[dotted] (1,-1.5)--(1,5.5);
    \draw[dotted] (2,-1.5)--(2,5.5);
    \draw[dotted] (3,-1.5)--(3,5.5);
    \draw[dotted] (4,-1.5)--(4,5.5);
    \draw[dotted] (5,-1.5)--(5,5.5);
    \end{scope}
    
    \fill(5,-1) circle[radius=0.05cm];
    \fill(5,0) circle[radius=0.05cm];
    \fill(5,1) circle[radius=0.05cm];
    \fill(5,2) circle[radius=0.05cm];
    \fill(5,3) circle[radius=0.05cm];
    \fill(5,4) circle[radius=0.05cm];
    \fill(5,5) circle[radius=0.05cm];
    \fill(4,5) circle[radius=0.05cm];
    \fill(3,5) circle[radius=0.05cm];
    \fill(2,5) circle[radius=0.05cm];
    \fill(1,5) circle[radius=0.05cm];
    \fill(0,5) circle[radius=0.05cm];
    \fill(-1,5) circle[radius=0.05cm];
    
    \begin{scope}[xshift=0.5cm, yshift=0.5cm, fill=red]
    \fill(0,0) circle[radius=0.05cm];
    \fill(1,0) circle[radius=0.05cm];
    \fill(0,1) circle[radius=0.05cm];
    \fill(-1,0) circle[radius=0.05cm];
    \fill(0,-1) circle[radius=0.05cm];
    \fill(1,1) circle[radius=0.05cm];
    \fill(-1,1) circle[radius=0.05cm];
    \fill(1,-1) circle[radius=0.05cm];
    \fill(-1,-1) circle[radius=0.05cm];
    
    \begin{scope}[xshift=3cm]
    \fill(0,0) circle[radius=0.05cm];
    \fill(1,0) circle[radius=0.05cm];
    \fill(0,1) circle[radius=0.05cm];
    \fill(-1,0) circle[radius=0.05cm];
    \fill(0,-1) circle[radius=0.05cm];
    \fill(1,1) circle[radius=0.05cm];
    \fill(-1,1) circle[radius=0.05cm];
    \fill(1,-1) circle[radius=0.05cm];
    \fill(-1,-1) circle[radius=0.05cm];
    \end{scope}
    
    \begin{scope}[yshift=3cm]
    \fill(0,0) circle[radius=0.05cm];
    \fill(1,0) circle[radius=0.05cm];
    \fill(0,1) circle[radius=0.05cm];
    \fill(-1,0) circle[radius=0.05cm];
    \fill(0,-1) circle[radius=0.05cm];
    \fill(1,1) circle[radius=0.05cm];
    \fill(-1,1) circle[radius=0.05cm];
    \fill(1,-1) circle[radius=0.05cm];
    \fill(-1,-1) circle[radius=0.05cm];
    \end{scope}
    
    \begin{scope}[xshift=3cm, yshift=3cm]
    \fill(0,0) circle[radius=0.05cm];
    \fill(1,0) circle[radius=0.05cm];
    \fill(0,1) circle[radius=0.05cm];
    \fill(-1,0) circle[radius=0.05cm];
    \fill(0,-1) circle[radius=0.05cm];
    \fill(1,1) circle[radius=0.05cm];
    \fill(-1,1) circle[radius=0.05cm];
    \fill(1,-1) circle[radius=0.05cm];
    \fill(-1,-1) circle[radius=0.05cm];
    \end{scope}
    
    \end{scope}
    \end{scope}
    \end{tikzpicture}
    \caption{The Construction A lattice (the black dots) and its shadow (the red dots), where we depict the origin with the thick black dot. The Construction A lattice turns out to be the $2$-dimensional square lattice rotated by $-\arctan(1/2)$. On the other hand, the shadow elements lie at the center of the squares. The shaded region gives the subspace $\BF_5^2$.}
    \label{fig:n=2p=5}
\end{figure}

Following Section \ref{ss:fermionCFT}, we construct a fermionic code CFT from $C$. Although the resulting partition functions are determined by the complete weight enumerator of $C$, it is useful to choose an appropriate basis of $\Lambda(C)$. In this case, we take $v_1 = \frac{1}{\sqrt{5}}(1,2)$ and $v_2 = \frac{1}{\sqrt{5}}(2,-1)$, which provide an orthonormal basis of $\Lambda(C)$. As suggested in Fig.\ref{fig:n=2p=5}, the Construction A lattice can be decomposed into $\BZ \oplus \BZ$. Then a characteristic vector is $\chi = \sqrt{5}\,(1,1) = 3v_1 + v_2$, so the shadow of $\Lambda(C)$ consists of $\widetilde{\lambda} = m_1'v_1 +m_2' v_2$ where $m_1'$, $m_2'\in \BZ+\frac{1}{2}$. Finally, we obtain the partition functions of each sector
\begin{align}
\begin{aligned}
\label{eq:p=5n=2_partition}
    Z_\mathrm{NS}(\tau;\Lambda(C)) &= \frac{\theta_3(\tau)^2}{\eta(\tau)^2}\,,\quad
    Z_{\widetilde{\mathrm{NS}}}(\tau;\Lambda(C)) = \frac{\theta_4(\tau)^2}{\eta(\tau)^2}\,,\\
    Z_\mathrm{R}(\tau;\Lambda(C)) &= \frac{\theta_2(\tau)^2}{\eta(\tau)^2}\,,\quad
    Z_{\widetilde{\mathrm{R}}}(\tau;\Lambda(C)) = 0\,.
\end{aligned}
\end{align}
The above partition functions agree with $\left(Z^{\alpha,\beta}(\tau;\psi)\right)^4$ where $Z^{\alpha,\beta}(\tau;\psi)$ is the torus partition functions of a single Majorana-Weyl fermion in \eqref{eq:part_weyl}. This implies that the fermionic CFT consists of four chiral fermions.
Of course, the spectral gap of the NS sector is $\Delta_\mathrm{NS} = \frac{1}{2}$ by \eqref{eq:spectral_NS}. On the other hand, the conformal weight of the R sector ground states is $h_{\min}= \frac{1}{4}$.

\subsubsection{$n=2\,,\,$ $p=2$}
Let us take the simplest self-dual code $C\subset\BF_2^2$, which is generated by
\begin{align}
    G = \left[
    \begin{array}{cc}
        1 & 1
    \end{array}
    \right].
\end{align}
The codewords are $C=\{(0,0),(1,1)\}$, then the complete weight enumerator of $C$ is
\begin{align}
    W_C(\{x_a\}) = x_0^2 + x_1^2\,.
\end{align}

To construct the shadow of $C$, we divide $C$ into two subsets with respect to the weight.
Since codewords $(0,0)$ and $(1,1)$ are doubly-even and singly-even, respectively, we have $C_0=\{(0,0)\}$ and $C_2 = \{(1,1)\}$. Then $C_0^\perp = \BF_2^2$. The shadow of $C$ is given by
\begin{align}
    S(C) = C_0^\perp\backslash C = \{(0,1),(1,0)\}\,.
\end{align}
We choose an element of the shadow $S(C)$ by $s=(0,1)\in S(C)$. Following \eqref{eq:shadowcode_part}, for $n=2$, we set
\begin{align}
    \begin{aligned}
        C_1 = C_2 + s = \{(1,0)\}\,,\quad
        C_3 = C_0 + s = \{(0,1)\}\,.
    \end{aligned}
\end{align}

The Construction A lattice of $C\subset\BF_2^2$ is 
\begin{align}
    \Lambda(C) = \left\{\frac{c+2m}{\sqrt{2}}\in\BR^2 \;\middle|\;c\in C\,,\,m\in\BZ^2\right\}\,.
\end{align}
On the other hand, the shadow of $\Lambda(C)$ is
\begin{align}
    S(\Lambda(C)) = \Lambda(S(C)) = \left\{\frac{s+2m}{\sqrt{2}}\in\BR^2 \;\middle|\;s\in S(C)\,,\,m\in\BZ^2\right\}\,.
\end{align}
As shown in Fig.\ref{fig:n=2p=2}, the Construction A lattice forms the square lattice $\BZ\oplus\BZ$ rotated by the angle $\pi/4$, while the elements of the shadow $S(\Lambda(C))$ take place in each center of the squares. 

Let us compute the partition functions depending on the choice of spin structures. Each sector is given in terms of the subsets $C_i\subset\BF_2^n$ $(i=0,1,2,3)$ in \eqref{eq:sector_partition_functions_p2}.
In this example, we have $C_0=\{(0,0)\}$, $C_2=\{(1,1)\}$, $C_1=\{(1,0)\}$, and $C_3=\{(0,1)\}$. Then the corresponding complete weight enumerators \eqref{eq:subset_weight_enum} are
\begin{align}
    W_{C_0}(\{x_a\}) = x_0^2\,,\quad W_{C_2}(\{x_a\}) = x_1^2\,,\quad
    W_{C_1}(\{x_a\}) = W_{C_3}(\{x_a\}) =x_0 x_1\,.
\end{align}
Replacing $x_0$ and $x_1$ with $\theta_3(2\tau)$ and $\theta_2(2\tau)$, respectively, we obtain the partition functions for each subset $C_i$ $(i=0,1,2,3)$
\begin{align}
\begin{aligned}
    Z(\tau;C_0) = \frac{\theta_3(2\tau)^2}{\eta(\tau)^2}\,,\quad Z(\tau;C_2) = \frac{\theta_2(2\tau)^2}{\eta(\tau)^2}\,,\quad
    Z(\tau;C_1) = Z(\tau;C_3) = \frac{\theta_2(2\tau)\theta_3(2\tau)}{\eta(\tau)^2}\,.
\end{aligned}
\end{align}
From \eqref{eq:sector_partition_functions_p2}, the torus partition functions of the fermionic code CFT are
\begin{align}
    \begin{aligned}
        Z_{\mathrm{NS}}(\tau;\Lambda(C)) &= \frac{\theta_3(2\tau)^2 + \theta_2(2\tau)^2}{\eta(\tau)^2} = \frac{\theta_3(\tau)^2}{\eta(\tau)^2}\,,\\
        Z_{\widetilde{\mathrm{NS}}}(\tau;\Lambda(C)) &= \frac{\theta_3(2\tau)^2 - \theta_2(2\tau)^2}{\eta(\tau)^2} = \frac{\theta_4(\tau)^2}{\eta(\tau)^2}\,,\\
        Z_{\mathrm{R}}(\tau;\Lambda(C)) &= \frac{2\theta_3(2\tau)\theta_2(2\tau)}{\eta(\tau)^2} = \frac{\theta_2(\tau)^2}{\eta(\tau)^2}\,,
    \end{aligned}
\end{align}
and $Z_{\widetilde{\mathrm{R}}}(\tau;\Lambda(C))=0$.
Here we used the identities $\theta_2(2\tau) = \sqrt{\frac{1}{2}(\theta_3(\tau)^2-\theta_4(\tau)^2)}$ and $\theta_3(2\tau) = \sqrt{\frac{1}{2}(\theta_3(\tau)^2+\theta_4(\tau)^2)}$.
In the third equality, we also applied the Jacobi identity $\theta_3(\tau)^4-\theta_4(\tau)^4 = \theta_2(\tau)^4$.
These partition functions are the same ones as the previous example \eqref{eq:p=5n=2_partition}, which is expected because lattices in Fig.\ref{fig:n=2p=5} and Fig.\ref{fig:n=2p=2} are related by the rotation around the origin. Additionally, the spectral gap $\Delta_\mathrm{NS}=\frac{1}{2}$ and the minimum weight $h=\frac{1}{4}$ of the R sector can be obtained from \eqref{eq:spectral_p=2} and \eqref{eq:min_weight_in_R_p2}, respectively, and agree with the previous example.

\begin{figure}[t]
    \centering
    \begin{tikzpicture}
    \begin{scope}[rotate around={-45:(2,2)}, scale=0.7, very thick]
    
    \begin{scope}[rotate=45]
    \fill[gray!20] (2.7,-0.2)--(3.8,-0.2)--(3.8,0.9)--(2.7,0.9)--cycle;
    \draw[semithick,->,>=stealth](-2,0)--(8,0);
    \draw[semithick,->,>=stealth](2.83,-4.8)--(2.83,5);
    \end{scope}
    
    \fill(0,0) circle[radius=0.05cm];
    \fill(1,0) circle[radius=0.05cm];
    \fill(0,1) circle[radius=0.05cm];
    \fill(-1,0) circle[radius=0.05cm];
    \fill(0,-1) circle[radius=0.05cm];
    \fill(1,1) circle[radius=0.05cm];
    \fill(-1,1) circle[radius=0.05cm];
    \fill(1,-1) circle[radius=0.05cm];
    \fill(-1,-1) circle[radius=0.05cm];
    
    \begin{scope}[xshift=3cm]
    \fill(0,0) circle[radius=0.05cm];
    \fill(1,0) circle[radius=0.05cm];
    \fill(0,1) circle[radius=0.05cm];
    \fill(-1,0) circle[radius=0.05cm];
    \fill(0,-1) circle[radius=0.05cm];
    \fill(1,1) circle[radius=0.05cm];
    \fill(-1,1) circle[radius=0.05cm];
    \fill(1,-1) circle[radius=0.05cm];
    \fill(-1,-1) circle[radius=0.05cm];
    \end{scope}
    
    \begin{scope}[yshift=3cm]
    \fill(0,0) circle[radius=0.05cm];
    \fill(1,0) circle[radius=0.05cm];
    \fill(0,1) circle[radius=0.05cm];
    \fill(-1,0) circle[radius=0.05cm];
    \fill(0,-1) circle[radius=0.05cm];
    \fill(1,1) circle[radius=0.05cm];
    \fill(-1,1) circle[radius=0.05cm];
    \fill(1,-1) circle[radius=0.05cm];
    \fill(-1,-1) circle[radius=0.05cm];
    \end{scope}
    
    \begin{scope}[xshift=3cm, yshift=3cm]
    \fill(0,0) circle[radius=0.05cm];
    \fill(1,0) circle[radius=0.05cm];
    \fill(0,1) circle[radius=0.05cm];
    \fill(-1,0) circle[radius=0.05cm];
    \fill(0,-1) circle[radius=0.05cm];
    \fill(1,1) circle[radius=0.05cm];
    \fill(-1,1) circle[radius=0.05cm];
    \fill(1,-1) circle[radius=0.05cm];
    \fill(-1,-1) node (a) {} circle[radius=0.07cm];
    \end{scope}
    
    \begin{scope}[thin]
    \draw[dotted] (-1.5,-1)--(5.5,-1);
    \draw[dotted] (-1.5,0)--(5.5,0);
    \draw[dotted] (-1.5,1)--(5.5,1);
    \draw[dotted] (-1.5,2)--(5.5,2);
    \draw[dotted] (-1.5,3)--(5.5,3);
    \draw[dotted] (-1.5,4)--(5.5,4);
    \draw[dotted] (-1.5,5)--(5.5,5);
    
    \draw[dotted] (-1,-1.5)--(-1,5.5);
    \draw[dotted] (0,-1.5)--(0,5.5);
    \draw[dotted] (1,-1.5)--(1,5.5);
    \draw[dotted] (2,-1.5)--(2,5.5);
    \draw[dotted] (3,-1.5)--(3,5.5);
    \draw[dotted] (4,-1.5)--(4,5.5);
    \draw[dotted] (5,-1.5)--(5,5.5);
    \end{scope}
    
    \fill(5,-1) circle[radius=0.05cm];
    \fill(5,0) circle[radius=0.05cm];
    \fill(5,1) circle[radius=0.05cm];
    \fill(5,2) circle[radius=0.05cm];
    \fill(5,3) circle[radius=0.05cm];
    \fill(5,4) circle[radius=0.05cm];
    \fill(5,5) circle[radius=0.05cm];
    \fill(4,5) circle[radius=0.05cm];
    \fill(3,5) circle[radius=0.05cm];
    \fill(2,5) circle[radius=0.05cm];
    \fill(1,5) circle[radius=0.05cm];
    \fill(0,5) circle[radius=0.05cm];
    \fill(-1,5) circle[radius=0.05cm];
    
    \begin{scope}[xshift=0.5cm, yshift=0.5cm, fill=red]
    \fill(0,0) circle[radius=0.05cm];
    \fill(1,0) circle[radius=0.05cm];
    \fill(0,1) circle[radius=0.05cm];
    \fill(-1,0) circle[radius=0.05cm];
    \fill(0,-1) circle[radius=0.05cm];
    \fill(1,1) circle[radius=0.05cm];
    \fill(-1,1) circle[radius=0.05cm];
    \fill(1,-1) circle[radius=0.05cm];
    \fill(-1,-1) circle[radius=0.05cm];
    
    \begin{scope}[xshift=3cm]
    \fill(0,0) circle[radius=0.05cm];
    \fill(1,0) circle[radius=0.05cm];
    \fill(0,1) circle[radius=0.05cm];
    \fill(-1,0) circle[radius=0.05cm];
    \fill(0,-1) circle[radius=0.05cm];
    \fill(1,1) circle[radius=0.05cm];
    \fill(-1,1) circle[radius=0.05cm];
    \fill(1,-1) circle[radius=0.05cm];
    \fill(-1,-1) circle[radius=0.05cm];
    \end{scope}
    
    \begin{scope}[yshift=3cm]
    \fill(0,0) circle[radius=0.05cm];
    \fill(1,0) circle[radius=0.05cm];
    \fill(0,1) circle[radius=0.05cm];
    \fill(-1,0) circle[radius=0.05cm];
    \fill(0,-1) circle[radius=0.05cm];
    \fill(1,1) circle[radius=0.05cm];
    \fill(-1,1) circle[radius=0.05cm];
    \fill(1,-1) circle[radius=0.05cm];
    \fill(-1,-1) circle[radius=0.05cm];
    \end{scope}
    
    \begin{scope}[xshift=3cm, yshift=3cm]
    \fill(0,0) circle[radius=0.05cm];
    \fill(1,0) circle[radius=0.05cm];
    \fill(0,1) circle[radius=0.05cm];
    \fill(-1,0) circle[radius=0.05cm];
    \fill(0,-1) circle[radius=0.05cm];
    \fill(1,1) circle[radius=0.05cm];
    \fill(-1,1) circle[radius=0.05cm];
    \fill(1,-1) circle[radius=0.05cm];
    \fill(-1,-1) circle[radius=0.05cm];
    \end{scope}
    
    \end{scope}
    \end{scope}
    \end{tikzpicture}
    \caption{The Construction A lattice (the black dots) and its shadow (the red dots). The Construction A lattice is the square lattice rotated by $\pi/4$, and the shadow elements lie at the center of the squares. The shaded region gives the subspace $\BF_2^2$.}
    \label{fig:n=2p=2}
\end{figure}
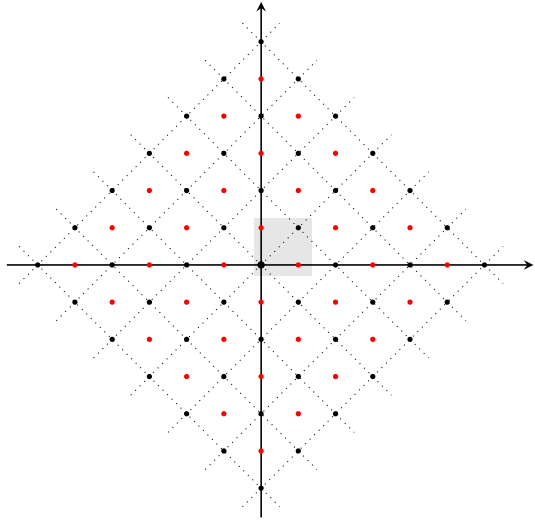

\subsubsection{$n=4\,,\,$ $p=7$}
For an odd prime $p\in4\BZ+3$, a self-dual code $C\subset\BF_p^n$ exists only when $n\in4\BZ$.
For $p=7$, 
the unique self-dual code $C$ of length $4$ is generated by
\begin{align}
    G = \left[
    \begin{array}{cccc}
        1 & 0 & 2 & 3 \\
        0 & 1 & 4 & 2
    \end{array}
    \right].
\end{align}
Then the Construction A lattice $\Lambda(C)$ is isomorphic to $\BZ^4$, which is consistent with the fact that an $n$-dimensional odd self-dual lattice is isomorphic to $\BZ^n$ for $n\leq8$ \cite{conway2013sphere}.
The partition functions of each sector are
\begin{align}
    \begin{aligned}
        Z_{\mathrm{NS}}(\tau;\Lambda(C)) &= \frac{\theta_3(\tau)^4}{\eta(\tau)^4}\,,\qquad
        Z_{\widetilde{\mathrm{NS}}}(\tau;\Lambda(C)) = \frac{\theta_4(\tau)^4}{\eta(\tau)^4}\,,\\
        Z_{\mathrm{R}}(\tau;\Lambda(C)) &= \frac{\theta_2(\tau)^4}{\eta(\tau)^4}\,,\qquad
        Z_{\widetilde{\mathrm{R}}}(\tau;\Lambda(C)) = 0\,,
    \end{aligned}
\end{align}
which implies that the fermionic code CFT consists of 8 Majorana-Weyl fermions.
The above examples with a relatively small length demonstrate our construction for $p=2$ and an odd prime $p\neq2$. In the next examples, our construction provides non-trivial fermionic CFTs.

\subsubsection{$n=12\,,\,$ $p=5$} \label{sec:code_n12p5}
Let us consider a linear code $C\subset (\BF_5)^{12}$ generated by (\cite{database})
\begin{equation}
    G =
    \begin{bmatrix}
    1&0&0&2&0&3&0&2&0&3&3&3 \\
    0&1&0&4&0&1&0&4&0&1&4&2 \\
    0&0&1&2&0&0&0&0&0&0&4&3 \\
    0&0&0&0&1&2&0&0&0&0&1&2 \\
    0&0&0&0&0&0&1&2&0&0&4&3 \\
    0&0&0&0&0&0&0&0&1&2&1&2
    \end{bmatrix}.
\end{equation}
Using the complete enumerator polynomial $W_C(\{x_a\})$, the partition functions of the corresponding CFT are
\begin{align}
\begin{aligned}
\label{eq:n=12p=5_partition}
    Z_{\mathrm{NS}}(\tau) &= q^{-\frac{1}{2}} + 276q^{\frac{1}{2}} + 2048q + 11202q^{\frac{3}{2}} + 49152q^2 + 184024q^{\frac{5}{2}} + 614400q^3 + \cdots\,, \\
    Z_{\widetilde{\mathrm{NS}}}(\tau) &= q^{-\frac{1}{2}} + 276q^{\frac{1}{2}} - 2048q + 11202q^{\frac{3}{2}} - 49152q^2 + 184024q^{\frac{5}{2}} - 614400q^3 + \cdots\,, \\
    Z_{\mathrm{R}}(\tau) &= 24 + 4096q + 98304q^2 + 1228800q^3 + 10747904q^4 + 74244096q^5 + \cdots \,,\\
    Z_{\widetilde{\mathrm{R}}}(\tau) &= -24\,.
\end{aligned}
\end{align}
Note that $Z_\mathrm{NS}(\tau)$ does not contain the $q^0$ term, which implies the absence of NS primary fields of the conformal weight $1/2$.
This can be thought of as a fermionic analog of the Monster CFT~\cite{frenkel1984natural,frenkel1989vertex}, which does not have any bosonic excitation of weight $1$.

There is a notion generalizing the Monster CFT called the extremal CFT.
For a bosonic CFT with the central charge $c=24k$ ($k\in\BZ_{>0}$),
it is called extremal if it does not include primary fields other than the identity of weight less than $k$ \cite{hohn2008conformal,Jankiewicz:2006mv}.
The Monster CFT is a prominent example of bosonic extremal CFT with $k=1$.

For a fermionic CFT with the central charge $c=12k^*$ ($k^*\in\BZ_{>0}$), an extremal CFT is defined to be one with no primary fields other than the identity of weight less than $k^*/2$~\cite{Witten:2007kt}.
An extremal fermionic CFT of $k^*=1$ is known to admit an action of a large discrete group related to the Conway group and have an $\CN=1$ supersymmetry. Furthermore, there are results about uniqueness \cite{hohn2007selbstduale,duncan2007super}.

From \eqref{eq:n=12p=5_partition}, the fermionic code CFT with $n=12$ has no primary fields other than the identity of weight less than $1/2$.
In particular, $Z_{\mathrm{NS}}(\tau)$ and $Z_{\mathrm{R}}(\tau)$ coincide with (3.35) and (3.48) in \cite{Witten:2007kt} respectively, which means the extremal CFT with $k^*=1$.

\subsubsection{$n=24\,,\,$ $p=5$}
For a self-dual code $C\subset (\BF_5)^{24}$ ($Q_{24}$ in \cite{LEON1982178}) generated by 
\begin{equation}
    G = 
    \begin{bmatrix}
    2&0&0&0&0&0&0&0&0&0&0&0&0&1&1&1&1&1&1&1&1&1&1&1\\
    0&2&0&0&0&0&0&0&0&0&0&0&4&0&4&1&4&4&4&1&1&1&4&1\\
    0&0&2&0&0&0&0&0&0&0&0&0&4&1&0&4&1&4&4&4&1&1&1&4\\
    0&0&0&2&0&0&0&0&0&0&0&0&4&4&1&0&4&1&4&4&4&1&1&1\\
    0&0&0&0&2&0&0&0&0&0&0&0&4&1&4&1&0&4&1&4&4&4&1&1\\
    0&0&0&0&0&2&0&0&0&0&0&0&4&1&1&4&1&0&4&1&4&4&4&1\\
    0&0&0&0&0&0&2&0&0&0&0&0&4&1&1&1&4&1&0&4&1&4&4&4\\
    0&0&0&0&0&0&0&2&0&0&0&0&4&4&1&1&1&4&1&0&4&1&4&4\\
    0&0&0&0&0&0&0&0&2&0&0&0&4&4&4&1&1&1&4&1&0&4&1&4\\
    0&0&0&0&0&0&0&0&0&2&0&0&4&4&4&4&1&1&1&4&1&0&4&1\\
    0&0&0&0&0&0&0&0&0&0&2&0&4&1&4&4&4&1&1&1&4&1&0&4\\
    0&0&0&0&0&0&0&0&0&0&0&2&4&4&1&4&4&4&1&1&1&4&1&0
    \end{bmatrix},
\end{equation}
the partition functions of the corresponding CFT are
\begin{align}
\begin{aligned}
    Z_{\mathrm{NS}}(\tau) &= q^{-1} + 24 + 4096q^{\frac{1}{2}} + 98580q + 1228800q^{\frac{3}{2}} + 10745856q^2 + \cdots\,, \\
    Z_{\widetilde{\mathrm{NS}}}(\tau) &= q^{-1} + 24 - 4096q^{\frac{1}{2}} + 98580q - 1228800q^{\frac{3}{2}} + 10745856q^2 + \cdots\,, \\
    Z_{\mathrm{R}}(\tau) &= 48 + 196608q + 21495808q^2 + 864288768q^3 + \cdots\,, \\
    Z_{\widetilde{\mathrm{R}}}(\tau) &= 48\,.
\end{aligned}
\end{align}
Note that the descendants $\alpha_{-1}^\mu\ket{0}$ contribute as $24 q^0$ to $Z_{\mathrm{NS}}(\tau)$, although there are no $\lambda\in\Lambda(C)$ such that $\lambda^2=1,2$ (corresponding to $q^{-\frac{1}{2}},q^0$),
Therefore, the minimum norm of the Construction A lattice $\Lambda(C)$ is $3$.

In $24$ dimensions, an odd self-dual lattice of minimum norm $3$ is known to be unique and called the odd Leech lattice \cite{o1944construction,conway2013sphere} with the lattice theta function
\begin{align}
    \Theta_{\Lambda_{\rm odd\, Leech}}(\tau) = 1 + 4096 q^{3/2} + 98256 q^2 + 1130496 q^{5/2} + 8384512 q^3 + \cdots\,.
\end{align}
Thus, we can conclude that $\Lambda(C)$ is the odd Leech lattice. This is consistent with the fact that the odd Leech lattice can be constructed using some self-dual codes over $\BZ_k$ by Construction A for any $k\geq3$ \cite{miezaki2012frames}.
For a binary self-dual code $C\subset \BF_2^{24}$, the Construction A lattice always contains an element $\lambda=(2,0,\cdots,0)/\sqrt{2}$ with $\lambda^2=2$. Then, any binary code cannot generate the odd Leech lattice since it does not have any element with $\lambda^2=1,2$.

\section{Finding supersymmetric CFTs}
\label{sec:SUSY}

We have constructed chiral fermionic CFTs from self-dual codes over $\BF_p$ and illustrated the construction using various examples.
For the ternary case $(p=3)$, all fermionic code CFTs have been known to possess supersymmetry \cite{Gaiotto:2018ypj}. On the other hand, supersymmetry does not necessarily emerge for the other cases $(p\neq3)$.
To figure out profiles of our class of fermionic CFTs, we search for supersymmetric  CFTs for a prime number $p$.
Instead of giving an explicit form of supercurrent, we require some necessary (not sufficient) conditions that strongly suggest the existence of supersymmetry following \cite{Bae:2021lvk,Bae:2020xzl}.
Recasting the conditions in terms of classical codes, we apply them to examples of self-dual codes and construct CFTs likely to be supersymmetric.

\subsection{Supersymmetry conditions}
Let us start with reviewing the fundamentals of supersymmetric CFTs following \cite{Polchinski:1998rr}.
In this section, we only care about the existence of supersymmetry.
Then we focus on $\CN=1$ supersymmetry, in particular, $\left(\CN,\bar{\CN}\right) = (1,0)$ since our fermionic CFTs are chiral.
The energy-momentum tensor $T(z)$ and the supercurrent $G(z)$ satisfy the following OPEs:
\begin{align}
\begin{aligned}
\label{eq:scft_ope}
    T(z)\,T(0) &\sim \frac{n}{2z^4} + \frac{2}{z^2}\,T(0) + \frac{1}{z}\,\partial T(0)\,, \\
    T(z)\,G(0) &\sim \frac{3}{2z^2}\,G(0) + \frac{1}{z} \,\partial G(0)\,, \\
    G(z)\,G(0) &\sim \frac{2n}{3z^2} + \frac{2}{z}\,T(0)\,,
\end{aligned}
\end{align}
where $n\in \BZ_{>0}$ is a central charge.
From this relation, $G(z)$ has conformal weight $h=\frac{3}{2}$. The Laurent expansions are
\begin{equation}
    T(z) = \sum_{m\,\in\,\BZ} \frac{L_m}{z^{m+2}}\,,\qquad
    G(z) = \sum_{r\,\in\,\BZ+\nu} \frac{G_r}{z^{r+3/2}}\,, 
\end{equation}
where $\nu=0$ for the NS sector and $\nu=1/2$ for the R sector.
The operators $\{L_m\}$ and $\{G_r\}$ form the superconformal algebra:
\begin{align}
\begin{aligned}
    [L_m,L_k] &= (m-k)L_{m+k} + \frac{n}{12}(m^3-m)\,\delta_{m+k,0}\,, \\
    [L_m,G_r] &= \left(\frac{1}{2}\,m-r\right)G_{m+r}\,, \\
    \{G_r,G_s\} &= 2L_{r+s} + \frac{n}{3}\left(r^2-\frac{1}{4}\right)\delta_{r+s,0}\,.
\end{aligned}
\end{align}
Since $[L_0,G_0]=0$, for a state $\ket{k}$ in the R sector such that $L_0\ket{k}=h\ket{k}$, $G_0\ket{k}$ has the same eigenvalue $h$ of $L_0$ and the opposite statistics unless $G_0\ket{k}=0$. The norm is
\begin{equation}
    \|\,G_0\ket{k}\|^2=\braket{k|\left(L_0-\frac{n}{24}\right)|k}=\left(h-\frac{n}{24}\right)\braket{k|k}\,,
\end{equation}
thus $G_0\ket{k}=0$ if and only if $h=\frac{n}{24}$ and the condition $h\geq \frac{n}{24}$ is necessary for the representation to be unitary.

For a unitary supersymmetric theory, we have the partition function of the $\widetilde{\mathrm{R}}$ sector: $Z_{\widetilde{\mathrm{R}}}(\tau)=\Tr[(-1)^F q^{L_0-\frac{n}{24}}]$ where $F$ is the fermion parity. 
\begin{itemize}
    \item For $h>\frac{n}{24}$, $\ket{k}$ and $G_0\ket{k}$ contribute as the opposite signs and thus vanish.
    \item For $h=\frac{n}{24}$, each $\ket{k}$ contributes as $\pm 1$ since $q^{L_0-\frac{n}{24}}=1$.
\end{itemize}
Therefore, $Z_{\widetilde{\mathrm{R}}}(\tau)$ becomes constant and counts the signed number of states with $h=\frac{n}{24}$, which is called the Witten index \cite{Witten:1982im,Witten:1982df}.

In this section, we aim to find fermionic code CFTs with supersymmetry.
Our fermionic CFTs have been defined by a set of vertex operators associated with the underlying momentum lattices.
We can show the existence of supersymmetry by giving a supercurrent satisfying the OPEs \eqref{eq:scft_ope}.
However, in general, it is hard to construct an explicit supercurrent from vertex operators because it requires careful treatment of cocycle factors.\footnote{Recently, it has been pointed out that a fundamental property of quantum error-correcting codes can be used to construct supercurrents from vertex operators \cite{Harvey:2020jvu}.}
Instead, we check the following three necessary conditions derived from the above discussion: 
\begin{enumerate}
    \item The Neveu-Schwarz sector contains a spin-$\frac{3}{2}$ Virasoro primary field.
    \item Any primary operator in the Ramond sector satisfies $h\geq \frac{n}{24}$.
    \item The Ramond-Ramond partition function $Z_{\widetilde{\mathrm{R}}}(\tau)=\Tr[(-1)^F q^{L_0-\frac{n}{24}}]$ is constant.
\end{enumerate}
Recently, it has been proved that the second condition implies the third one \cite{Bae:2021jkc}.
The above three conditions do not guarantee supersymmetry, but they strongly suggest the existence of supersymmetry.
In what follows, we call them the supersymmetry (SUSY) conditions.

In particular, the third condition is highly nontrivial because bosonic and fermionic excitations cancel each other.
It is natural to require supersymmetry as a mechanism for the cancellation.
Note that $Z_{\widetilde{\mathrm{R}}}(\tau)=0$ if a fermionic CFT contains a free fermion. In that case, the third condition does not provide nontrivial evidence of supersymmetry.
In the rest of this section, we further comment on free fermions in our fermionic CFTs.

In general, it is known that a holomorphic CFT can be decomposed into free fermions (including $h=\frac{1}{2}$ operators) and a sector without $h=\frac{1}{2}$ operators \cite{Goddard:1988wv}. For fermionic code CFTs, we can rephrase it as the decomposition of momentum lattices.
From the following discussion, an odd self-dual lattice $\Lambda\subset\BR^n$ can be decomposed into an integer lattice $\BZ^r$ $(0\leq r\leq n)$ corresponding to $r$ pairs of free fermions and a lattice orthogonal to it.

Let $\Lambda\subset\BR^n$ be a self-dual lattice including a norm 1 vector $\lambda_1\in\Lambda$: $\lambda_1\cdot\lambda_1=1$. Then we can decompose any $\lambda\in\Lambda$ into $\lambda_1$ direction and its orthogonal component by $\lambda=m\lambda_1+\lambda_\perp$ where $m=\lambda_1\cdot\lambda\in\BZ$ from self-duality. Thus, $\Lambda$ can be decomposed into the 1-dimensional integer lattice in $\lambda_1$ direction and its orthogonal component: $\Lambda=\lambda_1\BZ\oplus\Lambda_\perp$ where $\Lambda_\perp=\{\lambda\in\Lambda \mid \lambda_1\cdot\lambda=0\}$.
In this case, the partition function of the NS sector can be written as
\begin{align}
\begin{aligned}
	Z_{\mathrm{NS}}(\tau;\Lambda) &= \frac{1}{\eta(\tau)^n}\, \sum_{\lambda\,\in\,\Lambda}\, q^{\frac{1}{2}\lambda^2} \,,\\
	&= \frac{1}{\eta(\tau)^n}\, \sum_{m\,\in\,\BZ}\, \sum_{\lambda_\perp\in\,\Lambda_\perp}\, q^{\frac{1}{2}(m^2+\lambda_\perp^2)}\,, \\
	&= \left( \frac{1}{\eta(\tau)} \,\sum_{m\,\in\,\BZ}\, q^{\frac{1}{2}m^2} \right) \left( \frac{1}{\eta(\tau)^{n-1}} \sum_{\lambda_\perp\in\,\Lambda_\perp}\, q^{\frac{1}{2}\lambda_\perp^2} \right)\,.
\end{aligned}
\end{align}
The equivalence between two Majorana-Weyl fermions $\psi^1,\psi^2$ and one chiral boson $X(z)$
\begin{equation}
    \psi(z)=\frac{1}{\sqrt{2}}(\psi^1+\i\psi^2)(z) \cong e^{\i X(z)}\,,\quad
    \bar{\psi}(z)=\frac{1}{\sqrt{2}}(\psi^1-\i\psi^2)(z) \cong e^{-\i X(z)},
\end{equation}
tells us that the left part of $Z_{\mathrm{NS}}(\tau;\Lambda)$ is the same as the partition function of a pair of free fermions~\cite{Polchinski:1998rr}. Therefore, we can conclude that a pair of free fermions corresponds to $\lambda_1$ direction in the lattice. Since $\Lambda_\perp$ can be regarded as a self-dual lattice in $\BR^{n-1}$, we can repeat this procedure until no $h=\frac{1}{2}$ operators are left.

\subsection{Constraints on codes}
In the previous section, we have declared the three necessary conditions for supersymmetry.
Our construction of fermionic CFTs enables us to rewrite them simply in terms of classical codes.
In this section, we reduce the supersymmetry conditions to some constraints on underlying $p$-ary codes for an odd prime $p$ and $p=2$.

\subsubsection{For odd prime $p\neq2$}
Let us start with the supersymmetry condition 1, which imposes the existence of a spin-$\frac{3}{2}$ primary field in the NS sector.
As we have seen, the Virasoro primaries in our fermionic CFTs are only $V_k(z) = \,:e^{ik\cdot X(z)}:$ with the conformal weight $h=\frac{1}{2}k^2$. Then the first condition is equivalent to the existence of $\lambda\in\Lambda(C)$ with $\lambda^2=3$. Since the Construction A lattice satisfies $\lambda^2=\frac{1}{p}(c+p\,m)^2$ for $c\in C$, $m\in\BZ^n$ and $0\leq c_i<p$, it is sufficient to calculate $\lambda^2$ explicitly for all $c\in C$ and $m\in\{-1,0\}^n$ when $p>3$. The contribution of $i$-th element to $\lambda^2$ is $c_i^2/p$ or $(p-c_i)^2/p$.
Thus, the condition can be reduced by substituting
\begin{align}
    x_{a}\to x^{a^2}+x^{(p-a)^2}\,,
\end{align}
into the complete weight enumerator $W_C(\{x_a\})$.

\noindent \textbf{SUSY condition 1}: 
$W_C(\{x^{a^2}+x^{(p-a)^2}\})$ contains $x^{3p}$.

The coefficient of $x^{3p}$ represents the number of spin-$3/2$ primaries. Note that for $p=3$ we have to consider
\begin{equation}
    W_C(1+2x^9,x+x^4,x+x^4)\,,
\end{equation}
where the term $2x^9$ comes from $m_i = \pm1$.
Thus the corresponding CFT always contains spin-$3/2$ primaries from $(0,\cdots,0)\in C$.
For $p>3$, we do not have to include $x^m$ s.t. $m>3p$. For example, when $p=5$ it is enough to consider
\begin{equation}
    W_C(1, x^1, x^4+x^9,  x^4+x^9, x^1).
\end{equation}

The second condition requires the positive energy of any operator in the R sector.
From \eqref{eq:min_weight_in_R}, the second condition becomes

\noindent \textbf{SUSY condition 2}: 
\begin{equation}
    h_{\min} = \min_{c\,\in\, C} \frac{1}{2p} \sum_{i=1}^n \left( c_i-\frac{p}{2} \right)^2 \geq \frac{n}{24}.
\end{equation}

To reduce the third condition, we prepare a refined version of the complete weight enumerator $W_C(\{x_a\})$.

\begin{corollary}
Let $p$ be an odd prime and $C \subset \BF_p^n$ a self-dual code. Then
\begin{align}
\begin{aligned}
\label{eq:rr_part_refine}
    Z_{\widetilde{\mathrm{R}}}(\tau;\Lambda(C)) &= \frac{1}{\eta(\tau)^n} W_C \left( f^{1,1}_{0,p}(\tau),f^{1,1}_{1,p}(\tau),\cdots,f^{1,1}_{p-1,p}(\tau) \right)\,, \\
    &= \frac{1}{\eta(\tau)^n} RW_C \left( f^{1,1}_{1,p}(\tau),\cdots,f^{1,1}_{(p-1)/2,p}(\tau) \right)\,.
\end{aligned}
\end{align}
where $RW_C(\{x_a\})$ is the refined weight enumerator
\begin{align}
\begin{aligned}
\label{eq:def_refined_enum}
    RW_C(x_1,\cdots,x_{(p-1)/2}) :=&\; W_C(0,x_1,\cdots,x_{(p-1)/2},-x_{(p-1)/2},\cdots,-x_1)\,, \\
    =&\; \sum_{c\,\in\, C} \prod_{a=1}^{(p-1)/2} (-1)^{\mathrm{wt}_{p-a}(c)} x_a^{\mathrm{wt}_a(c)+\mathrm{wt}_{p-a}(c)} \,.
\end{aligned}
\end{align}
\label{cor:Z_RR_by_cwe}
\end{corollary}

\begin{proof}
The second equality in \eqref{eq:rr_part_refine} follows from
\begin{align}
\begin{aligned}
    f^{1,1}_{a,p}(\tau) &= \i^p\, \sum_{k\,\in\,\mathbb{Z}}\, (-1)^{(k+a)} q^{\frac{p}{2}\left(k+\frac{1}{2}+\frac{a}{p}\right)^2}\,, \\
    &= \i^p \,\sum_{k\,\in\,\mathbb{Z}}\, (-1)^{(-k+(p-a)-p)} q^{\frac{p}{2}\left(-k-\frac{3}{2}+\frac{p-a}{p}\right)^2}\,, \\
    &= \i^p \,(-1)^p\, \sum_{k'\,\in\,\mathbb{Z}}\, (-1)^{(k'+(p-a))} q^{\frac{p}{2}\left(k'+\frac{1}{2}+\frac{p-a}{p}\right)^2}\,, \\
    &= -f^{1,1}_{p-a,p}(\tau)\,.
\end{aligned}
\end{align}
where $k\rightarrow -k'-2$, and in particular $f^{1,1}_{0,p}(\tau)=0$.
\end{proof}

From Corollary \ref{cor:Z_RR_by_cwe}, the third condition becomes

\noindent \textbf{SUSY condition 3}:
\begin{equation}
    \frac{1}{\eta(\tau)^n} RW_C \left( f^{1,1}_{1,p}(\tau),\cdots,f^{1,1}_{(p-1)/2,p}(\tau) \right)
\end{equation}
is constant with respect to the torus modulus $\tau=\tau_1+\i\tau_2$.

We can check this by numerical computation or formulas for theta functions.
In some cases, $RW_C(x_1,\cdots,x_{(p-1)/2})=0$ before substituting the concrete form $f^{1,1}_{a,p}$.

Finally, we deal with the number of free fermions in a fermionic code CFT. As in the previous section, the number of free fermions is equal to that of norm 1 vectors in the Construction A lattice $\Lambda(C)$. Therefore, from a similar discussion for the first condition, the coefficient of $x^p$ in
\begin{equation}
	W_C(\{x^{a^2}+x^{(p-a)^2}\})
\end{equation}
is the number of free fermions.

\subsubsection{For $p=2$}
From the same discussion as when $p$ is an odd prime, the first condition becomes

\noindent \textbf{SUSY condition 1}: $W_C(1+2x^4,2x)$ contains $x^6$.

Equivalently, for $n\neq2$, the SUSY condition 1 requires the existence of a codeword $c\in C$ with $\mathrm{wt}_1(c)=2 \,\mathrm{or}\, 6$.
Note that, in the complete weight enumerator $W_C(1+2x^4,2x)$, the coefficient of $x^6$ counts the number of primary operators with the spin $s=\frac{3}{2}$.

From \eqref{eq:min_weight_in_R_p2}, the second condition becomes

\noindent \textbf{SUSY condition 2}: 
\begin{equation}
    \min_{s\,\in\, S(C)} \mathrm{wt}_1(s) \geq \frac{n}{6} \,.
\end{equation}

Regarding the third condition, to simplify the discussion, we define the complete weight enumerator corresponding to the $\widetilde{\mathrm{R}}$ sector by 
\begin{equation}
\label{eq:def_complete_weight_Rtilde}
    W_{\widetilde{\mathrm{R}}}(x_0,x_1):=W_{C_1}(x_0,x_1)-W_{C_3}(x_0,x_1) \,.
\end{equation}
From \eqref{eq:sector_partition_functions_p2}, the partition function of the $\widetilde{\mathrm{R}}$ sector can be written as
\begin{equation}
     Z_{\widetilde{\mathrm{R}}}(\tau;\Lambda(C)) = \frac{1}{\eta(\tau)^n} W_{\widetilde{\mathrm{R}}}(\theta_3(2\tau),\theta_2(2\tau)) \,.
\end{equation}
Then the third condition becomes

\noindent \textbf{SUSY condition 3}: 
\begin{equation}
    \frac{1}{\eta(\tau)^n} W_{\widetilde{\mathrm{R}}}(\theta_3(2\tau),\theta_2(2\tau))
\end{equation}
is constant with respect to $\tau$.

In some cases, $W_{\widetilde{\mathrm{R}}}(x_0,x_1)=0$ before substituting the Jacobi theta functions.

The number of free fermions is the coefficient of $x^2$ in $W_C(1,2x^1)$, which is equal to four times the number of codewords with $\mathrm{wt}_1(c)=2$.

\subsection{Examples}
In this section, we explore supersymmetric CFTs by applying the supersymmetry conditions to several examples of classical codes.
We summarize the properties of the self-dual codes related to the SUSY conditions in Table \ref{tab:susy_example}.
We demonstrate the SUSY conditions for six self-dual codes with various lengths and finite fields. The first column characterizes the type of self-dual code. The second column represents the number of spin-$\frac{3}{2}$ primary fields. The third column shows the minimum energy of operators in the R sector. The fourth column gives the RR partition function. Finally, the fifth column counts the number of free fermions.
Excluding fermionic CFTs with free fermions, two self-dual codes with $n=12$, $p=2$, and $n=36$, $p=2$ satisfy all the SUSY conditions.

\begin{table}[t]
    \centering
    \caption{The properties of self-dual codes associated with the SUSY conditions}
    \begin{tabular*}{14cm}{@{\extracolsep{\fill}}ccccc}
    \toprule
    $\;$ Self-dual code $C$ & spin-$\frac{3}{2}$ & $h_{\min}-\frac{n}{24}$ & $Z_{\widetilde{\mathrm{R}}}(\tau;\Lambda(C))$ & free fermions\\
    \midrule
    $\;$ $n=2\,,\,$ $p=5$ & $0$ & $\frac{1}{6}$ & $0$ & $4$\\[0.1cm]
    $\;$ $n=4\,,\,$ $p=3$ & $32$ & $\frac{1}{3}$ & $0$ & $8$\\[0.1cm]
    $\;$ $n=10\,,\,$ $p=5$ & $960$ & $-\frac{1}{6}$ & $0$ & $4$\\[0.1cm]
    $\;$ $n=12\,,\,$ $p=2$ & $2048$ & $0$ & $24$ & $0$\\[0.1cm]
    $\;$ $n=20\,,\,$ $p=5$ & $4608$ & $-\frac{1}{3}$ & $\Theta_{E_8}(q)/\eta(q)^8$ & $0$\\[0.1cm]
    $\;$ $n=36\,,\,$ $p=2$ & $1536$ & $0$ & $384$ & $0$\\
    \bottomrule
    \end{tabular*}
    \label{tab:susy_example}
\end{table}

\renewcommand{\arraystretch}{0.8}

\subsubsection{$n=2\,,\,$ $p=5$}
Let $C\subset\BF_5^2$ be a self-dual code generated by
\begin{equation}
    G = \begin{bmatrix} 1 & 2 \end{bmatrix}\,.
\end{equation}

As we have seen in section~\ref{sec:code_n2p5}, the complete weight enumerator of $C$ is
\begin{equation}
    W_C(x_0,x_1,x_2,x_3,x_4) = x_0^2+x_1x_2+x_2x_4+x_3x_1+x_4x_3 \,.
\end{equation}
Substituting $x^{a^2}+x^{(5-a)^2}$ for $x_a$, we obtain
\begin{equation}
    W_C(\{x^{a^2}+x^{(5-a)^2}\}) = 1 + 4x^5 + 4x^{10} + 4x^{20} + \cdots \,.
\end{equation}
Since there is no $x^{15}$ term, $C$ does not satisfy the condition 1.
Therefore, we can conclude that the fermionic CFT does not have supersymmetry.
In addition, the coefficient of $x^5$ is $4$, thus the corresponding CFT consists of 4 free fermions. This is consistent with the discussion about the partition functions in section~\ref{sec:code_n2p5}.

The minimum value of $\sum_{i=1}^2 \left( c_i-\frac{5}{2} \right)^2$ is $\frac{5}{2}$ for $c=(1,2)$ and then the SUSY condition 2 is satisfied:
\begin{equation}
    h_{\min} = \frac{1}{4} \geq \frac{2}{24}\,.
\end{equation}
Then, the resulting fermionic code CFT only contains operators with positive energy in the R sector.

For the self-dual code $C$, the refined weight enumerator is vanishing:
\begin{equation}
    RW_C(x_1,x_2) = W_C(0,x_1,x_2,-x_2,-x_1) = 0 \,.
\end{equation}
Then the RR partition function becomes $0$, and the SUSY condition 3 is immediately satisfied. The vanishing RR partition function follows from the inclusion of free fermions.

\subsubsection{$n=4\,,\,$ $p=3$}
Let $C\subset\BF_3^4$ be a self-dual code generated by
\begin{equation}
    G = \begin{bmatrix} 1 & 1 & 2 & 0 \\ 0 & 1 & 1 & 2 \end{bmatrix}\,.
\end{equation}
which is Example 2.2 in \cite{Gaiotto:2018ypj}.

The complete weight enumerator of $C$ is 
\begin{equation}
    W_C(x_0,x_1,x_2) = x_0^4 + x_0x_1^3 + 3x_0x_1^2x_2 + 3x_0x_1x_2^2 + x_0x_2^3 \,,
\end{equation}
which can be expanded as
\begin{equation}
    W_C(1+2x^9,x+x^4,x+x^4) = 1 + 8 x^3 + 24 x^6 + 32 x^9 + 24 x^{12} + \cdots \,.
\end{equation}
From the existence of the $x^9$ term, the fermionic code CFT contains 32 primaries with the spin $s=\frac{3}{2}$. Then the SUSY condition 1 is satisfied.
In addition, the CFT consists of 8 free fermions from the coefficient $8$ of the $x^3$ term, which agrees with the discussion in \cite{Gaiotto:2018ypj}.

The minimum value of $\sum_{i=1}^4 \left( c_i-\frac{3}{2} \right)^2$ is $3$ for $c=(1,1,2,0)$ and then the condition 2:
\begin{equation}
    h_{\min} = \frac{1}{2} \geq \frac{4}{24}
\end{equation}
is satisfied.

The refined weight enumerator is
\begin{equation}
    RW_C(x_1) = W_C(0,x_1,-x_1) = 0 \,,
\end{equation}
thus the condition 3 is immediately satisfied.

Therefore, the code $C$ satisfies all the SUSY conditions and indeed the CFT consisting of 8 free fermions is supersymmetric \cite{Goddard:1984hg}.

\subsubsection{$n=10\,,\,$ $p=5$}
Let $C\subset \BF_5^{10}$ be a self-dual code generated by
\begin{equation}
    G = \begin{bmatrix}
    1 & 0 & 0 & 4 & 3 & 0 & 0 & 2 & 4 & 2 \\
    0 & 1 & 0 & 2 & 3 & 0 & 0 & 4 & 3 & 4 \\
    0 & 0 & 1 & 2 & 1 & 0 & 0 & 2 & 4 & 2 \\
    0 & 0 & 0 & 0 & 0 & 1 & 0 & 4 & 2 & 2 \\
    0 & 0 & 0 & 0 & 0 & 0 & 1 & 3 & 3 & 1
    \end{bmatrix} \,,
\end{equation}
which is the bottom matrix for $\BF_5$ and $n=10$ in \cite{database}.

Since the complete weight enumerator consists of 165 terms, we do not write it here. If we substitute $x^{a^2}+x^{(5-a)^2}$ for $x_a$, it becomes
\begin{equation}
    W_C(\{x^{a^2}+x^{(5-a)^2}\}) = 1 + 4x^5 + 244x^{10} + 960x^{15} + 3120x^{20} + \cdots \,.
\end{equation}
Thus, the code $C$ satisfies the condition 1 and the corresponding CFT has 4 free fermions.

The minimum value of $\sum_{i=1}^{10} \left( c_i-\frac{5}{2} \right)^2$ is $\frac{5}{2}$ for $c=(2, 3, 2, 3, 2, 2, 3, 2, 3, 2)$ and then the condition 2 is not satisfied as
\begin{equation}
    h_{\min} = \frac{1}{4} < \frac{10}{24} \,.
\end{equation}
Therefore, the resulting fermionic code CFT does not satisfy the necessary condition for supersymmetry, which concludes that it is not supersymmetric.

The refined weight enumerator is
\begin{equation}
    RW_C(x_1,x_2) = W_C(0,x_1,x_2,-x_2,-x_1) = 0 \,,
\end{equation}
thus the condition 3 is immediately satisfied. This result is consistent with the existence of free fermions.

\subsubsection{$n=12\,,\,$ $p=2$}
Let $C\subset \BF_2^{12}$ be a self-dual code generated by
\begin{equation}
    G = \begin{bmatrix}
    1 & 0 & 0 & 0 & 1 & 0 & 0 & 0 & 1 & 1 & 1 & 1 \\
    0 & 1 & 0 & 0 & 1 & 0 & 0 & 1 & 1 & 1 & 1 & 0 \\
    0 & 0 & 1 & 0 & 1 & 0 & 0 & 1 & 0 & 0 & 0 & 1 \\
    0 & 0 & 0 & 1 & 0 & 0 & 0 & 1 & 1 & 0 & 0 & 1 \\
    0 & 0 & 0 & 0 & 0 & 1 & 0 & 1 & 0 & 1 & 0 & 1 \\
    0 & 0 & 0 & 0 & 0 & 0 & 1 & 1 & 0 & 0 & 1 & 1
    \end{bmatrix} \,,
\end{equation}
which is the bottom matrix for $\BF_2,\ n=12$ in \cite{database}.

The complete weight enumerator is
\begin{equation} \label{eq:cwe_n12p2}
    W_C(x_0,x_1) = x_0^{12} + 15x_0^8x_1^4 + 32x_0^6x_1^6 + 15x_0^4x_1^8 + x_1^{12} \,.
\end{equation}
Since
\begin{equation}
    W_C(1+2x^4,2x) = 1 + 264x^4 + 2048x^6 + 7944x^8 + \cdots \,,
\end{equation}
the condition 1 is satisfied from the existence of $x^6$ term and the CFT does not include free fermions from the absence of the $x^2$ term. We can reach the same conclusion based on whether or not the term with $x_1^2,x_1^6$ corresponding $\mathrm{wt}_1(c)=2,6$ exists in \eqref{eq:cwe_n12p2}.

The minimum weight of the shadow is $2$ for $s=(1,1,0,0,0,0,0,0,0,0,0,0)\in S(C)$ and the condition 2: $2\geq 12/6$ is satisfied.

The complete weight enumerator for the $\widetilde{\mathrm{R}}$ sector is
\begin{equation}
    W_{\widetilde{\mathrm{R}}}(x_0,x_1) = 6x_0^{10}x_1^2 - 12x_0^6x_1^6 + 6x_0^2x_1^{10} \,.
\end{equation}
It consists of the terms corresponding to $\mathrm{wt}_1(c)=2,6,10$, which reflects the fact that $\mathrm{wt}_1(\kappa)= \frac{n}{2}$ mod $4$ for $\kappa\in S(C)$ as discussed in \ref{sec:partiton_function_p2}. From the property of the Jacobi theta functions, the partition function of the $\widetilde{\mathrm{R}}$ sector is a nonzero constant as
\begin{equation}
    Z_{\widetilde{\mathrm{R}}}(\tau) = \frac{1}{\eta(\tau)^{12}} W_{\widetilde{\mathrm{R}}}(\theta_3(2\tau),\theta_2(2\tau)) = 24 \,.
\end{equation}
Therefore, the condition 3 is satisfied. This result is consistent with the fact that for a unitary supersymmetric CFT, $Z_{\widetilde{\mathrm{R}}}(\tau)$ can be a nonzero constant only if $h_{\min}=c/24$, i.e., the condition 2 is satisfied by equality.

Since there is only one self-dual lattice at $n=12$ that does not contain a vector $\lambda$ s.t. $\lambda^2=1$, which corresponds to a free fermion in the CFT, $\Lambda(C)$ is equivalent to the lattice in \ref{sec:code_n12p5}. Indeed, the partition function of the $\widetilde{\mathrm{R}}$ sector agrees with \eqref{eq:n=12p=5_partition}. (Although an overall factor $-1$ is multiplied, this is not significant because of the ambiguity in the $\widetilde{\mathrm{R}}$ sector as discussed in section \ref{ss:fermionCFT}.)

\subsubsection{$n=20\,,\,$ $p=5$}
\label{sss:n=20p=5}
Let $C\subset \BF_5^{20}$ be a self-dual code generated by
\begin{equation}
    G = \begin{bmatrix}
    2 & 2 & 2 & 2 & 2 & 2 & 2 & 2 & 2 & 2 & 2 & 2 & 2 & 2 & 2 & 2 & 2 & 2 & 2 & 2 \\
    0 & 1 & 0 & 0 & 0 & 0 & 0 & 0 & 0 & 0 & 3 & 0 & 1 & 4 & 1 & 1 & 2 & 1 & 1 & 0 \\
    0 & 0 & 1 & 0 & 0 & 0 & 0 & 0 & 0 & 0 & 4 & 1 & 1 & 0 & 1 & 3 & 0 & 1 & 2 & 1 \\
    0 & 0 & 0 & 1 & 0 & 0 & 0 & 0 & 0 & 0 & 4 & 0 & 3 & 2 & 2 & 1 & 0 & 0 & 4 & 3 \\
    0 & 0 & 0 & 0 & 1 & 0 & 0 & 0 & 0 & 0 & 4 & 4 & 2 & 2 & 2 & 0 & 2 & 4 & 3 & 1 \\
    0 & 0 & 0 & 0 & 0 & 1 & 0 & 0 & 0 & 0 & 3 & 3 & 4 & 2 & 0 & 0 & 1 & 2 & 0 & 4 \\
    0 & 0 & 0 & 0 & 0 & 0 & 1 & 0 & 0 & 0 & 1 & 1 & 0 & 1 & 3 & 2 & 1 & 0 & 1 & 4 \\
    0 & 0 & 0 & 0 & 0 & 0 & 0 & 1 & 0 & 0 & 4 & 2 & 0 & 4 & 1 & 3 & 1 & 3 & 3 & 3 \\
    0 & 0 & 0 & 0 & 0 & 0 & 0 & 0 & 1 & 0 & 3 & 3 & 3 & 2 & 3 & 0 & 2 & 2 & 1 & 0 \\
    0 & 0 & 0 & 0 & 0 & 0 & 0 & 0 & 0 & 1 & 1 & 0 & 2 & 1 & 2 & 4 & 2 & 3 & 3 & 1
    \end{bmatrix} \,,
\end{equation}
which we constructed as an appropriate example.

The complete weight enumerator substituted $x^{a^2}+x^{(5-a)^2}$ for $x_a$ is
\begin{equation}
    W_C(\{x^{a^2}+x^{(5-a)^2}\}) = 1 + 184x^{10} + 4608x^{15} + 68334x^{20} + \cdots \,,
\end{equation}
which shows the existence of spin-$\frac{3}{2}$ primary operators and the absence of primary operators with weight $h=\frac{1}{2}$.
Thus, the condition 1 is satisfied and there are no free fermions.

The second condition asks for the minimum energy of the R sector to be positive.
For the self-dual code $C$, the minimum value of $\sum_{i=1}^{20} \left( c_i-\frac{5}{2} \right)^2$ is $5$ for 
\begin{equation}
    c=(2,2,2,2,2,2,2,2,2,2,2,2,2,2,2,2,2,2,2,2) \in C\,,
\end{equation}
where the first row of the generator matrix guarantees the existence of the element in $C$.
In this case, the SUSY condition 2 is not satisfied:
\begin{equation}
    h_{\min} = \frac{1}{2} < \frac{20}{24} \,.
\end{equation}

The refined weight enumerator of $C$ is
\begin{equation}
    RW_C(x_1,x_2) = 4x_1^{20} + 912x_1^{15}x_2^5 + 1976x_1^{10}x_2^{10} - 912x_1^5x_2^{15} + 4x_2^{20}\,.
\end{equation}
We can easily verify the small $q$ expansion of the Ramond-Ramond partition function
\begin{align}
    \begin{aligned}
    \frac{1}{\eta(\tau)^{20}} RW_C \left( f^{1,1}_{1,5}(\tau),f^{1,1}_{2,5}(\tau) \right) &= 4q^{-\frac{1}{3}} + 992q^{\frac{2}{3}} + 16496q^{\frac{5}{3}} + 139008q^{\frac{8}{3}} + \cdots\,,\\
    &= 4(q^{-\frac{1}{3}} + 248 q^{\frac{2}{3}} + 4124 q^{\frac{5}{3}} + 34752 q^{\frac{8}{3}} + \cdots)\,.
    \end{aligned}
\end{align}
Therefore, the RR partition function is not constant, which means that the condition 3 is not satisfied. This is consistent with the fact that the partition function of the $\widetilde{\mathrm{R}}$ sector can be nonzero only if the CFT has no free fermions. In addition, the first term has a negative exponent, which corresponds to the fact that the condition 2 is not satisfied since the partition function is expressed as the sum of $q^{h-c/24}$.

Note that the above $q$ expansion exactly agrees with the one of $4\Theta_{E_8}(q)/\eta(q)^8$, where $\Theta_{E_8}(q)$ is the lattice theta function of the $E_8$ lattice.
Since the $E_8$ lattice is even self-dual, the corresponding theta function is invariant under the modular transformations up to phases. This is consistent with the modular property of the RR partition function. 
We further discuss fermionic code CFTs with non-constant RR partition functions in Appendix \ref{app:nonconstant}.
We examine several examples of non-constant RR partition functions and observe that they are related to particular modular forms called the Eisenstein series.

\subsubsection{$n=36\,,\,$ $p=2$}
Let $C\subset \BF_2^{36}$ be a self-dual code generated by
\begin{equation}
    G = \begin{bmatrix}
    1 & 0 & 0 & 0 & 0 & 0 & 0 & 0 & 0 & 0 & 0 & 0 & 0 & 0 & 0 & 0 & 0 & 0 & 0 & 0 & 1 & 1 & 1 & 0 & 1 & 1 & 0 & 0 & 0 & 0 & 0 & 1 & 0 & 0 & 0 & 1 \\
    0 & 1 & 0 & 0 & 0 & 0 & 0 & 0 & 0 & 0 & 0 & 0 & 0 & 0 & 0 & 0 & 0 & 0 & 0 & 0 & 1 & 1 & 1 & 0 & 1 & 1 & 0 & 0 & 0 & 0 & 0 & 1 & 1 & 1 & 1 & 0 \\
    0 & 0 & 1 & 0 & 0 & 0 & 0 & 0 & 0 & 0 & 0 & 0 & 0 & 1 & 0 & 0 & 0 & 0 & 0 & 0 & 1 & 1 & 0 & 1 & 0 & 1 & 1 & 0 & 0 & 0 & 1 & 0 & 1 & 1 & 1 & 1 \\
    0 & 0 & 0 & 1 & 0 & 0 & 0 & 0 & 0 & 0 & 0 & 0 & 0 & 1 & 0 & 0 & 0 & 0 & 0 & 0 & 1 & 1 & 0 & 1 & 0 & 1 & 1 & 0 & 0 & 0 & 0 & 1 & 1 & 1 & 0 & 0 \\
    0 & 0 & 0 & 0 & 1 & 0 & 0 & 0 & 0 & 0 & 0 & 0 & 0 & 1 & 0 & 0 & 0 & 1 & 0 & 0 & 1 & 1 & 1 & 0 & 1 & 0 & 1 & 1 & 0 & 0 & 0 & 1 & 0 & 1 & 0 & 1 \\
    0 & 0 & 0 & 0 & 0 & 1 & 0 & 0 & 0 & 0 & 0 & 0 & 0 & 1 & 0 & 0 & 0 & 1 & 0 & 0 & 1 & 1 & 1 & 0 & 1 & 0 & 1 & 1 & 1 & 1 & 0 & 1 & 1 & 0 & 0 & 1 \\
    0 & 0 & 0 & 0 & 0 & 0 & 1 & 0 & 0 & 0 & 0 & 0 & 0 & 1 & 0 & 0 & 0 & 0 & 0 & 0 & 0 & 1 & 0 & 0 & 0 & 0 & 0 & 1 & 0 & 0 & 0 & 0 & 0 & 0 & 1 & 1 \\
    0 & 0 & 0 & 0 & 0 & 0 & 0 & 1 & 0 & 0 & 0 & 0 & 0 & 1 & 0 & 0 & 0 & 0 & 0 & 0 & 0 & 1 & 0 & 0 & 0 & 0 & 1 & 0 & 0 & 0 & 1 & 1 & 0 & 0 & 0 & 0 \\
    0 & 0 & 0 & 0 & 0 & 0 & 0 & 0 & 1 & 0 & 0 & 0 & 0 & 0 & 0 & 0 & 0 & 0 & 0 & 0 & 1 & 1 & 1 & 1 & 0 & 1 & 0 & 0 & 0 & 1 & 0 & 1 & 0 & 1 & 0 & 1 \\
    0 & 0 & 0 & 0 & 0 & 0 & 0 & 0 & 0 & 1 & 0 & 0 & 0 & 0 & 0 & 0 & 0 & 0 & 0 & 0 & 1 & 1 & 1 & 1 & 1 & 0 & 0 & 0 & 1 & 0 & 1 & 0 & 1 & 0 & 1 & 0 \\
    0 & 0 & 0 & 0 & 0 & 0 & 0 & 0 & 0 & 0 & 1 & 0 & 0 & 1 & 0 & 0 & 0 & 0 & 0 & 0 & 1 & 0 & 1 & 0 & 1 & 1 & 0 & 0 & 0 & 1 & 1 & 0 & 1 & 0 & 1 & 0 \\
    0 & 0 & 0 & 0 & 0 & 0 & 0 & 0 & 0 & 0 & 0 & 1 & 0 & 1 & 0 & 0 & 0 & 0 & 0 & 0 & 1 & 0 & 0 & 1 & 1 & 1 & 0 & 0 & 1 & 0 & 0 & 1 & 1 & 0 & 1 & 0 \\
    0 & 0 & 0 & 0 & 0 & 0 & 0 & 0 & 0 & 0 & 0 & 0 & 1 & 1 & 0 & 0 & 0 & 0 & 0 & 0 & 0 & 0 & 1 & 1 & 0 & 0 & 1 & 1 & 0 & 0 & 1 & 1 & 1 & 1 & 0 & 0 \\
    0 & 0 & 0 & 0 & 0 & 0 & 0 & 0 & 0 & 0 & 0 & 0 & 0 & 0 & 1 & 0 & 0 & 1 & 0 & 0 & 0 & 0 & 1 & 1 & 0 & 0 & 1 & 1 & 1 & 0 & 0 & 1 & 0 & 1 & 1 & 0 \\
    0 & 0 & 0 & 0 & 0 & 0 & 0 & 0 & 0 & 0 & 0 & 0 & 0 & 0 & 0 & 1 & 0 & 1 & 0 & 0 & 1 & 1 & 1 & 1 & 1 & 1 & 0 & 0 & 1 & 0 & 0 & 1 & 1 & 0 & 1 & 0 \\
    0 & 0 & 0 & 0 & 0 & 0 & 0 & 0 & 0 & 0 & 0 & 0 & 0 & 0 & 0 & 0 & 1 & 1 & 0 & 0 & 1 & 1 & 0 & 0 & 1 & 1 & 1 & 1 & 1 & 1 & 0 & 0 & 0 & 0 & 0 & 0 \\
    0 & 0 & 0 & 0 & 0 & 0 & 0 & 0 & 0 & 0 & 0 & 0 & 0 & 0 & 0 & 0 & 0 & 0 & 1 & 0 & 0 & 1 & 1 & 1 & 1 & 1 & 0 & 0 & 0 & 1 & 1 & 0 & 0 & 1 & 1 & 0 \\
    0 & 0 & 0 & 0 & 0 & 0 & 0 & 0 & 0 & 0 & 0 & 0 & 0 & 0 & 0 & 0 & 0 & 0 & 0 & 1 & 1 & 0 & 0 & 0 & 1 & 1 & 1 & 1 & 1 & 0 & 1 & 0 & 0 & 1 & 0 & 1
    \end{bmatrix}
\end{equation}
which is the top matrix for $\BF_2,\ n=36,\ d=6$ in \cite{database}.

The complete weight enumerator of $C$ is
\begin{equation}
\begin{aligned}
    W_C(x_0,x_1) \;= x_0^{36}& + 24x_0^{30}x_1^6 + 225x_0^{28}x_1^8 + 1872x_0^{26}x_1^{10} + 9555 x_0^{24}x_1^{12} + 29160x_0^{22}x_1^{14} \\
    &+ 55755x_0^{20}x_1^{16} + 68960x_0^{18}x_1^{18} + 55755x_0^{16}x_1^{20} + 29160x_0^{14}x_1^{22} \\ &+ 9555x_0^{12}x_1^{24} 
    + 1872x_0^{10}x_1^{26} + 225x_0^8x_1^{28} + 24x_0^6x_1^{30} + x_1^{36} \,.
\end{aligned}
\end{equation}

The complete weight enumerator gives the expansion
\begin{equation}
    W_C(1+2x^4,2x) = 1 + 72 x^4 + 1536 x^6 + 60120 x^8 + \cdots \,.
\end{equation}
Then, the SUSY condition 1 is satisfied from the existence of the $x^6$ term and the fermionic code CFT does not contain free fermions because of the absence of the $x^2$ term.

The minimum weight of the shadow is $6$ for an element $s\in S(C)$:
\begin{equation}
    s=(0,1,0,0,0,1,1,0,0,0,0,1,0,0,0,0,0,0,0,0,0,0,0,0,1,0,0,0,0,0,0,1,0,0,0,0)\,.
\end{equation}
Therefore, the SUSY condition 2 is satisfied: $h_{\min} = 6\geq36/6$.

The complete weight enumerator for the $\widetilde{\mathrm{R}}$ sector is
\begin{equation}
    W_{\widetilde{\mathrm{R}}}(x_0,x_1) = 6x_0^{30}x_1^6 - 36x_0^{26}x_1^{10} + 90x_0^{22}x_1^{14} - 120x_0^{18}x_1^{18} + 90x_0^{14}x_1^{22} - 36x_0^{10}x_1^{26} + 6x_0^{6}x_1^{30}\,.
\end{equation}
The partition function of the $\widetilde{\mathrm{R}}$ sector is
\begin{equation}
    Z_{\widetilde{\mathrm{R}}}(\tau) = \frac{1}{\eta(\tau)^{36}}\, W_{\widetilde{\mathrm{R}}}(\theta_3(2\tau),\theta_2(2\tau)) = 384 \,,
\end{equation}
which means that the condition 3 is satisfied.
Therefore, the self-dual code $C$ satisfies all the SUSY conditions. Moreover, it does not contain free fermions. This strongly suggests the existence of supersymmetry in the fermionic code CFT.
To clarify the spectrum of the CFT, the partition functions of the other sectors can be calculated similarly from \eqref{eq:sector_partition_functions_p2} as
\begin{align}
\begin{aligned}
    Z_{\mathrm{NS}}(\tau) &= q^{-\frac{3}{2}} + 108q^{-\frac{1}{2}} + 1536 + 63414q^{\frac{1}{2}} + 2064384 q + \cdots \,, \\
    Z_{\widetilde{\mathrm{NS}}}(\tau) &= q^{-\frac{3}{2}} + 108q^{-\frac{1}{2}} - 1536 + 63414q^{\frac{1}{2}} - 2064384q + \cdots \,, \\
   Z_{\mathrm{R}}(\tau) &=  1152 + 4128768 q + 1307049984 q^2 + 127940165632 q^3 + \cdots \,.
 \end{aligned}
\end{align}

\section{Discussion}
\label{sec:discussion}

We have constructed chiral fermionic CFTs from linear self-dual codes over finite fields $\BF_p$ where $p$ is a prime number including $p=2$. The key ingredient of our construction is the relationship between self-dual codes and odd self-dual lattices via Construction A.
This relationship can be generalized to classical codes over $\BZ_k$, which is the ring of integers modulo $k$.
It is known that Construction A can be applied to self-dual codes over $\BZ_k$ and endows with self-dual lattices \cite{bannai1999type}.
Thus, we expect the generalization of our construction into such a class of classical codes.

In section \ref{sss:n=20p=5}, we found an example of fermionic CFTs with a non-constant RR partition function. Notably, the RR partition function is associated with the unique $8$-dimensional even self-dual lattice $E_8$. From the constraint of the modular property, a non-constant RR partition function should be written in a modular form.
In Appendix \ref{app:nonconstant}, we discuss non-constant RR partition functions focusing on binary self-dual codes. From several examples of self-dual codes, we observe that the RR partition functions are related to particular modular forms called the Eisenstein series.
Moreover, some examples show the RR partition functions written by the theta functions of $8$- and $16$-dimensional even self-dual lattices.
As suggested in Question 0.2 of \cite{Gaiotto:2018ypj}, it could be interesting to find a systematic construction of fermionic CFTs with the RR partition function related to even self-dual lattices.
    
This paper has mainly studied the torus partition functions with four sectors $\mathrm{NS}$, $\widetilde{\mathrm{NS}}$, $\mathrm{R}$, and $\widetilde{\mathrm{R}}$ depending on the choice of spin structures on the torus.
More generally, we can consider the higher-genus partition functions corresponding to fermionic CFTs living on higher-genus Riemann surfaces.
The higher-genus partition functions have been discussed for bosonic CFTs constructed from classical binary codes in \cite{Henriksson:2021qkt}. It tells us that higher-genus partition functions can be expressed as the higher-genus weight enumerator \cite{macwilliams1972generalizations}, which is a natural generalization of the complete weight enumerator $W_C(\{x_a\})$.
It could be extended to our fermionic CFTs constructed from classical $p$-ary codes. More recently, the average of higher-genus weight enumerators over $p$-ary self-dual codes has been computed in \cite{Kawabata:2022jxt}, which may enable us to study the averaged higher-genus partition functions of our fermionic CFTs.

Our fermionic CFTs have been defined by a set of vertex operators associated with the Construction A lattice.
In section \ref{sec:SUSY}, we have examined the three necessary conditions for fermionic CFTs to be supersymmetric.
To show that they are in fact supersymmetric, we have to give a supercurrent satisfying the OPEs \eqref{eq:scft_ope} from vertex operators. Recently, it has been pointed out that a fundamental property of quantum error-correcting codes can be used to construct supercurrents from vertex operators \cite{Harvey:2020jvu}.
If the construction of supercurrents can be generalized to our theories, it can be helpful to show the existence of supersymmetry.

\bigskip
\acknowledgments
We are grateful to T.\,Nishioka, K.\,Ohmori, and T.\,Okuda for valuable discussions.
The work of K.\,K. and S.\,Y. is supported by Forefront Physics and Mathematics Program to Drive Transformation (FoPM), a World-leading Innovative Graduate Study (WINGS) Program, the University of Tokyo. The work of K.\,K. is also supported by JSPS
KAKENHI Grant-in-Aid for JSPS fellows Grant No. 	23KJ0436.

\appendix

\section{List of notations}
\label{sec:list}

\renewcommand{\arraystretch}{1.1}

\begin{center}
\begin{longtable}{c p{11cm} l}
    \toprule
    Symbol & Definition & See \\
    \midrule
    $p$ & prime number & \\
    $\BF_q$ & finite field of order $q$ & \\
    $C$ & linear code over $\BF_p$ & \\
    $n$ & length of $C$ & \\
    $c$ & codeword $c\in C$ & \\
    $G$ & generator matrix of a linear code & \\
    $H$ & parity check matrix of a linear code & \\
    $C^\perp$ & dual code of $C$ & \\
    $d(C)$ & minimum distance of $C$ & Eq.\eqref{def:minimum_dis} \\
    $\mathrm{wt}(c)$ & Hamming weight of $c$ & Eq.\eqref{eq:def_hamming_wt} \\
    $\mathrm{wt}_a(c)$ & number of components of $c$ that are equal to $a$ & Eq.\eqref{eq:composition} \\
    $\mathrm{Lee}(c)$ & Lee weight of $c$ & Eq.\eqref{eq:def_Lee_wt} \\
    $\mathrm{Norm}(c)$ & Euclidean norm of $c$ & Eq.\eqref{eq:def_norm_c} \\
    $W_C$ & complete weight enumerator of $C$ & Eq.\eqref{eq:def_complete_weight}\\
    $\Lambda(C)$ & Construction A lattice from a linear code $C$ & Eq.\eqref{eq:def_constA_lattice}\\
    $\Lambda^*$ & dual lattice of $\Lambda$ & Eq.\eqref{eq:dual_lattice}\\
    $\chi$ & characteristic vector & Eq.\eqref{eq:def_characteristic}\\
    $\Lambda_0$, $\Lambda_2$ & subset of $\Lambda(C)$ consisting of elements with even (odd) norm & Eq.\eqref{eq:lattice_NS}\\
    $\Lambda_1$, $\Lambda_3$ & shift of $\Lambda_0$ and $\Lambda_2$ by $\chi/2$ & Eq.\eqref{eq:def_lambda1_lambda3}\\
    $S(\Lambda)$ & shadow of a lattice $\Lambda$ & \\
    $Z_*(\tau)$ & partition functions of each sector ($*=\mathrm{NS},\,\widetilde{\mathrm{NS}},\,\mathrm{R},\,\widetilde{\mathrm{R}}$ in \eqref{eq:def_sector_label}) & Eq.\eqref{eq:def_partition_function_ff}\\
    $Z^{\alpha,\beta}(\tau)$ & partition functions associated with spin structures $(\alpha,\beta)$ & Eq.\eqref{eq:collective_part}\\
    $f_{a,p}^{\alpha,\beta}(\tau)$ & functions of $\tau$ depending on spin structures $(\alpha,\beta)$ ($a\in\BF_p$) & Eq.\eqref{eq:def_f_function}\\
    $C_0$, $C_2$ & subset of $C$ consisting of doubly-even (singly-even) codewords & Eq.\eqref{eq:def_C0_C2}\\
    $S(C)$ & shadow of a code $C$ & Eq.\eqref{eq:def_shadow_code}\\
    $s$ & element of $S(C)$ & \\
    $C_1$, $C_3$ & shift of $C_0$ and $C_2$ by $s\in S(C)$ & Eq.\eqref{eq:shadowcode_part}\\
    $\Delta_{\mathrm{NS}}$ & spectral gap of the NS sector & \\
    $h_{\min}$ & minimum conformal weight of the R sector & \\
    $RW_C$ & refined weight enumerator of $C$ & Eq.\eqref{eq:def_refined_enum}\\
    $W_{\widetilde{\mathrm{R}}}$ & complete weight enumerator corresponding to the $\widetilde{\mathrm{R}}$ sector & Eq.\eqref{eq:def_complete_weight_Rtilde}\\[0.3cm]
    \bottomrule
\end{longtable}
\end{center}

\section{Non-constant Ramond-Ramond partition functions}
\label{app:nonconstant}

\renewcommand{\arraystretch}{0.75}

In section \ref{sec:SUSY}, we explored fermionic CFTs with constant RR partition functions. While most of the examples gave rise to 0 or a constant, we encountered the non-constant RR partition function written by the theta function of the $E_8$ lattice in section \ref{sss:n=20p=5}.
In this appendix, we further discuss non-constant RR partition functions of fermionic code CFTs and clearly express them in modular forms. In some cases, we find the RR partition function associated with even self-dual lattices. We focus on binary singly-even self-dual codes in \cite{database}.

Let $C\subset \BF_2^n$ be a singly-even self-dual code. We define $l_{\min}\in\BN$ as the minimum value of $l$ such that $|\{s\in C_1 \mid \mathrm{wt}_1(s)=l \}| \neq |\{s\in C_3 \mid \mathrm{wt}_1(s)=l \}|$. 
If no such $l$ exists, $W_{\widetilde{\mathrm{R}}}(x_0,x_1)=0$ and thus $Z_{\widetilde{\mathrm{R}}}(\tau)=0$. In the following, we assume that such $l$ exists. Then, $l_{\min}$ is equal to the minimum integer such that $x_0^{n-l}x_1^l$ is included in $W_{\widetilde{\mathrm{R}}}(x_0,x_1)$. In addition, the lowest degree term in $Z_{\widetilde{\mathrm{R}}}(\tau)$ is $q^h$ where $h=\frac{1}{4}l_{\min}-\frac{n}{24}$. 

Let us introduce
\begin{equation}
    \Theta_{\widetilde{\mathrm{R}}}(\tau) := \eta(\tau)^{-24h}\, Z_{\widetilde{\mathrm{R}}}(\tau) \,,
\end{equation}
which satisfies $\Theta_{\widetilde{\mathrm{R}}}(\tau)=\CO(q^0)$ when we take $q\to0$ since $\eta(\tau)=\CO(q^{\frac{1}{24}})$. 
From \eqref{T_transformation_p2} and \eqref{S_transformation_p2}, the property of the Dedekind eta $\eta(\tau)$ and the fact that $\mathrm{wt}_1(s)=\frac{n}{2} \mod 4$ for $s\in S(C)$, the modular transformation law is
\begin{align}
    \Theta_{\widetilde{\mathrm{R}}}(\tau+1) &= \exp\left[i2\pi\left(-\frac{l_{\min}}{4}+\frac{n}{24}+\frac{n}{12}\right)\right]  \Theta_{\widetilde{\mathrm{R}}}(\tau) = \Theta_{\widetilde{\mathrm{R}}}(\tau) \,, \\
    \Theta_{\widetilde{\mathrm{R}}}(-1/\tau) &= \tau^{-12h} (-i)^{-3 l_{\min}+\frac{n}{2}} (-1)^{\frac{n}{2}} \Theta_{\widetilde{\mathrm{R}}}(\tau) = \tau^{-12h}\, \Theta_{\widetilde{\mathrm{R}}}(\tau) \,.
\end{align}
Thus, $\Theta_{\widetilde{\mathrm{R}}}(\tau)$ is a modular form of weight $-12h=\frac{n}{2}-3l_{\min}$.

Since the dimensions of spaces of modular forms of weight $4,6,8,10,14$ are $1$, if $-12h$ is such a value, then the RR partition function must be a constant multiple of $E_{-12h}(\tau)$ where $E_{2k}(\tau)$ is the Eisenstein series of weight $2k$: ($\zeta(s)$: the Riemann zeta function)
\begin{align}
    E_{2k}(\tau) = \frac{1}{2\zeta(2k)}\sum_{(m,n)\,\in\,\BZ^2\backslash(0,0)} \frac{1}{(m+n\tau)^{2k}}\,.
\end{align}

We can directly calculate that the RR partition functions of the binary codes with length $n<18$ or $n=22,24$ are all $0$ or constant. Those of codes with $p=5$, $n\leq 16$ or $p=7$, $n\leq 12$ are also $0$ or constant. Several examples of the non-constant RR partition functions are presented below.

\subsection*{$n=18\,,\,$ $p=2$}
A self-dual code generated by the bottom matrix for $\BF_2, n=18$ in \cite{database}:
\begin{equation}
    G = \begin{bmatrix}
    1 & 0 & 0 & 0 & 0 & 0 & 0 & 1 & 1 & 0 & 0 & 0 & 0 & 0 & 0 & 1 & 1 & 1 \\
    0 & 1 & 0 & 0 & 0 & 0 & 0 & 0 & 0 & 0 & 0 & 0 & 0 & 1 & 0 & 1 & 1 & 0 \\
    0 & 0 & 1 & 0 & 0 & 0 & 0 & 1 & 1 & 0 & 0 & 0 & 0 & 0 & 1 & 0 & 0 & 0 \\
    0 & 0 & 0 & 1 & 0 & 0 & 0 & 1 & 0 & 1 & 1 & 0 & 0 & 0 & 1 & 1 & 1 & 1 \\
    0 & 0 & 0 & 0 & 1 & 0 & 0 & 0 & 1 & 1 & 1 & 0 & 0 & 0 & 1 & 1 & 1 & 1 \\
    0 & 0 & 0 & 0 & 0 & 1 & 0 & 1 & 1 & 0 & 1 & 0 & 0 & 0 & 0 & 0 & 0 & 0 \\
    0 & 0 & 0 & 0 & 0 & 0 & 1 & 1 & 1 & 1 & 0 & 0 & 0 & 0 & 0 & 0 & 0 & 0 \\
    0 & 0 & 0 & 0 & 0 & 0 & 0 & 0 & 0 & 0 & 0 & 1 & 0 & 1 & 0 & 1 & 0 & 1 \\
    0 & 0 & 0 & 0 & 0 & 0 & 0 & 0 & 0 & 0 & 0 & 0 & 1 & 1 & 0 & 0 & 1 & 1
    \end{bmatrix}\,.
\end{equation}

The complete weight enumerator for the $\widetilde{\mathrm{R}}$ sector is
\begin{equation}
    W_{\widetilde{\mathrm{R}}}(x_0,x_1) = x_0^{17} x_1 - 34 x_0^{13} x_1^5 + 34 x_0^5 x_1^{13} - x_0 x_1^{17}\,.
\end{equation}

The partition function of the $\widetilde{\mathrm{R}}$ sector is
\begin{equation}
    Z_{\widetilde{\mathrm{R}}}(\tau) = \frac{1}{\eta(\tau)^{18}} W_{\widetilde{\mathrm{R}}}(\theta_3(2\tau),\theta_2(2\tau))
    = \frac{2}{\eta(\tau)^{12}} E_6(\tau)\,.
\end{equation}

This is the only code with the non-constant RR partition function at $n=18,\, p=2$.

\subsection*{$n=20\,,\,$ $p=2$}
The bottom matrix for $\BF_2, n=20$ in \cite{database}:
\begin{equation}
    G = \begin{bmatrix}
    1 & 0 & 0 & 0 & 0 & 0 & 0 & 0 & 0 & 1 & 1 & 0 & 0 & 0 & 0 & 0 & 0 & 1 & 1 & 1 \\
    0 & 1 & 0 & 0 & 0 & 0 & 0 & 0 & 0 & 0 & 1 & 1 & 1 & 0 & 0 & 0 & 1 & 1 & 1 & 1 \\
    0 & 0 & 1 & 0 & 0 & 1 & 0 & 0 & 0 & 1 & 1 & 0 & 0 & 0 & 0 & 1 & 1 & 1 & 1 & 0 \\
    0 & 0 & 0 & 1 & 0 & 1 & 0 & 0 & 0 & 1 & 0 & 1 & 1 & 0 & 0 & 1 & 1 & 0 & 0 & 1 \\
    0 & 0 & 0 & 0 & 1 & 1 & 0 & 0 & 0 & 0 & 1 & 1 & 1 & 0 & 0 & 0 & 0 & 1 & 1 & 1 \\
    0 & 0 & 0 & 0 & 0 & 0 & 1 & 0 & 0 & 0 & 1 & 1 & 1 & 0 & 0 & 1 & 0 & 0 & 0 & 1 \\
    0 & 0 & 0 & 0 & 0 & 0 & 0 & 1 & 0 & 1 & 1 & 0 & 1 & 0 & 0 & 0 & 0 & 0 & 0 & 0 \\
    0 & 0 & 0 & 0 & 0 & 0 & 0 & 0 & 1 & 1 & 1 & 1 & 0 & 0 & 0 & 0 & 0 & 0 & 0 & 0 \\
    0 & 0 & 0 & 0 & 0 & 0 & 0 & 0 & 0 & 0 & 0 & 0 & 0 & 1 & 0 & 1 & 0 & 1 & 0 & 1 \\
    0 & 0 & 0 & 0 & 0 & 0 & 0 & 0 & 0 & 0 & 0 & 0 & 0 & 0 & 1 & 1 & 0 & 0 & 1 & 1
    \end{bmatrix}\,.
\end{equation}

The complete weight enumerator for the $\widetilde{\mathrm{R}}$ sector is
\begin{equation}
    W_{\widetilde{\mathrm{R}}}(x_0,x_1) = x_0^{18} x_1^2 + 12 x_0^{14} x_1^6 - 26 x_0^{10} x_1^{10} + 12 x_0^6 x_1^{14} + 
 x_0^2 x_1^{18}\,.
\end{equation}

The partition function of the $\widetilde{\mathrm{R}}$ sector is
\begin{equation}
\label{eq:nonconst_n20}
    Z_{\widetilde{\mathrm{R}}}(\tau) = \frac{1}{\eta(\tau)^{20}} W_{\widetilde{\mathrm{R}}}(\theta_3(2\tau),\theta_2(2\tau))
    = \frac{4}{\eta(\tau)^{8}} E_4(\tau)\,.
\end{equation}
Note that the Eisenstein series $E_4(\tau)$ can be written as the lattice theta function of the $E_8$ lattice, which is the unique even self-dual lattice in 8 dimensions.
The RR partition functions of all other codes at $n=20,\, p=2$ are $0$ or $1,2,3,4,6,10$ times \eqref{eq:nonconst_n20}.

\subsection*{$n=26\,,\,$ $p=2$}
The seventh matrix from the bottom for $\BF_2, n=26$ in \cite{database}:
\begin{equation}
    G = \begin{bmatrix}
    1 & 0 & 0 & 0 & 0 & 0 & 0 & 0 & 0 & 0 & 0 & 0 & 0 & 0 & 1 & 1 & 1 & 1 & 1 & 0 & 0 & 1 & 0 & 0 & 0 & 1 \\
    0 & 1 & 0 & 0 & 0 & 0 & 0 & 0 & 0 & 0 & 0 & 0 & 0 & 0 & 0 & 1 & 1 & 0 & 1 & 1 & 0 & 1 & 0 & 1 & 1 & 0 \\
    0 & 0 & 1 & 0 & 0 & 0 & 0 & 0 & 0 & 0 & 0 & 0 & 0 & 0 & 0 & 1 & 0 & 0 & 1 & 0 & 1 & 0 & 0 & 0 & 0 & 0 \\
    0 & 0 & 0 & 1 & 0 & 0 & 0 & 0 & 0 & 0 & 0 & 0 & 0 & 0 & 1 & 1 & 1 & 1 & 1 & 0 & 0 & 0 & 0 & 1 & 0 & 1 \\
    0 & 0 & 0 & 0 & 1 & 0 & 0 & 0 & 0 & 0 & 0 & 0 & 0 & 0 & 1 & 1 & 0 & 0 & 0 & 0 & 1 & 1 & 0 & 1 & 1 & 1 \\
    0 & 0 & 0 & 0 & 0 & 1 & 0 & 0 & 0 & 0 & 0 & 0 & 0 & 0 & 1 & 0 & 1 & 0 & 0 & 1 & 0 & 1 & 1 & 1 & 0 & 1 \\
    0 & 0 & 0 & 0 & 0 & 0 & 1 & 0 & 0 & 0 & 0 & 0 & 1 & 0 & 0 & 0 & 1 & 0 & 1 & 0 & 1 & 1 & 0 & 1 & 0 & 1 \\
    0 & 0 & 0 & 0 & 0 & 0 & 0 & 1 & 0 & 0 & 0 & 0 & 1 & 0 & 1 & 0 & 0 & 1 & 1 & 0 & 1 & 0 & 0 & 0 & 1 & 1 \\
    0 & 0 & 0 & 0 & 0 & 0 & 0 & 0 & 1 & 0 & 0 & 0 & 1 & 0 & 0 & 1 & 0 & 0 & 0 & 1 & 1 & 1 & 1 & 1 & 0 & 0 \\
    0 & 0 & 0 & 0 & 0 & 0 & 0 & 0 & 0 & 1 & 0 & 0 & 1 & 0 & 1 & 1 & 1 & 0 & 0 & 0 & 1 & 0 & 1 & 0 & 0 & 1 \\
    0 & 0 & 0 & 0 & 0 & 0 & 0 & 0 & 0 & 0 & 1 & 0 & 1 & 0 & 0 & 0 & 0 & 1 & 0 & 1 & 0 & 0 & 0 & 0 & 1 & 1 \\
    0 & 0 & 0 & 0 & 0 & 0 & 0 & 0 & 0 & 0 & 0 & 1 & 0 & 0 & 1 & 0 & 0 & 1 & 1 & 1 & 1 & 1 & 0 & 1 & 0 & 0 \\
    0 & 0 & 0 & 0 & 0 & 0 & 0 & 0 & 0 & 0 & 0 & 0 & 0 & 1 & 0 & 0 & 1 & 0 & 0 & 0 & 0 & 0 & 0 & 0 & 1 & 1
    \end{bmatrix}\,.
\end{equation}

The complete weight enumerator for the $\widetilde{\mathrm{R}}$ sector is
\begin{equation}
    W_{\widetilde{\mathrm{R}}}(x_0,x_1) = x_0^{25} x_1^1 - 20 x_0^{21} x_1^5 - 475 x_0^{17} x_1^9 + 475 x_0^9 x_1^{17} + 20 x_0^5 x_1^{21} - x_0 ^1 x_1^{25}\,.
\end{equation}

The partition function of the $\widetilde{\mathrm{R}}$ sector is
\begin{equation}
    Z_{\widetilde{\mathrm{R}}}(\tau) = \frac{1}{\eta(\tau)^{26}} W_{\widetilde{\mathrm{R}}}(\theta_3(2\tau),\theta_2(2\tau))
    = \frac{2}{\eta(\tau)^{20}} E_{10}(\tau)\,.
\end{equation}

The RR partition functions of all other codes at $n=26,\, p=2$ are $0$ or this.

\subsection*{$n=28\,,\,$ $p=2$}
The second matrix from the bottom for $\BF_2, n=28$ in \cite{database}:
\begin{equation}
    G = \begin{bmatrix}
    1 & 0 & 0 & 0 & 0 & 0 & 0 & 0 & 0 & 0 & 0 & 0 & 0 & 0 & 0 & 1 & 1 & 1 & 1 & 1 & 1 & 1 & 1 & 1 & 1 & 1 & 1 & 1 \\
    0 & 1 & 0 & 0 & 0 & 0 & 0 & 0 & 0 & 0 & 0 & 0 & 0 & 0 & 1 & 0 & 0 & 0 & 1 & 0 & 0 & 0 & 0 & 0 & 1 & 0 & 1 & 1 \\
    0 & 0 & 1 & 0 & 0 & 0 & 0 & 0 & 0 & 0 & 0 & 0 & 0 & 0 & 1 & 1 & 0 & 0 & 1 & 1 & 0 & 0 & 0 & 0 & 1 & 1 & 1 & 0 \\
    0 & 0 & 0 & 1 & 0 & 0 & 0 & 0 & 0 & 0 & 0 & 0 & 0 & 0 & 1 & 0 & 0 & 1 & 0 & 1 & 1 & 1 & 1 & 1 & 0 & 1 & 1 & 0 \\
    0 & 0 & 0 & 0 & 1 & 0 & 0 & 0 & 0 & 0 & 0 & 0 & 0 & 0 & 1 & 0 & 1 & 1 & 1 & 0 & 0 & 1 & 0 & 0 & 1 & 0 & 1 & 0 \\
    0 & 0 & 0 & 0 & 0 & 1 & 0 & 0 & 0 & 0 & 0 & 0 & 0 & 0 & 1 & 1 & 0 & 1 & 0 & 0 & 0 & 0 & 1 & 0 & 1 & 0 & 1 & 1 \\
    0 & 0 & 0 & 0 & 0 & 0 & 1 & 0 & 0 & 0 & 0 & 0 & 0 & 0 & 1 & 0 & 1 & 0 & 0 & 1 & 0 & 0 & 0 & 1 & 1 & 0 & 1 & 1 \\
    0 & 0 & 0 & 0 & 0 & 0 & 0 & 1 & 0 & 0 & 0 & 0 & 0 & 0 & 1 & 1 & 1 & 0 & 0 & 0 & 1 & 1 & 1 & 1 & 1 & 1 & 0 & 0 \\
    0 & 0 & 0 & 0 & 0 & 0 & 0 & 0 & 1 & 0 & 0 & 0 & 0 & 0 & 1 & 1 & 1 & 1 & 1 & 1 & 1 & 0 & 1 & 1 & 0 & 0 & 0 & 0 \\
    0 & 0 & 0 & 0 & 0 & 0 & 0 & 0 & 0 & 1 & 0 & 0 & 0 & 0 & 1 & 0 & 0 & 0 & 1 & 1 & 0 & 1 & 1 & 0 & 1 & 0 & 0 & 1 \\
    0 & 0 & 0 & 0 & 0 & 0 & 0 & 0 & 0 & 0 & 1 & 0 & 0 & 0 & 1 & 1 & 1 & 1 & 0 & 1 & 1 & 1 & 0 & 0 & 0 & 1 & 0 & 1 \\
    0 & 0 & 0 & 0 & 0 & 0 & 0 & 0 & 0 & 0 & 0 & 1 & 0 & 0 & 1 & 1 & 0 & 0 & 1 & 0 & 0 & 1 & 0 & 1 & 0 & 0 & 1 & 1 \\
    0 & 0 & 0 & 0 & 0 & 0 & 0 & 0 & 0 & 0 & 0 & 0 & 1 & 0 & 1 & 0 & 1 & 0 & 1 & 0 & 0 & 0 & 1 & 0 & 0 & 1 & 1 & 1 \\
    0 & 0 & 0 & 0 & 0 & 0 & 0 & 0 & 0 & 0 & 0 & 0 & 0 & 1 & 1 & 0 & 0 & 1 & 1 & 0 & 0 & 0 & 0 & 1 & 1 & 1 & 0 & 1
    \end{bmatrix}\,.
\end{equation}

The complete weight enumerator for the $\widetilde{\mathrm{R}}$ sector is
\begin{equation}
    W_{\widetilde{\mathrm{R}}}(x_0,x_1) = x_0^{26} x_1^2 + 26 x_0^{22} x_1^6 + 143 x_0^{18} x_1^{10} - 340 x_0^{14} x_1^{14} + 143 x_0^{10} x_1^{18} + 26 x_0^6 x_1^{22} + x_0^2 x_1^{26}\,.
\end{equation}

The partition function of the $\widetilde{\mathrm{R}}$ sector is
\begin{equation}
    Z_{\widetilde{\mathrm{R}}}(\tau) = \frac{1}{\eta(\tau)^{28}} W_{\widetilde{\mathrm{R}}}(\theta_3(2\tau),\theta_2(2\tau))
    = \frac{4}{\eta(\tau)^{16}} E_8(\tau)\,.
\end{equation}
Note that $E_8(\tau)$ can be expressed by the theta function associated with the 16-dimensional even self-dual lattices.
The RR partition functions of all other codes at $n=28,\, p=2$ are $0$ or constant multiples of this.

\subsection*{$n=36\,,\,$ $p=2$}
The second matrix from the top for $\BF_2, n=36, d=8$ in \cite{database}:
\begin{equation}
    G = \begin{bmatrix}
    1&0&0&0&0&0&0&0&0&0&0&0&0&0&0&1&0&0&0&1&0&0&0&0&1&0&1&0&0&0&1&0&1&0&0&1 \\
    0&1&0&0&0&0&0&0&0&0&0&0&0&0&0&1&0&0&0&1&0&0&0&0&1&0&1&0&0&0&0&1&0&1&1&0 \\
    0&0&1&0&0&0&0&0&0&0&0&0&0&0&0&1&0&0&0&1&0&0&1&1&0&1&0&1&0&1&1&1&0&0&1&1 \\
    0&0&0&1&0&0&0&0&0&0&0&0&0&0&0&1&0&0&0&1&0&0&1&1&0&1&0&1&1&0&0&0&0&0&0&0 \\
    0&0&0&0&1&0&0&0&0&0&0&0&0&0&0&1&0&0&0&1&0&1&0&1&1&1&1&0&0&0&1&1&0&0&1&1 \\
    0&0&0&0&0&1&0&0&0&0&0&0&0&0&0&1&0&0&0&1&0&1&0&1&1&1&0&1&1&1&0&0&1&1&0&0 \\
    0&0&0&0&0&0&1&0&0&0&0&0&0&0&0&1&0&0&0&1&0&1&1&0&0&1&0&0&0&0&1&0&0&0&0&1 \\
    0&0&0&0&0&0&0&1&0&0&0&0&0&0&0&1&0&0&0&1&0&1&1&0&1&0&0&0&1&1&1&0&1&1&0&1 \\
    0&0&0&0&0&0&0&0&1&0&0&0&0&0&0&0&0&1&0&0&0&1&1&0&1&1&1&1&1&1&1&0&1&1&1&0 \\
    0&0&0&0&0&0&0&0&0&1&0&0&0&0&0&0&0&1&0&0&0&1&0&1&0&0&0&0&1&1&0&1&1&1&0&1 \\
    0&0&0&0&0&0&0&0&0&0&1&0&0&0&0&1&0&1&0&0&0&1&1&1&1&1&0&1&1&0&0&1&0&1&0&0 \\
    0&0&0&0&0&0&0&0&0&0&0&1&0&0&0&1&0&1&0&0&0&1&0&0&1&1&1&0&0&1&1&0&0&1&0&0 \\
    0&0&0&0&0&0&0&0&0&0&0&0&1&0&0&0&0&1&0&1&0&1&1&1&0&0&0&1&0&1&0&0&0&1&0&1 \\
    0&0&0&0&0&0&0&0&0&0&0&0&0&1&0&0&0&1&0&1&0&1&0&0&0&0&1&0&1&0&0&0&0&1&1&0 \\
    0&0&0&0&0&0&0&0&0&0&0&0&0&0&1&1&0&0&0&0&0&0&1&1&1&1&0&0&0&0&0&0&1&1&1&1 \\
    0&0&0&0&0&0&0&0&0&0&0&0&0&0&0&0&1&1&0&0&0&0&1&1&1&1&1&1&0&0&0&0&0&0&0&0 \\
    0&0&0&0&0&0&0&0&0&0&0&0&0&0&0&0&0&0&1&1&0&0&1&1&1&1&0&0&0&0&1&1&1&1&0&0 \\
    0&0&0&0&0&0&0&0&0&0&0&0&0&0&0&0&0&0&0&0&1&1&1&1&0&0&0&0&0&0&1&1&0&0&1&1
    \end{bmatrix}\,.
\end{equation}

The complete weight enumerator for the $\widetilde{\mathrm{R}}$ sector is
\begin{equation}
\begin{aligned}
    W_{\widetilde{\mathrm{R}}}(x_0,x_1) =&\; x_0^{34}x_1^2+34x_0^{30}x_1^6+544x_0^{26}x_1^{10}+1598x_0^{22}x_1^{14}-4354x_0^{18}x_1^{18} \\
    &+ 1598x_0^{14}x_1^{22}+544x_0^{10}x_1^{26}+34x_0^6x_1^{30}+x_0^2x_1^{34}\,.
\end{aligned}
\end{equation}

The partition function of the $\widetilde{\mathrm{R}}$ sector is
\begin{align}
\begin{aligned}
    Z_{\widetilde{\mathrm{R}}}(\tau) &= \frac{1}{\eta(\tau)^{36}} W_{\widetilde{\mathrm{R}}}(\theta_3(2\tau),\theta_2(2\tau))\\
    &= 4 j(\tau) - 384
    = \frac{4}{\eta(\tau)^{24}}\left(\frac{17}{18}E_4(\tau)^3+\frac{1}{18}E_6(\tau)^2\right)\,,
\end{aligned}
\end{align}
where $j(\tau)$ is the $j$ function (sometimes called Klein's $j$ function).

\bibliographystyle{JHEP}
\bibliography{SUSY_ECC}
\end{document}